\documentclass[11pt]{article}

\usepackage[margin=1in]{geometry} 
\usepackage[utf8]{inputenc}
\usepackage[T1]{fontenc}
\usepackage{lmodern} 
\usepackage{wrapfig}
\usepackage{ulem}
\usepackage{amsmath, amssymb, amsthm}
\usepackage{bm} 
\usepackage{lineno}

\usepackage{caption}
\usepackage{makecell}
\usepackage{graphicx}
\usepackage{booktabs}   
\usepackage{longtable}  
\usepackage{caption}
\usepackage{subcaption}
\usepackage{enumitem}   
\usepackage{multirow}
\usepackage[authoryear]{natbib}
\usepackage{microtype}  
\usepackage{xcolor}
\usepackage[colorlinks=true, allcolors=blue]{hyperref}
\usepackage{url}

\usepackage{algorithm}
\usepackage[noend]{algpseudocode}

\newtheorem{theorem}{Theorem}[section]
\newtheorem{lemma}[theorem]{Lemma}
\newtheorem{cor}{Corollary}[section]
\newtheorem{definition}[theorem]{Definition}

\newtheorem{remark}{Remark}[section]

\numberwithin{equation}{section}
\allowdisplaybreaks

\newcommand{\cC}{{\cal C}}

\newcommand{\cG}{{\cal G}}

\newcommand{\cL}{{\cal L}}
\newcommand{\cM}{{\cal M}}
\newcommand{\cN}{{\cal N}}

\newcommand{\cP}{{\cal P}}

\newcommand{\cU}{{\cal U}}

\newcommand{\Np}{{\mathbb{N}}^{+}}
\newcommand{\Nn}{{\mathbb{N}}^{-}}

\def\R{{\mathbb R}}
\def\Z{{\mathbb Z}}
\def\z{\zeta}
\def\e{\mathrm e}

\def\al{\alpha}

\def\ep{\epsilon}
\def\g{\gamma}

\def\h{\hat}
\def\ka{\kappa}
\def\la{\lambda}

\def\om{\omega}

\def\si{\sigma}
\def\Si{\Sigma}

\def\vep{\varepsilon}

\def\benr{\begin{eqnarray}}
	\def\eenr{\end{eqnarray}}
\def\benrr{\begin{eqnarray*}}
	\def\eenrr{\end{eqnarray*}}
\def\iny{\infty}

\def\noi{\noindent}
\def\nn{\nonumber}

\DeclareMathOperator*{\argmin}{arg\,min}
\DeclareMathOperator*{\argmax}{arg\,max}

\title{Your Ordinary Sized Title}
\author{Author 1 \and Author 2}
\date{\today}

\begin{document}

\def\spacingset#1{\renewcommand{\baselinestretch}%
{#1}\small\normalsize} \spacingset{1}


  \title{\bf High Dimensional Change Point Models for Two-Directional Data}

  \author{
    Abhishek Kaul\textsuperscript{a}, 
    Dipesh Baral\textsuperscript{a}, 
    Stergios B. Fotopoulos\textsuperscript{b}, \\
    Venkata K. Jandhyala\textsuperscript{a} and Rebecca Killick \textsuperscript{c}, \\[8 pt]
    \textsuperscript{a}Department of Mathematics and Statistics,\\
    \textsuperscript{b}Department of Finance and Management Science,\\
    Washington State University, Pullman, WA 99164, USA\\[4 pt]
	\textsuperscript{c}School of Mathematical Sciences\\
   Lancaster University, Lancaster LA1 4YF, UK}

\date{}
  \maketitle

\begin{abstract}
 We develop methodology for recovery of change points for data observed on more than one temporal index where changes may occur simultaneous in both indices, where the spatial component may be high dimensional. The work is motivated by climate monitoring problems where long series of data are available, e.g., daily observations (index 1) over several years (index 2). Such data may be evolving over the annual time scale, along with dynamic seasonal changes in the shorter time scale. We model this as a high dimensional mean process observed on a two dimensional grid with change points. Asymptotic estimation and inference results are developed under a single change point setup, including rates of convergence of the proposed method as well the resulting limiting distributions. The method is extended to the case of multiple changes. Theoretical results are supported numerically with monte-carlo simulations. We implement our work on a large scale climate data for the Pacific Northwest region of the United States.          
\end{abstract}

\noindent%

\spacingset{1} 

\section{Introduction}\label{sec:intro}

Climate influences global agriculture by steering growing season length, crop phenology, and yield potential, e.g., \citep{Krishan}. Recent studies have attempted to delineate the risks and cascading effects of climate change on crop development times and productivity associated with shifts in seasonality such as earlier spring onset, e.g., \citep{chen2019long, hu2005earlier}, longer frost-free seasons/shorter winters \citep{asse2018warmer}, and altered precipitation levels \citep{solomon2007summary}. For instance, wheat anthesis period in some regions has shifted due to rising temperatures, \citep{liu2016similar}. Similarly, phenology records indicate flowering and leaf-out dates have advanced by 15–20 days over the past century in temperate regions, \citep{buntgen}. Such climate trends from 1980–2008 have reduced global wheat yields by 5.5\%, \citep{lobell2011climate}. In view of these trends it is vital that the scientific community is able to track and measure evolving climatic variables such as temperature and precipitation at a granular level, in context of both seasonal shifts as well as long term changes. Our objective here is to develop a statistical model and associated inferential methods that accommodate the main features of evolving climate variable data.

\vspace{-3mm}
\begin{enumerate}[noitemsep, leftmargin=*, wide,label={\textcolor{red}{[\thesection.\arabic{enumi}]}}]
\item Climate data is usually available on more than one temporal axis. For e.g.,  $2m$-temperature\footnote{Temperature of air at 2m above surface, \url{https://codes.ecmwf.int/grib/param-db/167}.} which is available for long time periods. Here two temporal axes may be considered, one in the {\it days} units and the other in {\it years} units. There may be long term trends in this variable (along {\it years} temporal axis) as well as seasonal trends (along {\it days} axis). While increasing trends in temperature have been identified in longer time scales (over one temporal axis - {\it annual axis}), see ,e.g., IPCC 2023 report, \cite{lee2023ipcc}, the distribution of this excess heat may not be uniform across the secondary temporal axis ({\it days axis}), i.e., it may be the case that summer months are longer or with more intense heat while the winter months may not exhibit much change in comparison to prior years. Such changes in seasonality patterns may have significant consequences in the context of agricultural practices and modelling these features shall be one of the main goal of this article. \label{item.1}
\item The changes in a univariate setting with respect to a single location hold with statistical evidence with respect to the variation of that particular location. On the other hand, a multivariate analysis yields aggregate trends with respect to aggregate variation across all locations. Since the latter is with respect to aggregate variation, they are representative of changes in broader climate patterns. In contrast, a univariate analysis may reflect changes in localized weather patterns fail to consider the influence of larger regional or global scale patterns. This has also been noted in \citep{lund2023good}. Accordingly our focus in this article shall be on the larger scale multivariate setting. \label{item.2}      
\end{enumerate}

\vspace{-3.25mm}
Change point models are canonical to the literature and are a reliable tool for climatologists to infer temporal changes of climate, see, e.g., \citep{Lund2007} and the review papers \cite{jandhyala2013inference, beaulieu2012change, reeves2007review}. We build upon them while retaining the two features \ref{item.1} and \ref{item.2} above. The latter is accommodated by a high dimensional framework where the number of parameters may exceed the number of observations, which allows us to consider a large number of spatial locations. High dimensional change point models have gained considerable recent interest, e.g.,   \cite{wang2018high, cho2016change, Yu2020, Kaul2021single} amongst many others. Our main contribution however shall be on the aspect \ref{item.1}, i.e., to accommodate a multi-temporal framework where trends may be evolving over two axes simultaneously. 

Existing change point literature uses one of the following to deal with more than one time stamp ({\it day $\times$ year}). (i) Eliminate variation due to the shorter term ({\it days}) axis by measuring average temperature over each unit of the longer term axis ({\it years}), e.g., \cite{Kaul2024}. (ii) Retain variation but eliminate mean effect due to short term axis by {\it de-trending or de-seasonalizing}, i.e., analyzing residuals of an ANOVA type model after removal of shorter term mean effect, e.g., \cite{lund2023good}. This is perhaps the most common operation under this type of data. (iii) Supply a more refined functional model for periodicity over the shorter term axis ({\it days}) instead of the mean utilized in (ii), e.g., \citep{tucker2024elastic}.  Method (i) removes the shorter term index and thus cannot obtain simultaneous changes along this axis. Method (ii) and Method (iii) are conceptually similar and have the same deficiency. In the context of the running example, de-trending assumes that an increasing trend of temperature along the longer term axis  translates to an evenly distributed increase of temperature across each unit of the shorter temporal axis, i.e., each day of the year is incrementally but evenly hotter over time. In other words, de-trending by definition does not allow for a change in seasonality over time (simultaneous trend in shorter temporal axis).

Our motivation arises from $2m$-temperature data collected for the Pacific Northwest (PNW) region of the United States. This area is of significant agricultural importance. The study region is located between latitudes $42^\circ$N and $50^\circ$N, and longitudes $115^\circ$W and $125^\circ$W, encompassing parts of Washington, Oregon, Idaho, and British Columbia with a spatial resolution of $0.5^\circ\times0.5^\circ,$ resulting in $357$ locations. Temporal resolution is one day. Further details are provided in Section \ref{sec:real.data}. The objectives are to identify long term temperature changes as well as shifts in patterns of seasonality. To this end, consider high dimensional realizations with dynamic mean vectors on two-dimensions:

\begin{minipage}{0.6\textwidth}
	\benr\label{model:rvmcp}
	x_{(w,h)}&=&\begin{cases}\theta_{(1)}^0+\vep_{(w,h)}   & w> \tau^0_w,\,\, \&\,\, h>\tau^0_h,\\
		\theta_{(2)}^0+\vep_{(w,h)}   & w\le \tau^0_w,\,\, \&\,\, h>\tau^0_h,\\
		\theta_{(3)}^0+\vep_{(w,h)} & w\le \tau^0_w,\,\, \&\,\, h\le\tau^0_h,\\
		\theta_{(4)}^0+\vep_{(w,h)} & w> \tau^0_w,\,\, \&\,\, h\le \tau^0_h.
	\end{cases}\nn\\
	&=&\sum_{j=1}^4\theta_{(j)}^0{\bf 1}\big[(w,h)\in Q_j(\tau^0)\big]+\vep_{(w,h)},\\
	&&\hspace{2.25cm} w=1,...,T_w,\,\,h=1,...,T_h,\nn
	\eenr
\end{minipage}
\hspace{8mm}
\begin{minipage}{0.4\textwidth}
\includegraphics[width=0.8\textwidth]{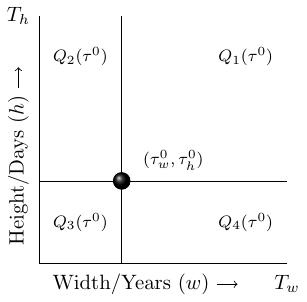}
\end{minipage}

Here $Q_j(\tau^0)$ are indices in the $j^{th}$ quadrant with origin at  $\tau^0=(\tau^0_w,\tau^0_h)\in\{1,...,T_w\}\times\{1,...,T_h\}$\footnote{Quadrants are ordered per cartesian convention,  $Q_1(\tau^0)=\{(w,h);\, w> \tau^0_w,\,\, \,\, h>\tau^0_h\},$ $Q_2(\tau^0)=\{(w,h);\, w\le \tau^0_w,\,\, \,\, h>\tau^0_h\},$ $Q_3(\tau^0)=\{(w,h);\,  w\le \tau^0_w,\,\, \,\, h\le\tau^0_h\},$ and $Q_4(\tau^0)=\{(w,h);\,  w> \tau^0_w,\,\, \,\, h\le \tau^0_h\}.$}. We observe $x_{(w,h)}\in\R^{p},$ $1\le w\le T_w,$ $1\le h\le T_h.$ The errors $\vep_{(w,h)}\in\R^p$ are zero mean random variables. Unknown parameters are the change points $\tau^0=(\tau^0_w,\tau_h^0)^T,$ and the mean vectors $\theta_{(j)}^0\in\R^p,$ $j=1,...,4,$ with $p$ being high dimensional with respect to the sampling period $T_wT_h.$

For our application, the horizontal (over $w$) is the {\it years} axis, $w=1,...,25.$ The vertical (over $h$) is the {\it days} axis, $h=1,...,365.$ We observe a $p=357$ dimensional temperature vector at each time $(w,h),$ whose $j^{th}$ component $x_{w,h,j}$ represents the $2m$-temperature for the $j^{th}$ grid location.

\begin{wrapfigure}{r}{0.3\textwidth}
  \begin{center}
\includegraphics[width=0.3\textwidth]{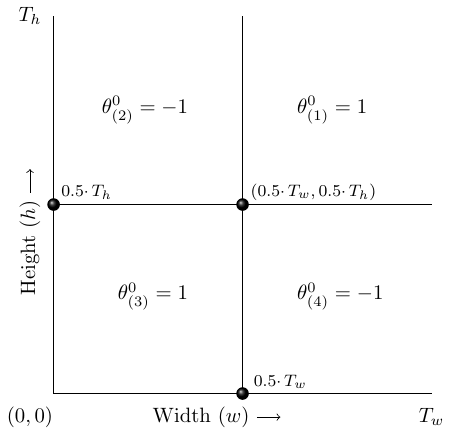}
  \end{center}
\captionsetup{width=0.3\textwidth}
 \caption{\footnotesize{Example to provide intuition on drawback of utilizing existing 1d-methods to obtain a 2d-change}}
\label{fig:one.d.problem}
\end{wrapfigure}
As indicated by a reviewer, one can recover change points of Model (\ref{model:rvmcp}) by a sequential 1d procedure, by looking one direction at a time and utilizing available one-dimensional methods. This can be done by collapsing the $h$ index and using existing methods to recover $\tau^0_w,$ i.e.,  where the $\big\{x_{(w,h)},\,\, w\le \tau^0_w,\, h\in\{1,...,T_h\}\big\}$ is the pre-change and $\big\{x_{(w,h)},\,\,w> \tau^0_w,\,h\in\{1,...,T_h\}\big\}$ is post-change data. This is followed by a symmetrical operation to  recover $\tau^0_h.$ The following discusses why this approach may at worst fail and at best be statistically inefficient.  A more detailed discussion is provided later in Remark \ref{rem:1d.limitation} and a numerical study in Section \ref{sec:numerical}. 

Consider a visualizable $p=1$ case. Suppose $\theta_{(1)}=\theta_{(3)}=1$ and $\theta_{(2)}=\theta_{(4)}=-1$ (see, Figure \ref{fig:one.d.problem}). Let change be at the mid-point $(\tau^0_w,\tau^0_h)=\big(\lfloor 0.5\cdotp T_w\rfloor,\lfloor 0.5\cdotp T_h\rfloor\big).$ Objective is to use a sequential 1d process and recover $\tau^0_w$ by collapsing $h$-axis. Doing so will lead to the same {\it effective mean}\footnote{weighted average of means in Quadrants 2 and 3, $\big(1\big/\tau^0_wT_h\big)\sum_{w=1}^{\tau^0_w}\sum_{h=1}^{T_h}E(x_{(w,h)})$} in the pre-data (prior to $\tau^0_w$- Quadrants 2 and 3) and the {\it effective mean}\footnote{weighted average of means in Quadrants 1 and 4, $\big(1\big/(T_w-\tau^0_w)T_h\big)\sum_{w=\tau^0+1}^{T_w}\sum_{h=1}^{T_h}E(x_{(w,h)})$} post-data (post $\tau^0_w$- Quadrants 1 and 4). Thus existing 1d-change point methods shall conclude "no change" along the $w$-axis. The same observation holds symmetrically for the $h$-axis. Consequently this sequential 1d procedure shall fail to recover any change in any direction. While this is a pathological example that evens out mean parameters, however the reason for failure under this case can also severely handicap statistical efficiency in a general case as explained next. 

Recall (e.g., \cite{bai1994}) that efficiency/precision of a change point method relies on the {\it jump size}, which is the discrimination between pre and post parameters that a method is designed to measure. E.g., in a 1d mean shift setting, most methods yield the jump size as the $\ell_2$ gap between mean parameters. This 1d approach fails when applied under the above example (Figure \ref{fig:one.d.problem}) since the effective jump size in the $w$ direction becomes the {\it \uline{$\ell_2$ gap} between \uline{weighted average} of the means in the pre and post segments} which turn out to be zero in Figure \ref{fig:one.d.problem}. In contrast, our approach shall yield a jump size in the $w$ direction as  the {\it \uline{weighted average} of the \uline{$\ell_2$ gap} in the means in pre and post segments} (see $\xi_w$ in (\ref{def:weight.jump.horiz.vert})), which in our example is non-zero. Key observation here is whether averaging happens before or after measuring $\ell_2$ gap between parameters. This shall effectively create in comparison a uniformly superior method. This intuition is made precise In Remark \ref{rem:1d.limitation} where we also show that our approach always yield a larger jump size in comparison to this sequential 1d- approach, in turn yielding a much sharper estimates.  Section \ref{sec:numerical} shall confirm this empirically. 

The problem of partitioning a two-dimensional sampling period under ($p=1$) can also be viewed from other perspectives.  Spatial anomaly detection (SAD) aims to achieve a similar objective, e.g., via scan statistics developed in \cite{zhang2010spatial, li2011spatial} amongst others. Another perspective on the objective and contributions of this article can be gained with a comparison to {\it clustering},  which is heavily utilized in modern machine learning tasks. Our approach is instead via change point parameters which can be  viewed as a bi-clustering task with an additional and critical requirement of contiguity of partitions along both axes (see, for e.g. partitioning in Figure \ref{fig:mult.parameters}). In contrast a anomaly detection/clustering algorithms may yield a anomalies/clusters with disjointed observations anywhere on the 2d- observation space. This is particularly relevant in the context of temporally observed data where discontinuity of partitions can lead to uninterpretable models. 

Further reasons for our choice to pursue a change point framework are as follows. To construct a concrete parametric framework that charaterizes paritions via change points and thus allows for traditional statistical inference. Recall this is not the case for clustering methods. Relatedly, to develop methods with statistical regularities which can be mathematically established, for e.g., that of consistency, rate of convergence and limiting distributions/inference. There is very limited related work utilizing change point models for multi-dimensional partitioning, this includes \cite{ushakova2023micro} and \cite{wang2025optimal}, both under univariate frameworks. The former develops a Metropolis Hastings type sampler to obtain bi-directional estimates in a Bayesian framework, which by construction relies on a likelihood based estimator which utilizes a sequential approach looking at one direction at a time. The latter works under a framework driven by regions of interest instead of change point parameters and focuses on results that obtain detection limit lower bounds. 

 The following provides additional control parameters that we shall show are directly related to the statistical behavior of the proposed method. Define the jump vectors for model (\ref{model:rvmcp}), 
\benr\label{def:jump.vec.quadrants}
\eta_{(1)}^0=\theta_{(2)}^0-\theta_{(1)}^0,\,\,\eta_{(2)}^0=\theta_{(3)}^0-\theta_{(2)}^0,\,\,\eta_{(3)}^0=\theta_{(3)}^0-\theta_{(4)}^0,\,\,{\rm and}\,\,\eta_{(4)}^0=\theta_{(1)}^0-\theta_{(4)}^0.
\eenr
Then we call their $\ell_2$ magnitudes as jump sizes across quadrants, i.e., 
\benr\label{def:jump.size.quad}
\xi_j=\|\eta^0_{(j)}\|_2,\quad j=1,...,4,\quad \overline\xi=\max_{j}\{\xi_j\},\quad{\rm}\quad \underline\xi=\min_{j}\{\xi_j\} .
\eenr
Additionally, define the proportion of observations in quadrants and along individual axes,
\benr\label{def:weights.quad}
&\om_j=|Q_j(\tau^0)|\big/{T_wT_h},\quad j=1,2,3,4,\quad{\rm and}\quad \underline\om=\min_{j}\{\om_j\}\nn\\
&\om_w=(T_w-\tau^0_w)\big/{T_w},\quad\om_h=(T_h-\tau^0_h)\big/{T_h},\,\,{\rm and},\,\, \om_{\min}=\om_w\wedge\om_h.
\eenr
Next define width and height-wise weighted jump sizes, which is critical to our analysis,
\benr\label{def:weight.jump.horiz.vert}
\,\,\,\xi_w^2= \om_h\xi_1^2+(1-\om_h)\xi_3^2,\quad \xi_h^2=\om_w\xi_4^2+(1-\om_w)\xi_2^2.\quad{\rm and }\quad\xi_{\min}=\xi_{w}\wedge\xi_h.
\eenr
A visual description of these control parameters is also provided in Figure \ref{fig:parameters}.\footnote{The means, weights and jump parameters depend on $T=T_w T_h$ directly since the change points $(\tau^0_w,\tau^0_h)$ may change with $T.$ These parameters may also depend on $T$ via $p,$ since under our high dimensional setting $p$ may be increasing with $T$ ($p$ is a sequence in $T$).}

\begin{remark}\label{rem:partitions} {\rm Model (\ref{model:rvmcp}) reduces to a single mean  when $\tau^0=(T_w,T_h).$ Two partitions result when $\tau^0$ lies on one of the axes. A change point $\tau^0=(\tau^0_w,\tau^0_h)$ may also result in three partitions if the means of two adjacent quadrants are identical, as visualized in Figure \ref{fig:parameters} (Left Panel).}

\begin{figure}[H]
	\centering
	\begin{minipage}{0.3\textwidth}
\includegraphics[width=0.95\textwidth]{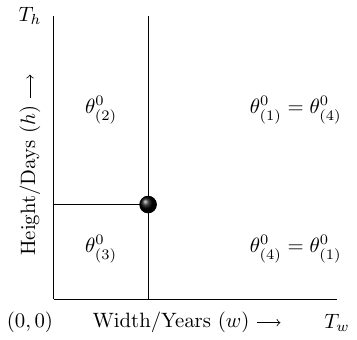}
	\end{minipage}
	\hspace{3mm}
	\begin{minipage}{0.3\textwidth}
	\includegraphics[width=0.95\textwidth]{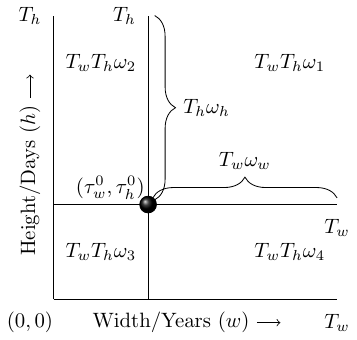}
	\end{minipage}
	\hspace{1mm}
	\begin{minipage}{0.3\textwidth}
	\includegraphics[width=0.95\textwidth]{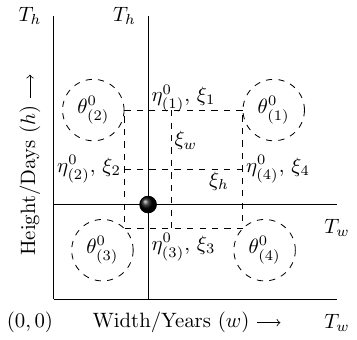}
	\end{minipage}
	\caption{\footnotesize{{\it Left: example of a three partition case described by Model (\ref{model:rvmcp})}. {\it Center: parameters measuring proportion of observations per region} {\it Right: mean, jump vector and jump size parameters.}}}
	\label{fig:parameters}
\end{figure}
\end{remark}

\begin{remark} {\rm We develop a method for the two-dimensional framework of Model (\ref{model:rvmcp}). Then similar to binary segmentation, we shall hierarchically search for further changes leading to a partitioning model as in Figure \ref{fig:mult.parameters}. The distinction being that we recursively partition each region into two dimensional subregions, as opposed to binary splits induced in one dimensional change points.}	

\begin{figure}[H]
\centering
	\begin{minipage}{0.48\textwidth}
		\includegraphics[width=0.58\textwidth]{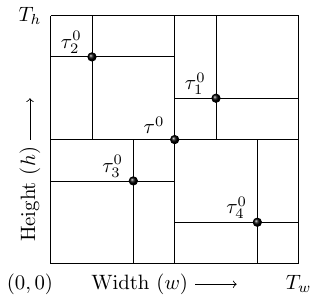}
	\end{minipage}
	\begin{minipage}{0.48\textwidth}
		\includegraphics[width=0.58\textwidth]{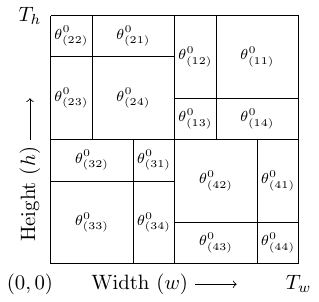}
	\end{minipage}
	\caption{\footnotesize{Mean and change point parameters of a multi-partition model obtained by a recursive implementation of proposed methodology (two-dimensional analog of binary segmentation).}}
	\label{fig:mult.parameters}
\end{figure}
\end{remark}

\begin{remark}(Limitation of existing 1d-methods applied sequentially under considered 2d-setting) \label{rem:1d.limitation} {\rm Suppose one were to estimate $\tau^0_w$ and $\tau^0_h$ by collapsing $h$ and $w$ axis, one at a time, utilizing one of many available methods, e.g., \cite{wang2020univariate, wang2018high, Kaul2021single}. Statistical quality of the resulting estimate $\h\tau_w$ shall depend upon a Jump size $\psi_w,$ i.e., $(\h\tau_w-\tau^0_w)=O_p(T_h^{-1}\psi^{-2}_w).$\footnote{Existing literature has results of the form $(\h\tau-\tau)=O_p(\psi^{-2}),$ e.g., \cite{wang2018high, Kaul2021single}. In this case each index on the $w$ axis has $T_h$ replicates and consequently this rate shall also be scaled by $T_h.$ This can also be obtained by rigorously following arguments similar to \cite{Kaul2021} under a squared loss.} Here $\psi_w$ is the jump size characterized by effective pre ($\mu_w^{pre}$) and post means $(\mu_w^{post})$that shall be the average expected values of the pre and post data as: 
		\benr
		&&\mu_w^{pre}=(1/\tau^0_wT_h)\sum_{w=1}^{\tau^0_w}\sum_{h=1}^{T_h}E(x_{(w,h)})=\om_h\theta_{(2)}+(1-\om_{h})\theta_{(3)}\nn\\
		&&\mu_w^{post}=(1/(T_w-\tau^0_w)T_h)\sum_{w=\tau^0_w+1}^{T_w}\sum_{h=1}^{T_h}E(x_{(w,h)})=\om_h\theta_{(1)}+(1-\om_{h})\theta_{(4)}\nn
		\eenr		
	Finally, $\psi_w^2=\|\mu_w^{pre}-\mu_w^{post}\|^2_2.$ Symmetrical for $\tau^0_h.$ Now for simplicity consider $p=1$ and $\tau^0_w$ and $\tau^0_w$ to be at the mid-point. Then $\psi^2_w=\big(0.5(\theta_{(2)}-\theta_{(1)})+0.5(\theta_{(3)}-\theta_{(4)})\big)^2.$ Then compare this quantity to the jump size $\xi_w^2=0.5(\theta_{(2)}-\theta_{(1)})^2+0.5(\theta_{(3)}-\theta_{(4)})^2$ described in (\ref{def:weight.jump.horiz.vert}) that shall result from our method and yield estimates satisfying $(\h\tau_w-\tau^0_w)=O_p(T_h^{-1}\xi^{-2}_w).$ It is straightforward to show that $\xi_w^2\ge\psi_w^2$ and consequently our method shall yield uniformly higher statistical efficiency in comparison to this sequential 1d approach. Moreover it shall not be suseptible to unexpected faliures such as that described in the example in Figure \ref{fig:one.d.problem}, where $\psi_w^2=0$ and $\xi_w^2=4.$ Finally, when there is no change along one axis (say, $h$-axis), i.e, $\theta_{(3)}=\theta_{(1)}$ and $\theta_{(4)}=\theta_{(1)}.$ Then $\xi^2_w=\psi^2_w,$ which is the case of 1d change point methods and here our method shall yield the same rate. Same holds symmetrically for $\tau^0_h.$ One can also utilize basic algebraic inequalities to observe that all of these arguments hold identically under a general multivariate $p>1$ and any change point $(\tau^0_w,\tau^0_h).$ Numerical confirmations of these assertions are provided in Section \ref{sec:numerical}.}
\end{remark}

A brief review of large scale change point models under a 1d-segmenting axis follows. Under a  fixed $p$ setting, the results of \cite{harchaoui2010multiple} consider a least squares estimator together with a total variation regularization. In a high dimensional multiple change framework \cite{wang2018high} provides a projected cusum estimator. In similar large scale settings the works of \cite{cho2016change, wang2019statistically} develop other CUSUM based estimators. The articles \cite{bhattacharjee2019change, Kaul2021single, Kaul2023+} consider squared loss type estimators in the same setting. The problem of post-estimation inference has been considered in \cite{bai2010common, Eichinger2018, fotopoulos2010exact} and more recently in \cite{Kaul2021single} and \cite{wang2020dating} who develop necessary limiting distributions for obtaining confidence intervals for these models.   High dimensional change point models have also been studied with several other mdoels, e.g., graphical models in \cite{kaul2023inference}, \cite{wang2017optimal}, \cite{atchade2017scalable}, stochastic block models in \cite{wang2018optimal}, \cite{bhattacharjee2018change}, markov random fields by \cite{roy2017change} are among other settings, wherein all these articles by construction assume a one-dimensional change axis. The remainder is organized as follows. Section \ref{sec:methods} describes proposed methodology. Section \ref{sec:main} develops results on statistical behavior of the proposed estimator pertaining to both estimation and inference. Section \ref{sec:numerical} provides Monte-Carlo simulation results that numerically support theoretical results. Finally in Section \ref{sec:real.data} we implement our methods on the $2m$-temperature data. All proofs are provided in the Supplement.

\noi{\it Notation}: $\R$ is the real line. For $\delta\in\R^p,$ $\|\delta\|_1,$ $\|\delta\|_2,$ $\|\delta\|_{\iny}$ represent the  1-norm, Euclidean norm, and sup-norm. For any $U\subseteq\{1,2,...,p\},$ let $\delta_U=(\delta_j)_{j\in U}$ represent the subvector of $\delta$ containing components corresponding to indices in $U.$ Let $|U|$ and $U^c$ represent cardinality and complement of $U.$  Notation $x_j$ represents the $j^{th}$ component of vector $x.$ We use $x_{(j)}$ to represent a vector. Denote by $a\wedge b=\min\{a,b\},$ and $a\vee b=\max\{a,b\}.$ We use a generic $c_u>0$ to represent universal constants. All limits are with respect to $T_w,$ and $T_h.$ Notation $\Rightarrow$ represents convergence in distribution.

\section{Methodology}\label{sec:methods}

Begin with a squared loss evaluated at any grid point $\tau=(\tau_y,\tau_d)^T$, and any $\theta_{(j)}\in\R^p$, $j=1,...,4,$ let $\theta$ represent the concatenation of $\theta_{(j)}'s,$ then define,
\benr\label{def:sq.loss}
\cL(\tau_w,\tau_h,\theta)=\frac{1}{T_wT_h}\sum_{j=1}^4\,\,\sum_{(w,h)\in Q_j(\tau)}\|x_{(w,h)}-\theta_{(j)}\|_2^2.
\eenr
Now define a component-wise plug-in estimator for $\tau^0=(\tau^0_w,\tau^0_h)^T$ as follows,
\benr\label{est:optimal}
\tilde\tau_w(\h\tau_h,\h\theta)=\argmin_{1\le\tau_w< T_w} \cL(\tau_w,\h\tau_h,\h\theta),\quad{\rm and}\quad \tilde\tau_h(\h\tau_w,\h\theta)=\argmin_{1\le\tau_h< T_h} \cL(\h\tau_w,\tau_h,\h\theta),
\eenr
where $\h\tau=(\h\tau_w,\h\tau_h)$ and $\h\theta$ represent some preliminary estimates that shall be specified in the sequel. The estimator (\ref{est:optimal}) separates target change point and all remaining nuisance parameters, which is somewhat reminiscent of the EM-Algorithm. In order to allow high dimensional means we utilize $\ell_1$ regularized quadrant-wise sample means. More precisely, for any $\tau=(\tau_y,\tau_d)^T,$ let,
\benr\label{def:empmeans}
\bar x_{(j)}(\tau)=\frac{1}{|Q_j(\tau)|}\sum_{(w,h)\in Q_j(\tau)}x_{(w,h)},\quad j=1,2,3,4.
\eenr
be the quadrant-wise sample means. Now consider soft-thresholding, $k_{\la}(x)={\rm sign}(x)(|x|-\la)_{+},$ $\la>0,$ $x\in\R^p,$ where ${\rm sign}(\cdotp),$ $|\cdotp|,$ and $(\cdotp)_{+},$ are applied component-wise. Here $(x)_{+}=x,$ if $x\ge 0,$ and $x=0$ if $x<0.$ Then for any $\la_j,$ define $\ell_1$ regularized quadrant-wise mean estimates,
\benr\label{est:softthresh}
\h\theta_{(j)}(\tau)=k_{\la_j}\big(x_{(j)}(\tau)\big),\quad j=1,2,3,4.
\eenr
It is well known in the literature, e.g. (\cite{donoho1995noising}, that the soft-thresholding operation in (\ref{est:softthresh}) is equivalent to the following $\ell_1$ regularization.
\benr\label{est:softL1construction}
\h\theta_{(j)}(\tau)&=&\argmin_{\theta\in\R^p}\big\|\bar x_{(j)}(\tau)-\theta\big\|^2_2+\la_j\|\theta\|_1,\quad \la_j>0,\quad j=1,2,3,4.
\eenr


\begin{algorithm}
	\caption{Estimation of $\tau^0=(\tau^0_w,\tau^0_h)^T$ away from boundaries.}
	\label{alg:single}
	\begin{algorithmic}[1]
		\Statex Initialize a user chosen change point $\check\tau=(\check\tau_w,\check\tau_h),$
		\State Compute estimates $\check\theta_{(j)}=\h\theta_{(j)}(\check\tau),$ $j=1,2,3,4.$ and update componentwise,
		\benr
		\h\tau_w=\argmin_{1\le \tau_w<T_h} \cL(\tau_w,\check\tau_h,\check\theta)\quad{\rm and}\quad
		\h\tau_h=\argmin_{1\le \tau_h<T_h} \cL(\check\tau_w,\tau_h,\check\theta)\nn
		\eenr
		\State Update mean estimates to $\h\theta_{(j)}=\h\theta_{(j)}(\h\tau),$ $j=1,2,3,4,$ and update,
		\benr
		\tilde\tau_w=\argmin_{1\le \tau_w<T_w} \cL(\tau_w,\h\tau_h,\h\theta)\quad{\rm and}\quad
		\tilde\tau_h=\argmin_{1\le \tau_h<T_h} \cL(\h\tau_w,\tau_h,\h\theta)\nn
		\eenr
		\Statex Output: $\tilde\tau=(\tilde\tau_w,\tilde\tau_h).$
	\end{algorithmic}
\end{algorithm}

Our methodology to recover $\tau$ is presented as Algorithm \ref{alg:single} and visualized in Figure \ref{fig:schematic}. It improves a nearly arbitrary $\check\tau,$ first to a near optimal $\h\tau$, i.e., a localization rate of $O_p\big(T_h^{-1}\xi_w^{-2}s\log^2(p\vee T_wT_h)\big)$ for $\h\tau_w$ and symmetrically for $\h\tau_h.$ A second iteration then improves this rate to  $O_p\big(T_h^{-1}\xi_w^{-2}\big)$ or $O_p\big(T_w^{-1}\xi_h^{-2}\big)$ for $\tilde\tau_w$ and $\tilde\tau_h,$ respectively.  The idea of alternatively updating change and nuisance parameters considederably reduces computational burden in comparison to a full grid search which involves simultaneous optimization of the change point and mean parameters. This approach was first developed in a linear regression framework in \cite{Kaul2019}.  This approach has has also been considered in a mean shift setting \cite{Kaul2021single, McGonigle2021} as well as a covariance shift setting \cite{kaul2023inference}. More broadly, the EM algorithm and the Gibbs sampler in principle also work in a similar manner iterating between target and nuisance parameters.

\begin{figure}
\centering
\includegraphics[width=0.8\textwidth]{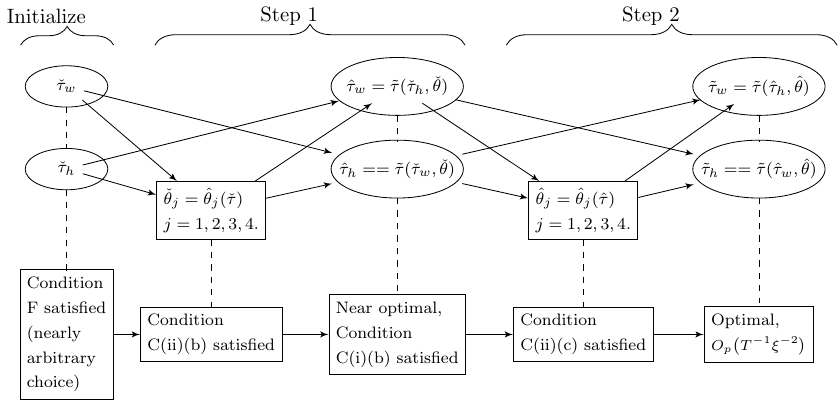}
\caption{\footnotesize{A schematic of the underlying working mechanism of Algorithm~\ref{alg:single}.}}
\label{fig:schematic}
\end{figure}

\section{Theoretical results}\label{sec:main}

The main purpose of this section is to obtain statistical properties of Algorithm \ref{alg:single}.  Towards this goal, Sub-section \ref{subsec:plugin} studies the behavior of the  plugin squared loss estimator (\ref{est:optimal}). This subsection shall be agnostic to the preliminary estimates and relies only upon reasonable statistical precision from these preliminary estimates. Sub-section \ref{subsec:feasible} shall then utilize these results along with further necessary details to establish the desired properties of Algorithm \ref{alg:single}. All results are under a high dimensional scaling (Condition E) below. 

{\it {{\noi{\bf Condition A (on underlying distributions):}} The vectors $\vep_{(w,h)}=(\vep_{(w,h,1)},...,\vep_{(w,h,p)})^T,$ $w=1,..,T_w,$ $h=1,...,T_{h}$ are independent and identically distributed (i.i.d.) subexponential random vectors with variance proxy $\si^2<\iny$ (see, Definition \ref{def:sube} and \ref{def:submult}) of the supplement}}

Subexponential distributions form a large class that subsumes the subgaussian class. It includes the Gaussian, Laplace, mean centered Chi-square, amongst several other well known distributions.

{\it {{\noi{\bf Condition B (on model parameters):}} (i) Covariance $\Sigma:=E\vep_{(w,h)}\vep_{(w,h)}^T$ has bounded eigenvalues, i.e., $0<\ka^2\le\rm{min eigen}(\Si) <\rm{max eigen}(\Si)\le\phi^2<\iny,$ with constants $\ka^2,$ $\phi^2.$\\
		(ii) Assume a change point exists and is separated from the parametric boundary on both axes, i.e., for some positive sequence $\underline\om\to 0,$ we have $\min_{j}\{|Q_j(\tau^0)|\}\ge T_wT_h\underline\om\to\iny,$\\~
		(iii) Let $\overline\xi$  and  $\xi_j,$ $j=1,...,4$ be as defined in (\ref{def:jump.size.quad}) and let $\xi_w,\xi_h,\xi_{\min}$ be as defined in (\ref{def:weight.jump.horiz.vert}). Then we assume that
		$\overline\xi\le c_u\xi_{\min},$ for some constant $c_u>0.$}}

Condition B(i) assumes a positive definite covariance over $p.$ This condition is usually implicit in Condition A, we state it here in favor of explicit clarity on parametric assumptions. Conditions B(ii) and B(iii) are separation conditions ensuring the {\it jump signal} is not dominated by noise. Moreover, B(ii) assumes a diverging number of observations in each quadrant of the model (\ref{model:rvmcp}). Analogous conditions are typical in a partitioning framework, see, e.g., Definition 3.1 of \cite{rockova2020posterior}. Next define sets of non-zero indices corresponding to mean vectors $\theta_{(j)}^0,$ $j=1,2,3,4,$
\benr\label{def:setS}
S_{j}=\big\{k\in\{1,2,...,p\};\,\,\theta^0_{(j)k}\ne 0\big\},\quad j=1,2,3,4,
\eenr
Define maximum cardinality $\max_{1\le j\le  4}|S_{j}|=s\ge 1.$ The parameter $s$ measures sparsity in  model (\ref{model:rvmcp}). We make this assumption directly on the mean vectors $\theta_{(j)}^0$s. We refer to \cite{Kaul2021single,   Kaul2023+} where it has also been illustrated that assuming this sparsity is equivalent to assuming sparsity of jump vectors, e.g., as done in \cite{wang2018high} and \cite{enikeeva2013high} under a 1d change axis framework. 

%

\subsection{Statistical properties of plugin 2d-squared loss estimator}\label{subsec:plugin}

This subsection details the statistical behavior of the plugin squared loss estimator (\ref{est:optimal}). These results are {\it agnostic} about the choice of the estimators used to obtain the preliminary estimates and instead rely on the following condition. This condition is a temporary placeholder, it shall be illustrated in the next subsection that the iterative construction of Algorithm \ref{alg:single} leads to preliminary estimates that satisfy all requirements. 

{\it {{\noi{\bf Condition C (on preliminary estimates):}} Let $c_{u1}>0$ be a suitably chosen small enough constant and let $\pi_T\to 0$ be a positive sequence. Then we assume either of the combinations of  \big[(i)(a), (ii)(a,b)\big] or \big[(i)(b), (ii)(a,c)\big] below hold with probability at least $1-\pi_T.$\\~
		(i)\,\, ({\bf on preliminary location estimates $\h \tau$}):\\~
		(a)\,\, Assume that $\h\tau_w,$ and $\h\tau_h$ satisfy the absolute error bound, 
		\benr
		|\h\tau_w-\tau^0_w|\le c_{u1}T_{w},\,\,{\rm and}\,\,\,\, |\h\tau_h-\tau^0_h|\le c_{u1}T_{h}\nn
		\eenr
		(b)\,\,Assume that $\h\tau_w,$ and $\h\tau_h$ satisfy the absolute error bound of 
		\benr
		|\h\tau_w-\tau^0_w|\le T_{w}r_{T}^2,\quad{\rm and}\quad |\h\tau_h-\tau^0_h|\le T_{h}r_{T}^2,\quad r_{T}=\frac{c_{u1}}{s^{1/2}\log(p\vee T_wT_h)}\footnotemark \nn
		\eenr
		(ii)\,\, ({\bf on preliminary mean estimates $\h\theta$}): Assume one of the pairs (a,b) or (a,c)  hold.\\~  	
		(a)\,\, The estimates $\h\theta_{(j)},$ $j=1,...,4$ satisfy $\|(\h\theta_{(j)})_{S_j^c}\|_1\le 3\|(\h\theta_{(j)}-\theta_{(j)}^0)_{S_j}\|_1,$ for each $j=1,...,4.$ Here $S_j,$ are sets of non-zero components as defined in (\ref{def:setS}). \\~
		(b)\,\, The estimates $\h\theta_{(j)},$ $j=1,...,4$  satisfy the $\ell_2$ bound,
		\benr
		\max_{1\le j\le 4}\|\h\theta_{(j)}-\theta_{(j)}^0\|_2\le c_{u1}\xi_{\min}.\nn 
		\eenr
		(c)\,\,The estimates $\h\theta_{(j)},$ $j=1,...,4$  satisfy the $\ell_2$ bound,
		\benr
		\max_{1\le j\le 4}\|\h\theta_{(j)}-\theta_{(j)}^0\|_2\le \xi_{\min}r_{T}\nn
		\eenr}}
\footnotetext{This is a sequence in both $T_w$ and $T_h,$ however to ease notation we present it in shorthand as $r_T$}\\~
The combination of $\big[(i)(b), (ii)(a,c)\big]$ is stronger version of the first $\big[(i)(a), (ii)(a,b)\big].$ In Subsection \ref{subsec:feasible} we exploit this distinction to show that preliminary estimates obtained recursively via Algorithm 1 satisfy the two considered combinations of this condition, at the two successive iterations, respectively. Conditions C(i) and C(ii)(b) are quite weak on the quality of these estimates, in particular C(i) is satisfied by any $\h\tau_w,\h\tau_h$ in $o(T)$-neighborhood's of $\tau^0_w,\tau^0_h,$ respectively. 

\begin{theorem}\label{thm:cpoptimal}  Suppose Conditions A, B, C(i)(a) and C(ii)(a,b) hold. Then, we have,
	\benr
	(1a)\,\, 	|\tilde\tau_w-\tau^0_w|=O\big(T_h^{-1}\xi^{-2}_{w}s\log^2(p\vee T_wT_h)\big),\quad
	(1b)\,\, |\tilde\tau_h-\tau^0_h|=O\big(T_w^{-1}\xi^{-2}_{h}s\log^2(p\vee T_wT_h)\big),\nn
	\eenr
	with probability at least $1-2\exp\{-c_{1}\log (p\vee T_wT_h)\}-\pi_T,$ for constant $c_1>0$ that does not depend on any model parameters,  where the orders are w.r.t. $T_w$ and $T_h,$ respectively.  Suppose instead Condition C(i)(b) and C(ii)(a,c) hold. Then, we have,
	\benr
	(2a)\,\, 	|\tilde\tau_w-\tau^0_w|= O_p\big(T_h^{-1}\xi^{-2}_{w}\big),\quad
	(2b)\,\, |\tilde\tau_h-\tau^0_h|= O_p\big(T_w^{-1}\xi^{-2}_{h}\big),\nn
	\eenr
as before, orders are w.r.t. $T_w$ and $T_h,$ respectively. 
\end{theorem}

Theorem \ref{thm:cpoptimal} provides rates of estimation yielded by (\ref{est:optimal}) under two varying conditions on the preliminary estimates. The second is a tighter result under the assumed tighter condition. Algorithm \ref{alg:single} exploits this property by its step wise construction enabling the first step to yield the rate in (1a) and (1b) and Step 2 to yield the sharper rate in (2a) and (2b), respectively. 

Next we characterize limiting distributions of $\tilde\tau,$ to construct asymptotically valid confidence intervals for $\tau^0=(\tau^0_w,\tau^0_h).$ We provide distributions under two cases of the jump size magnitude. First a vanishing jump $\surd{T_h}\xi_w\to 0,$ $\surd{T_w}\xi_h\to 0,$ and next a non-vanishing jump $\surd{T_h}\xi_w\to \xi_{(w,\iny)},$ $\surd{T_w}\xi_h\to \xi_{(h,\iny)}$ where $0<\xi_{(w,\iny)},\,\,\xi_{(h,\iny)}<\iny.$ The following condition shall ensure stability of variances of limiting processes to be characterized.

\vspace{1mm}
{\it {{\noi{\bf Condition D (stability of asymptotic variances):}} Let $\Si,$ $\eta^0_{(j)},$ $j=1,2,3,4,$ and $\xi_w,\xi_h$ be as in Condition B, (\ref{def:jump.vec.quadrants}) and (\ref{def:weight.jump.horiz.vert}) respectively. Then, assume the following limits exist,
		\benr
		&(i)&\,\,\, \frac{1}{\xi^2_w}\Big[\om_h\eta^{0T}_{(1)}\Si\eta^0_{(1)}+(1-\om_h)\eta^{0T}_{(3)}\Si\eta^0_{(3)}\Big]\to \si^2_{(w,\iny)},\quad{\rm and}\nn\\
		&(ii)&\,\,\, \frac{1}{\xi^2_h}\Big[\om_w\eta^{0T}_{(4)}\Si\eta^0_{(4)}+(1-\om_w)\eta^{0T}_{(2)}\Si\eta^0_{(2)}\Big]\to \si^2_{(h,\iny)},\nn
		\eenr
		with $0<\si_{(w,\iny)},\,\,\si_{(h,\iny)}<\iny.$ Here the limit of (i) is with respect to $T_w\to\iny$ (with $T_h<\iny$ or $T_h\to\iny$), and symmetrically (ii) is with respect to  $T_h\to\iny$ (with $T_w<\iny$ or $T_w\to\iny$).}}		

\vspace{1.5mm}
The quantities $\si^2_{(w,\iny)}$ and $\si^2_{(h,\iny)}$ are variances of limiting processes which in turn determine the asymptotic variance of the $\tilde\tau_w$ and $\tilde\tau_h,$ respectively. Roughly, these expression say that variance of $\tilde\tau$'s are determined by a weighted aggregation of the variances of components of  $x_{(w,h)}$ (diagonal elements of $\Si$), with more weight to those where there is more change in the mean (measured by $\eta$'s). There is a further weighting for the number of observations in each quadrant ($\omega$'s). Similar conditions are quite common in the literature, e.g., this condition serves the same purpose as that of assuming limiting stability of the Gram matrix in context of ordinary linear regression.

\begin{theorem}\label{thm:wc.vanishing} Suppose Conditions A, B and D hold. Assume the vanishing jump size regime, $\surd{(T_h)}\xi_w\to 0$ and $\surd{(T_w)}\xi_h\to 0.$  Further suppose Conditions C(i)(b) and C(ii)(a,c) are satisfied with $r_T=o(1)\big/\{s^{1/2}\log(p\vee T)\}.$ Then, we have,
	\benr\label{eq:wc.vanishing}
	&&T_h^{-1}\xi^{2}_w(\tilde\tau_w-\tau^0_w)\Rightarrow \argmax_{\z\in\R}\big\{2\si_{(w,\iny)}W_w(\z)-|\z|\},\qquad T_w\to\iny,\nn\\
	&&T_w^{-1}\xi^{2}_h(\tilde\tau_h-\tau^0_h)\Rightarrow \argmax_{\z\in\R}\big\{2\si_{(h,\iny)}W_h(\z)-|\z|\},\hspace{2mm}\qquad T_h\to \iny,
	\eenr
	where $W_w(\z),$ and $W_h(\z)$ are both two sided Brownian motions\footnote{A two-sided Brownian motion $W(\z)$ is defined as $W(0) = 0,$ $W(\z) = W_1(\z),$ $\z > 0$ and $W(\z) = W_2(-\z),$ $\z < 0,$ where $W_1(\z),$ $W_2(\z)$ are two independent Brownian motions defined on the non-negative half real line}. 
\end{theorem}

A change of variable $\z=\si_{\iny}^2\z',$ yields $\argmax_{\z\in\R}\big\{2\si_{\iny}W(\z)-|\z|\}=^d\si_{\iny}^2\argmax_{\z'\in\R}\big\{2W(\z')-|\z'|\},$ whose cdf is available in \cite{yao1987approximating}. Next consider the non-vanishing regime, $\surd{(T_h)}\xi_w\to \xi_{(w,\iny)},$ and $\surd{(T_w)}\xi_h\to \xi_{(h,\iny)}.$ For this we require the following additional distributional assumption.

\vspace{1.5mm}
{\it {{\noi{\bf Condition A$'$ (additional distributional assumptions):}} Suppose Conditions A, B and D. Assume the non-vanishing jump size regime, i.e.,  $\surd{(T_h)}\xi_w\to\xi_{(w,\iny)},$ and $\surd{(T_w)}\xi_h\to\xi_{(h,\iny)},$ with $0<\xi_{(w,\iny)},\xi_{(h,\iny)}<\iny.$ For each $w=1,...,T_w$ and $h=1,...,T_h$ define,
		\benr
		\psi_{w,T_w}=\Big[\sum_{h=\tau^0_h+1}^{T_h}\vep_{(w,h)}^T\eta^0_{(1)}+\sum_{h=1}^{\tau_h^0}\vep_{(w,h)}^T\eta^0_{(3)}\Big],\quad
		\psi_{h,T_h}=\Big[\sum_{w=\tau^0_w+1}^{T_w}\vep_{(w,h)}^T\eta^0_{(4)}+\sum_{w=1}^{\tau_w^0}\vep_{(w,h)}^T\eta^0_{(2)}\Big]\nn
		\eenr
		Then we assume that for any constants $c_1,c_2\in\R,$ and for some  distribution $\cP,$ which is continuous and supported in $\R,$ we have,
		\benr
		(i)\,\, c_1+c_2\psi_{w,T_w}\Rightarrow \cP\big(c_1,c_2^2\xi_{(w,\iny)}^2\si^2_{(w,\iny)}\big),\quad
		(ii)\,\, c_1+c_2\psi_{h,T_h}\Rightarrow \cP\big(c_1,c_2^2\xi_{(h,\iny)}^2\si^2_{(h,\iny)}\big).\nn
		\eenr
		Here $\si^2_{(w,\iny)}$ and $\si^2_{(h,\iny)}$ are as defined in Condition D.}}

\vspace{1.5mm}
Further, define a negative drift two sided random walk initializing at the origin.
\benr\label{eq:neg.drift.rw}
\cC_{\iny}(\z,\xi,\si^2)=
\begin{cases}\sum_{t=1}^{\z} z_t, & \z\in \Np=\{1,2,3,...\} \\ 	
	0,				  &	\z=0 \\
	\sum_{t=1}^{-\z}z_t^*,		  &	\z\in \Nn=\{-1,-2,-3,...\},
\end{cases}
\eenr
where $z_t,z_t^*$ are independent copies of a $\cP\big(-\xi^2,4\xi^2\si^2\big)$ distribution, which are also independent over all $t,$ for a distribution law $\cP$\footnote{If one assumes $\vep_{(w,h)}\sim^{i.i.d} \cN(0,\Si),$ then $\cP$ shall also be a normal distribution.} that shall be determined by the form of the underlying distribution in model (\ref{model:rvmcp}) (see, Condition A$'$). Finally, let,
\benr\label{eq:42}
\,\,\,\,\cC_{(w,\iny)}(\z)=\cC_{\iny}\big(\z,\xi_{(w,\iny)},\si_{(w,\iny)}^2\big)\,\, {\rm and}\,\, \cC_{(h,\iny)}(\z)=\cC_{\iny}\big(\z,\xi_{(h,\iny)},\si_{(h,\iny)}^2\big),
\eenr
where $\si_{(w,\iny)}^2$ and $\si_{(h,\iny)}^2$ are variance parameters as defined earlier for the vanishing regime. 

The only additional requirement of Condition A$'$, in comparison to Conditions B and D is of continuous distributions. Here $\mu,\si^2$ represent mean and variance, i.e,  $E\cP(\mu,\si^2)=\mu,$ and ${\rm var}\big(\cP(\mu,\si^2)\big)=\si^2.$  If one assumes $\vep_{(w,h)}\sim \cN(0,\Si),$ then we have $\cP=^d\cN.$ 

\begin{theorem}\label{thm:wc.non.vanishing} Suppose Conditions A$'$, B, D hold. Assume that the jump sizes are non-vanishing $\surd{(T_h)}\xi_w\to \xi_{(w,\iny)},$ and $\surd{(T_w)}\xi_h\to \xi_{(h,\iny)}.$ Further suppose Conditions C(i)(b) and C(ii)(a,c) are satisfied with $r_T=o(1)\big/\{s^{1/2}\log(p\vee T)\}.$ Then, we have,
	\benr\label{eq:wc.non.vanishing}
	&&(\tilde\tau_w-\tau^0_w)\Rightarrow \argmax_{\z\in \Z}\cC_{(w,\iny)}(\z),\qquad T_w\to\iny,\nn\\
	&&(\tilde\tau_h-\tau^0_h)\Rightarrow \argmax_{\z\in \Z}\cC_{(h,\iny)}(\z),\hspace{2mm}\qquad T_h\to \iny,
	\eenr
	where $\cC_{(w,\iny)}(\z)$ and $\cC_{(h,\iny)}(\z)$ are as defined in (\ref{eq:neg.drift.rw}) and (\ref{eq:42}).
\end{theorem}

\begin{remark}\label{rem:intervals}{\bf (Construction of adaptive confidence intervals)} {\textnormal{Theorem \ref{thm:wc.vanishing} and Theorem \ref{thm:wc.non.vanishing} enable construction of componentwise confidence intervals for $(\tau^0_w,\tau^0_h)^T$ as: }}

\vspace{2mm}
\hspace{1.5cm}
\begin{minipage}{0.45\textwidth}
{\bf Non-vanishing regime}\\~
$CI_w:\,\big[\tilde\tau_w\pm q_{(\al,w)}^{nv}\big],$ \\
$CI_h:\,\big[\tilde\tau_h\pm q_{(\al,h)}^{nv}\big]$		
\end{minipage}	
\hspace{-1.5cm}
\begin{minipage}{0.45\textwidth}
{\bf Vanishing regime}\\~
$CI_w:\,\big[\tilde\tau_w\pm q_{\al}^{v}\si^2_{(w,\iny)}\big/(T_h\xi_w^2)\big]$,\\
$CI_h:\,\big[\tilde\tau_h\pm q_{\al}^{v}\si^2_{(h,\iny)}\big/(T_w\xi_h^2)\big]$		
\end{minipage}

\vspace{2mm}
{\textnormal {Here $q_{(\al,w)}^{nv},$ $q_{(\al,h)}^{nv},$ are the $(1-\al/2)^{th}$ quantiles of the distributions  $\argmax_{\z\in \Z}\cC_{(w,\iny)}(\z)$ and $\argmax_{\z\in \Z}\cC_{(h,\iny)}(\z)$ of the non-vanishing regime, respectively, and $q_{\al}^{v}$ the corresponding quantile of the distribution $\argmax_{\z\in\R}\big\{2W(\z)-|\z|\}.$ The latter available from its cdf in \cite{yao1987approximating} and the former two obtained as monte-carlo approximations. }}

{\textnormal{The relationship between the two regimes is examined in Theorem 3.6 of \cite{Kaul2023+}, where it is shown these are discrete and continuous versions of the same stochastic process. Specifically, it yields that margin of errors for the two regimes are {\it asymptotically equivalent} when jump size is vanishing, $q_{(\al,w)}^{nv}\asymp q_{\al}^{v}\si^2_{(w,\iny)}\big/(T_h\xi_w^2),$ as $(\surd{T_h}\xi_w)\to 0$ and analogously for the second component. Consequently, one may always utilize the intervals in the non-vanishing regime, and doing so yields $(1-\al)$ asymptotic coverage, $pr\Big((\tilde\tau_w-q_{(\al,w)}^{nv})\le\tau^0_w\le (\tilde\tau_w+q_{(\al,w)}^{nv})\Big)\to (1-\al),$ and $pr\Big((\tilde\tau_h-q_{(\al,h)}^{nv})\le\tau^0_h\le (\tilde\tau_h+q_{(\al,w)}^{nv})\Big)\to (1-\al),$ {\it irrespective of whether the underlying regime is vanishing or non-vanishing}. In other words, the intervals $CI_w,\,CI_h$ are regime adaptive.  }}
\end{remark}

\subsection{Statistical properties of Algorithm \ref{alg:single}}\label{subsec:feasible}

This subsection demonstrates that Algorithm \ref{alg:single} retains all statistical results developed earlier while fully eliminating Condition C, and replacing it with simply a scaling assumption of form $s\log^2 p=o\big(\surd (T_wT_h)\big)$ (Condition E below) on the rate of model parameters.

{\it {{\noi{\bf Condition E (parametric rate restrictions):}}  Let  $\xi_{\min}$ be as in (\ref{def:weight.jump.horiz.vert}) and let $s,p$ be the sparsity parameter (see, (\ref{def:setS})) and dimension size, respectively. Then, assume the following,
		\benr
		\Big(\frac{c_u\si}{\xi_{\min}}\Big)\Big\{\frac{s\log^{2} (p\vee T_wT_h)}{\surd (T_wT_h\underline\om)}\Big\}= o(1).\nn
		\eenr
		Additionally, assume that $s\log(p\vee T_wT_h)\le c_u T_wT_h\underline\om,$ for some constant $c_u>0.$\footnote{We do not necessarily require the order $o(1)$ here to hold simultaneously w.r.t $T_w,T_h.$ If this order holds w.r.t. $T_w\to\iny$ ($T_h<\iny$ or $T_h\to\iny$) then it is sufficient in context of the width change parameter $\tau_w^0,$ and symmetrically for the height change parameter.}
}}

Next is an assumption on the initializer of Algorithm \ref{alg:single}, followed by a discussion illustrating its mildness and a simple strategy for its choice in practice.

\vspace{1.5mm}
{\it {{\noi{\bf Condition F (initializer of Algorithm \ref{alg:single}):}} Let $\psi=\max_{1\le j\le 4}\|\eta^0_{(j)}\|_{\iny},$ and assume that the initializer $\check\tau=(\check\tau_w,\check\tau_h)^T$ of Algorithm \ref{alg:single} satisfies the relations.
		\benr
		(i)\,\, |\check\tau_w-\tau^0_w|\le \frac{c_{u1}T_w\underline\om}{\big(\surd{s\psi\big/\xi_{\min}}\big)},\quad{\rm and}\quad  (ii)\,\,|\check\tau_h-\tau^0_h|\le \frac{c_{u1}T_h\underline\om}{\big(\surd{s\psi\big/\xi_{\min}}\big)}.\nn
		\eenr		
		Additionally assume  $(iii)\,\,\min_{1\le j\le }|Q_j(\check\tau)|\ge c_uT_wT_h\underline\om.$ Here $\underline\om$ is as defined in Condition B, $c_{u}>0$ is any constant and $c_{u1}>0$ is an appropriately chosen small enough constant.}}

Requirement (iii) of Condition F is innocuous, it is satisfied with $\check\tau=(\lfloor T_wk_w\rfloor,\,\lfloor T_hk_h\rfloor)^T=,$ with any $(k_w,k_h)^T\in [c_{u1},c_{u2}]\times[c_{u1},c_{u2}]\subset(0,1)\times(0,1).$ Requirements (i) and (ii) are symmetrical versions in the horizontal and vertical directions. Theoretically valid initializers $\check\tau_w$ and $\check\tau_h$ in very wide $o(T_w)$ and $o(T_h)$ neighborhoods of $\tau^0_w$ and $\tau^0_h$ can be obtained by means of a preliminary coarse grid search as: consider $\log T_w\log T_h$ equally separated values in $\cP\subset\{1,...,T_w\}\times\{1,...,T_h\}$ forming a coarse grid of possible initializers. Then, select the best fitting value $(\check\tau_w,\check\tau_h)$ for Algorithm \ref{alg:single}, i.e., 
$
\check\tau=\argmin_{\tau\in\cP}\cL\big(\tau_w,\tau_h,\theta(\tau)\big).
$
A similar preliminary coarse grid search has also been  utilized in \cite{roy2017change, Kaul2019, Kaul2021single, atchade2017scalable} and more recently in \cite{McGonigle2021}. In all implementations we consider a preliminary grid search of $(\check\tau_w,\check\tau_h)\in\{\lfloor 0.25\cdotp T_w\rfloor,\lfloor 0.5\cdotp T_w\rfloor,\lfloor 0.75\cdotp T_w\rfloor\}\times \{\lfloor 0.25\cdotp T_h\rfloor,\lfloor 0.5\cdotp T_h\rfloor,\lfloor 0.75\cdotp T_h\rfloor\}$ for the initializer.

The following results show that Algorithm 1 retains all desired statistical properties while fully eliminating the placeholder Condition C and replacing it with Condition E. 

\begin{cor}\label{cor:alg1.validity} Suppose Condition A, B, E and F hold and assume that the regularizers for mean estimates of Step 1 of Algorithm 1 are chosen as in (\ref{eq:la.step1.choice}). Then, \\~
	(a) Step 1 estimate $\h\tau=(\h\tau_w,\h\tau_h)^T$ satisfies the bounds (1a) and (1b) of Theorem \ref{thm:cpoptimal}.\\~
	Additionally, suppose the regularizers for mean estimates of Step 2 are chosen as in  (\ref{eq:la.step2.choice}) and assume  $(\psi/{\xi_{\min}}\big)\le c_u\surd\{\log(p\vee T)\}.$ Then, \\~
	(b) Step 2 estimate $\tilde\tau=(\tilde\tau_w,\tilde\tau_h)^T$ satisfies the sharp bounds (2a) and (2b) of Theorem \ref{thm:cpoptimal}.\\~
	Furthermore, assume Condition A$'$ and D holds. Then,\\~
	(c) Step 2 estimate $\tilde\tau=(\tilde\tau_w,\tilde\tau_h)^T$ satisfies the limiting distributions of Theorem \ref{thm:wc.vanishing} and Theorem \ref{thm:wc.non.vanishing}, in the vanishing and non-vanishing jump size regimes, respectively. 
\end{cor}

This result shows that any $\check\tau=(\check\tau_w,\check\tau_h)^T$ under Condition F yields Step 1 means $\check\theta_{(j)}=\h\theta_{(j)}(\check\tau),$ $j=1,2,3,4,$ of (\ref{est:softthresh}) satisfying the weaker Condition C(ii)(a,c). Parts (1a) and (1b) of Theorem \ref{thm:cpoptimal} now guarantees that $\h\tau=(\h\tau_w,\h\tau_h)^T,$ are near optimal estimates. This near optimal $\h\tau$ together with the updated mean estimates $\h\theta_{(j)}=\h\theta_{(j)}(\h\tau),$ $j=1,2,3,4,$ satisfy the stronger requirements of Condition C(i)(b) and C(ii)(a,c). This allows us to perform another update $\tilde\tau=(\tilde\tau_w,\tilde\tau_h).$ Parts (2a) and(2b) of Theorem \ref{thm:cpoptimal} now guarantee optimality of Step 2,i.e.,  $o(T_w)$-nbd.$\longrightarrow^{\rm Step 1}$ near optimal-nbd., $O_p(T_h^{-1}\xi^{-2}_ws\log^2 p)$ $\longrightarrow^{\rm Step 2}$ optimal-nbd., $O_p(T_h^{-1}\xi^{-2}_w),$ in context of the width change parameter $\tau_w$, and symmetrically for illustrated in Figure \ref{fig:schematic}.

\section{Numerical experiments}\label{sec:numerical}

This section provides Monte-Carlo simulations to support developments above. In all to follow noise variables $\vep_{(w,h)}\in\R^p$ are generated as i.i.d. Gaussian r.v.'s, $\vep_{(w,h)}\sim^{i.i.d} \cN(0,\Si),$ $w=1,...,T_w,\,\,h=1,...,T_h.$ Covariance matrix $\Si$ is chosen to be a toeplitz type matrix defined as $\Si_{ij}=\rho^{|i-j|},$ $i,j=1,...,p$ and $\rho=0.5.$ Tuning parameters are set as $\la=\la_j,$ $j=1,...,4$ and are chosen via a usual BIC criterion, e.g., \citep{Kaul2019, Kaul2021single} with the considered 2d-squared loss (\ref{def:sq.loss}).

\subsection{Simulation A: Limitations of a sequential application of existing 1d-methods}

\begin{wrapfigure}{r}{0.4\textwidth}
\resizebox{0.4\textwidth}{!}{
\begin{tabular}{lcccccr}
\toprule
$\theta_1$ & \makecell{JumpSize\\ $\xi_w$(2d)} &  \makecell{JumpSize\\ $\psi_w$(1d)} & \makecell{Bias\\Alg.1}  & \makecell{Bias\\ Seq.1d}&  \makecell{Rmse\\Alg.1} & \makecell{Rmse\\ Seq.1d} \\
\midrule
0.25 & 0.56 & 0.11 & 0.99 & 3.51 & 4.03 & 8.74\\
0.28 & 0.60 & 0.15 & 0.78 & 3.63 & 2.81 & 8.52\\
0.30 & 0.63 & 0.18 & 0.57 & 3.31 & 2.36 & 8.01\\
0.33 & 0.67 & 0.22 & 0.54 & 3.12 & 1.78 & 7.74\\
0.36 & 0.71 & 0.25 & 0.51 & 1.97 & 1.57 & 6.64\\
0.38 & 0.75 & 0.29 & 0.39 & 2.49 & 1.31 & 6.32\\
0.41 & 0.79 & 0.32 & 0.43 & 1.90 & 1.25 & 5.64\\
0.43 & 0.83 & 0.36 & 0.42 & 1.43 & 1.37 & 4.22\\
0.46 & 0.87 & 0.39 & 0.29 & 1.23 & 1.05 & 3.77\\
0.49 & 0.91 & 0.43 & 0.19 & 1.00 & 0.85 & 3.49\\
0.51 & 0.96 & 0.46 & 0.20 & 0.92 & 0.81 & 2.48\\
0.54 & 1.00 & 0.50 & 0.17 & 0.87 & 0.71 & 2.59\\
0.57 & 1.04 & 0.54 & 0.17 & 0.70 & 0.69 & 2.05\\
0.59 & 1.08 & 0.57 & 0.09 & 0.58 & 0.48 & 1.63\\
0.62 & 1.13 & 0.61 & 0.07 & 0.58 & 0.43 & 1.55\\
0.64 & 1.17 & 0.64 & 0.05 & 0.65 & 0.33 & 1.41\\
0.67 & 1.21 & 0.68 & 0.03 & 0.60 & 0.23 & 1.18\\
0.70 & 1.26 & 0.71 & 0.03 & 0.51 & 0.20 & 1.09\\
0.72 & 1.30 & 0.75 & 0.02 & 0.58 & 0.13 & 1.05\\
0.75 & 1.35 & 0.78 & 0.02 & 0.52 & 0.14 & 1.01\\
\bottomrule
\end{tabular}}
\captionof{table}{\footnotesize{{\it Results of Simulation A with respect to $\tau^0_w$ (horizontal axis). Observed jump size under proposed approach is uniformly larger and localization metrics of bias and rmse are uniformly lower.}}}  
\label{tab:sim1dcomp}
\end{wrapfigure}
This section numerically presents the deficiency of existing 1d-change point methods to recover change points when applied sequentially for $\tau^0_w$ and $\tau^0_h$ under Model \ref{model:rvmcp}. 

We generate $20$ combinations of $p$-dimensional mean vectors $\theta_{(1)},\theta_{(2)},\theta_{(3)},\theta_{(4)}.$ To obtain visualizable results we alter only one parameter which is the first component of the first mean vector. We set $\theta_{(2)}=\theta_{(4)}=0_{p\times 1},$ $\theta_{(3)}=(0.25_{s\times 1},0_{(p-s)\times 1}),$ where $0.25_{s\times 1}$ represents an s-vector with all components set to $0.25.$ We set $\theta_{(1)}=((\theta_1)_{s\times 1},0_{(p-s)\times 1})$ where $\theta_1$ is in a uniform grid between $0.25$ and $0.75.$  Dimensions are set to $p=100$ and $s=5,$ which are not altered as the reason for the deficiency is the discrepancy of jump sizes and same will remain true irrespective of dimensions. Sampling periods are set to $T_w=T_h=35$ and the change point is set to $\tau_w^0= \lfloor 0.4\cdotp T_w\rfloor$ and $\tau_h^0= \lfloor 0.4\cdotp T_\rfloor.$  For each parameter combination we compute corresponding jump sizes $\psi_w$ (see, Remark \ref{rem:1d.limitation}) for the sequential 1d procedure and  $\xi_w$ (see, (\ref{def:weight.jump.horiz.vert})) for the proposed method. We obtain monte-carlo approximations (over $500$ replications) of bias \big($|E(\h\tau-\tau)|$\big) and rmse  \big($E^{1/2}(\h\tau-\tau)^{2}$\big). The sequential-1d method forming the baseline is implemented via the method of \cite{wang2018high} by the author provided r-package under default settings. Results with respect to $\tau^0_w$ are reported in Table \ref{tab:sim1dcomp} and visualized in Figure \ref{fig:sim1dcomp}.

Results support the discussion in Section \ref{sec:intro}. Following are key observations stated w.r.t the horizontal direction $w.$ The jump size measure under the proposed methodology is uniformly higher than a sequential 1d approach for all considered mean parameter values, i.e., our method can `see' more discrimination between pre and post segments in comparison to a sequential 1d approach, see Left Panel of Figure \ref{fig:sim1dcomp}. This distinction manifests in the resulting precision in which the change point parameter is localized, see, Right Panel Figure \ref{fig:sim1dcomp}. Specifically, the proposed Algorithm 1 yields a bias that is uniformly and proportionately smaller. Same observation holds for rmse. Finally, we emphasize that that the deficiency of the sequential 1d approach observed here is due to the foundational issue of the method being implemented under a mis-specified Model \ref{model:rvmcp} and we expect the same with any other 1d-method besides that of the considered \cite{wang2018high}.

\begin{figure}[htbp]
\centering
	\begin{minipage}{0.42\textwidth}
	\includegraphics[width=\textwidth]{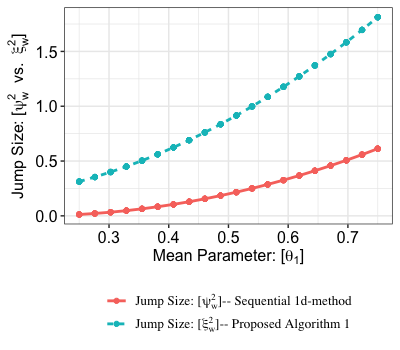}
	\end{minipage}
	\hspace{3mm}
	\begin{minipage}{0.42\textwidth}
	\includegraphics[width=\textwidth]{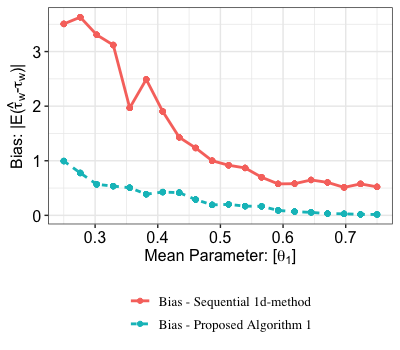}
	\end{minipage}
	\caption{\footnotesize{{\it Visualized results of Simulation A with respect to $\tau^0_w$ (horizontal axis). Left Panel: y-axis: Jump Sizes $\psi_w^2$ obtained by a sequential 1d approach, and $\xi_w^2$ obtained by proposed 2d-approach.  Right panel: y-axis- Monte-carlo approximation of bias in $\h\tau_w$ obtained over $500$ replications. Observed jump size under proposed approach is uniformly larger and bias is uniformly and proportionately lower for proposed method}}}
	\label{fig:sim1dcomp}
\end{figure}


 \subsection{Simulation B: Estimation and Inference}

This section provides numerical support to proposed Algorithm \ref{alg:single} and the inference results of Theorem \ref{thm:wc.vanishing} and Theorem \ref{thm:wc.non.vanishing}. In all cases considered, the mean vectors are set as  $\theta_{(1)}^0=\theta_{(3)}^0=\big(\theta_{s\times 1}^T,0...,0\big)^T_{p\times 1},$ where $\theta=(0.75,...,0.25)_{s\times 1},$ contains evenly spaced $s=5$ entries. The remaining two means are set to zero, i.e.,  $\theta_{(2)}^0=\theta_{(4)}^0=0_{p\times 1}.$  We consider all combinations of the sampling periods  $T_w,T_h\in\{30,35,40,45\},$  dimension $p\in\{10,50,100,250\}.$ The $2d$-change point $\tau^0=(\tau_w^0,\tau^0_h)^T$ chosen as all combinations of $\tau_w^0\in \big\{\lfloor 0.2\cdotp T_w\rfloor,\lfloor 0.4\cdotp T_w\rfloor,\lfloor 0.6\cdotp T_w\rfloor,\lfloor 0.8\cdotp T_w\rfloor\big\}$ and   $\tau_h^0\in \big\{\lfloor 0.2\cdotp T_h\rfloor,\lfloor 0.4\cdotp T_h\rfloor,\lfloor 0.6\cdotp T_h\rfloor,\lfloor 0.8\cdotp T_h\rfloor\big\}.$ 

We construct confidence intervals using both the limiting distributions of Theorem \ref{thm:wc.vanishing} and Theorem \ref{thm:wc.non.vanishing}. The significance level is set to $\al=0.05.$ Confidence intervals are constructed component-wise, i.e., for the width change parameter as, $\big[(\tilde\tau_w-ME_w),\, (\tilde\tau_w+ME_w)\big],$ where $\tilde\tau_w$ is the width component of the output of Algorithm 1.The margin of error ($ME_w$) is computed as  $ME_w=q_{\alpha}^v\si^2_{(w,\iny)}\big/\big(T_h\xi^2_w\big)$ or $ME=q_{\al}^{nv}$ based on the results of Theorem \ref{thm:wc.vanishing} and Theorem \ref{thm:wc.non.vanishing}, respectively. Here $q_{\al}^v$ represents the $\big(1-\alpha/2\big)^{th}$ quantile of the argmax of two sided negative drift Brownian motion of Theorem \ref{thm:wc.vanishing}. This critical value is evaluated as $c_{\alpha}=11.03$ by using its distribution function provided in \cite{yao1987approximating}. The $\big(1-\alpha/2\big)^{th}$ quantile $q_{\al}^{nv}$ of the argmax of the two sided negative drift random walk is computed as its monte carlo approximation by simulating $4000$ realizations of this distribution. As per the assumed data generating process, the distribution $\cP$ here is Gaussian. The implementation of the confidence interval, we utilize plugin estimates of $\si^2_{(w,\iny)}$ and $\xi^2_w,$ whose computational details of which are provided in Appendix \ref{app:numerical.supplement} of the supplement. Symmetrical computations are carried out for the height parameter $\tau_h.$

\begin{wrapfigure}{r}{0.6\textwidth}
	\centering
	\resizebox{0.6\textwidth}{!}{
		\begin{tabular}{llllllll}
			\toprule
			\multicolumn{2}{c}{\multirow{2}{*}{\begin{tabular}[c]{@{}c@{}}$T_h=30,$\\ $\tau^0_h/T_h=0.2$\end{tabular}}} & \multicolumn{3}{c}{$p=100$}                                                                                          & \multicolumn{3}{c}{$p=250$}                                                                                          \\ \cmidrule{3-8}
			\multicolumn{2}{c}{}                                                                                        & \multicolumn{1}{c}{\multirow{2}{*}{bias (rmse)}} & \multicolumn{2}{c}{coverage (av. ME)}                             & \multicolumn{1}{c}{\multirow{2}{*}{bias (rmse)}} & \multicolumn{2}{c}{coverage (av. ME)}                             \\ \cmidrule{1-2} \cmidrule{4-5} \cmidrule{7-8}
			\multicolumn{1}{c}{$\tau^0_w/T_w$}                         & \multicolumn{1}{c}{$T_w$}                        & \multicolumn{1}{c}{}                             & \multicolumn{1}{c}{Vanishing} & \multicolumn{1}{c}{Non-Vanishing} & \multicolumn{1}{c}{}                             & \multicolumn{1}{c}{Vanishing} & \multicolumn{1}{c}{Non-Vanishing} \\ \midrule
			0.2                                                      & 30                                               & 0.032 (0.738)                                    & 0.966 (0.352)                 & 0.966 (0.002)                     & 0.024 (0.74)                                     & 0.968 (0.306)                 & 0.968 (0.002)                     \\
			0.2                                                      & 35                                               & 0.016 (0.167)                                    & 0.972 (0.362)                 & 0.972 (0)                         & 0.008 (0.19)                                     & 0.964 (0.325)                 & 0.964 (0)                         \\
			0.2                                                      & 40                                               & 0.008 (0.167)                                    & 0.972 (0.375)                 & 0.972 (0)                         & 0.01 (0.173)                                     & 0.97 (0.342)                  & 0.97 (0)                          \\
			0.2                                                      & 45                                               & 0.026 (0.173)                                    & 0.97 (0.381)                  & 0.97 (0)                          & 0.01 (0.615)                                     & 0.964 (0.35)                  & 0.964 (0.002)                     \\ \midrule
			0.4                                                      & 30                                               & 0.048 (0.29)                                     & 0.952 (0.457)                 & 0.952 (0.002)                     & 0.104 (0.544)                                    & 0.938 (0.438)                 & 0.938 (0)                         \\
			0.4                                                      & 35                                               & 0.04 (0.303)                                     & 0.962 (0.463)                 & 0.962 (0)                         & 0.086 (0.546)                                    & 0.932 (0.451)                 & 0.932 (0)                         \\
			0.4                                                      & 40                                               & 0.038 (0.224)                                    & 0.95 (0.47)                   & 0.95 (0)                          & 0.044 (0.261)                                    & 0.956 (0.455)                 & 0.956 (0)                         \\
			0.4                                                      & 45                                               & 0.024 (0.245)                                    & 0.962 (0.481)                 & 0.962 (0)                         & 0.028 (0.228)                                    & 0.97 (0.463)                  & 0.97 (0.002)                      \\ \midrule
			0.6                                                      & 30                                               & 0.142 (0.801)                                    & 0.922 (0.49)                  & 0.922 (0.01)                      & 0.122 (0.852)                                    & 0.928 (0.467)                 & 0.93 (0.01)                       \\
			0.6                                                      & 35                                               & 0.068 (0.346)                                    & 0.94 (0.491)                  & 0.94 (0)                          & 0.114 (0.68)                                     & 0.934 (0.481)                 & 0.934 (0.008)                     \\
			0.6                                                      & 40                                               & 0.068 (0.696)                                    & 0.954 (0.498)                 & 0.954 (0.004)                     & 0.072 (0.369)                                    & 0.942 (0.485)                 & 0.942 (0.004)                     \\
			0.6                                                      & 45                                               & 0.046 (0.272)                                    & 0.958 (0.505)                 & 0.96 (0.004)                      & 0.066 (0.387)                                    & 0.948 (0.49)                  & 0.954 (0.008)                     \\ \midrule
			0.8                                                      & 30                                               & 0.276 (1.31)                                     & 0.856 (0.444)                 & 0.858 (0.022)                     & 0.568 (2.416)                                    & 0.776 (0.424)                 & 0.78 (0.034)                      \\
			0.8                                                      & 35                                               & 0.12 (0.639)                                     & 0.928 (0.456)                 & 0.93 (0.01)                       & 0.27 (0.838)                                     & 0.82 (0.433)                  & 0.82 (0.018)                      \\
			0.8                                                      & 40                                               & 0.176 (0.699)                                    & 0.884 (0.468)                 & 0.892 (0.018)                     & 0.168 (0.593)                                    & 0.88 (0.431)                  & 0.882 (0.01)                      \\
			0.8                                                      & 45                                               & 0.094 (0.377)                                    & 0.928 (0.466)                 & 0.928 (0.006)                     & 0.294 (1.599)                                    & 0.866 (0.445)                 & 0.866 (0.014)                     \\ \bottomrule
	\end{tabular}}
\captionof{table}{\footnotesize{Simulation results for estimation of $\tau^0_w$ based on 500 replications. All reported metrics rounded to three decimals. Other data generating parameters: $T_h=30,$ $\tau^0_h=\lfloor 0.2\cdotp T_h\rfloor$ and $p\in\{100,250\}.$}}
	\label{tab:wres2}
\end{wrapfigure}

Partial results are provided in Table \ref{tab:wres2}, the remainder can be found in Appendix \ref{app:numerical.supplement}. Change estimates in both directions are observed to exhibit little bias with an expected deterioration with larger dimension sizes $p.$ The proposed inference methodology is observed to provide good control at the nominal significance level if one provides a slight leeway given the discrete nature of the underlying problem and confidence intervals. The cases where coverage is observed to be significantly away from nominal are again the larger values of $p,$ or where change points are near the parametric boundary (see, Table \ref{tab:wres2}, case: $p=250,$ $\tau^0=0.8$). It can be observed that this deviation from nominal coverage is primarily due to bias in estimation and not the computation of margin of error (the margin of error as expected remains stable for all values of $T_w$).  Importantly, bias is observed to diminish and coverage to catch up to nominal as the effective sample size of $T_wT_h\underline\om$ increases, i.e., when the sampling periods increase or the change point moves away from parametric boundaries.

\section{Application: Temperature patterns in the Pacific North-West Region of the United States}\label{sec:real.data}

\begin{wrapfigure}{r}{0.35\textwidth}
{\centering
		\includegraphics[width=0.35\textwidth]{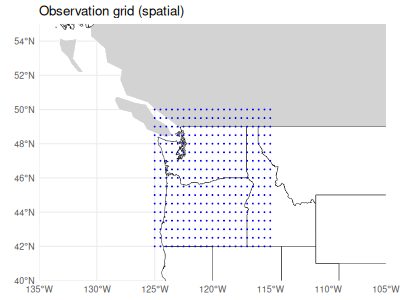}
	\caption{\footnotesize{Geographical grid of $p=357$ locations (blue dots). Daily 2m-Temperature data is collected for each grid point over $25$ years.}}}
	\label{fig:obsgrid}
\end{wrapfigure}

We consider daily $2m$-temperature data over the Pacific Northwest (PNW) region of United States and Canada, located between latitudes $42^\circ N$ and $50^\circ N,$ and longitudes $115^\circ W$ and $125^\circ W,$ encompassing parts of Washington, Oregon, Idaho, and British Columbia. Washington produces approximately 60\% of the apples consumed in the U.S, Idaho grows nearly one-third of the country's potatoes, and Oregon produces over 99\% of the nation's hazelnuts, contributing to the multi-billion dollar agricultural economy of the Pacific Northwest. Data is collected from the Copernicus Climate Change Service (\url{https://cds.climate.copernicus.eu/}). It is collected at a spatial resolution of $0.5^\circ\times 0.5^\circ$ which results in a total of $p=357$ spatial grid points over the considered PNW region. Temporal resolution is of one day over a period of $T_w=25$ years, from January 1, 2000, to December 31, 2024. All leap day's were withheld and thus we had $T_h=365.$  {\it ERA5} is a comprehensive global atmospheric reanalysis product compiled by the European Centre for Medium-Range Weather Forecasts. Its products have been extensively used in climate data analysis, e.g.,  (\citep{JIANG2021125660, nogueira2020inter}) amongst many others.

We first allow for boundary values of the change point by utilizing the conventional approach of an $\ell_0$ regularization as described in the following remark.

\begin{remark}\label{rem:boundary} {\bf(Boundary cases of $\tau^0=(T_w,\tau_h^0)^T,$ $\tau^0=(\tau^0_w,T_h)^T$ or $\tau^0=(T_w,T_h)^T$)}. \textnormal{This can be handled by replacing Step 1 of Algorithm \ref{alg:single} with a $0$-norm regularized version,}
	\benr\label{eq:45}
	\h\tau^*_w&=&\begin{cases} 	T_w			  &	 {\rm if}\,\, \{\cL(T_w,\h\tau_h,\h\theta)-\cL(\h\tau_w,\h\tau_h,\h\theta)\}<\g_w, \\
		\h\tau_w & {\rm else},
	\end{cases}\nn\\
	\h\tau^*_h&=&\begin{cases} 	T_h			  &	 {\rm if}\,\, \{\cL(\h\tau_w,T_h,\h\theta)-\cL(\h\tau_h, \h\tau_h,\check\theta)\}<\g_h, \\
		\h\tau_h & {\rm else},
	\end{cases}
	\eenr
	\textnormal{Here $\g_w,\g_h$ are tuning parameters and $\h\tau_w$ and $\h\tau_h$ are from Step 1 of Algorithm \ref{alg:single}. This regularization is common in the literature, e.g., \cite{fryzlewicz2014wild, wang2018high},   it declares the presence or absence of a change. Selection consistency \big($pr(\tilde\tau_w=T_w)\to 1,$ $T_w\to\iny,$ when $\tau^0_w=T_w$ and symmetrical for $\tilde\tau_h.$\big)  can be also be verified via conventional arguments, see, e.g. \cite{Kaul2019}. Algorithm \ref{alg:single.regularized} below aggregates all steps to also search for boundary values.}
\end{remark}

	\begin{algorithm}\caption{Estimation of $\tau^0=(\tau^0_w,\tau^0_h)^T$ with boundary selection}
	\label{alg:single.regularized}
		\begin{algorithmic}[1]

			\Statex Initialize a user chosen change point $\check\tau=(\check\tau_w,\check\tau_h),$

			\State Compute estimates $\check\theta_{(j)}=\h\theta_{(j)}(\check\tau),$ $j=1,2,3,4.$ and change point estimate $\h\tau=(\h\tau_y,\h\tau_d)^T$ as Step 1 of Algorithm \ref{alg:single}. Additionally perform selection as $\h\tau^*=(\h\tau_w^*,\h\tau_h^*)^T$ of (\ref{eq:45}).
			\State If $\h\tau^*_w=T_w,$ or $\h\tau^*_h=T_h,$ then set $\tilde\tau_w=T_w,$ or $\tilde\tau_h=T_h,$ respectively, else update to $\h\theta_{(j)}=\h\theta_{(j)}(\h\tau),$ $j=1,2,3,4,$ and update,
			\benr
			\tilde\tau_w=\argmin_{1\le \tau_w<T_w} \cL(\tau_w,\h\tau_h^*,\h\theta),\quad
			\tilde\tau_h=\argmin_{1\le \tau_h<T_h} \cL(\h\tau_w^*,\tau_h,\h\theta)\nn
			\eenr
			\Statex Output: $\tilde\tau=(\tilde\tau_w,\tilde\tau_h).$
		\end{algorithmic}
	\end{algorithm}
	
Algorithm \ref{alg:single.regularized} is applied recursively to search for multiple 2d change points, i.e., each segment is further partitioned until no further changes can be found. This is the two dimensional analog of the commonly utilized one dimensional approach of binary segmentation. The additional tuning parameters $\g_w,\g_h$ are also chosen by first setting $\g_w=\g_h=\g$ and then applying a BIC criteria.

The method segments the temporal $25\times 365$ region of the observed $357\times 25\times 365$ tensor into $46$ contiguous partitions. We present cross-sectional visualizations with respect to three spatial locations, namely, Seattle, Spokane and Pullman\footnote{The closest grid point was utilized for corresponding city}. The first two are the most populous cities in Washington. The latter the largest city in Whitman County, which is a major wheat producer.    

\begin{figure}[]
	\centering
	\begin{minipage}{0.3\textwidth}
	\includegraphics[width=\textwidth]{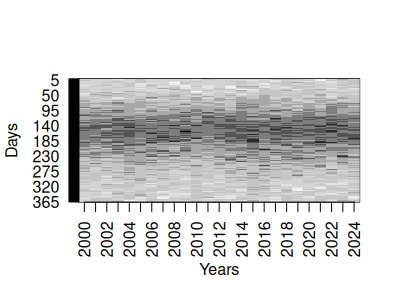}
	\end{minipage}
	\hspace{3mm}
	\begin{minipage}{0.3\textwidth}
	\includegraphics[width=\textwidth]{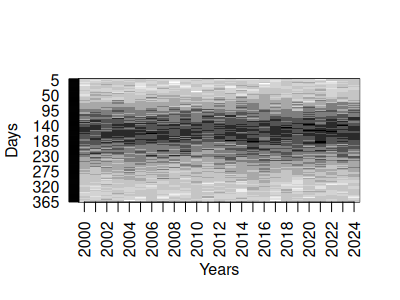}
	\end{minipage}
	\hspace{1mm}
	\begin{minipage}{0.3\textwidth}
	\includegraphics[width=\textwidth]{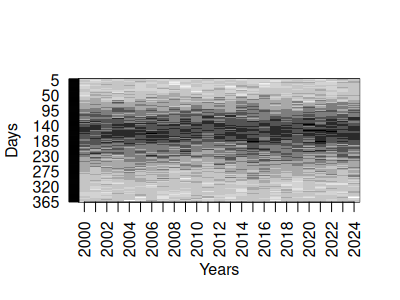}
	\end{minipage}

\vspace{-5mm}
	\begin{minipage}{0.3\textwidth}
	\includegraphics[width=\textwidth]{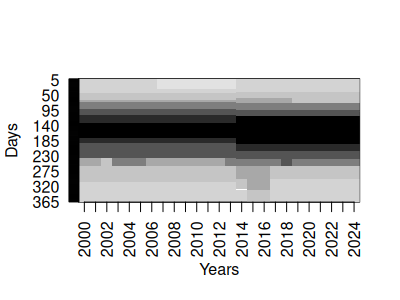}
	\end{minipage}
	\hspace{3mm}
	\begin{minipage}{0.3\textwidth}
	\includegraphics[width=\textwidth]{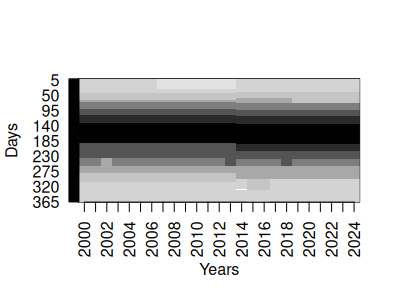}
	\end{minipage}
	\hspace{1mm}
	\begin{minipage}{0.3\textwidth}
	\includegraphics[width=\textwidth]{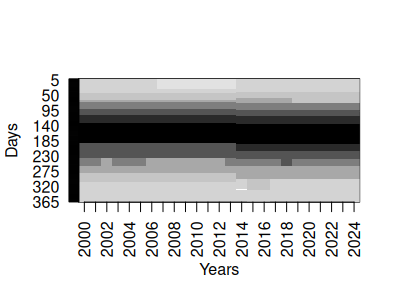}
	\end{minipage}
	\caption{\footnotesize{{\it Raw data (Top row) and Means over estimated partitions (Bottom row) data for Seattle (Left), Spokane (Center) and Pullman (right). Each pixel value represents the corresponding $2m$-temperature with a gradient of increasing temperature from {\it white} to {\it black}  }}}  
	\label{fig:raw.and.mean}
\end{figure}

Estimated partitions appear to have a primary segment at the year $2013$  visualized in Figure \ref{fig:raw.and.mean}. The mean temperature pre and post $2013$ was found to be $49.97F\to 51.20F$ $49.75F\to 50.90F,$ and $48.63F\to 49.87F$ at Seattle, Spokane and Pullman, respectively. In each case a pre and post means comparison was found to have $p-value\,< \,10^{-3}$ by the ordinary t-test. This increase does not come as a surprise and may also be gleaned from many studies pointing to the same trend including the IPCC report \cite{lee2023ipcc}. However, the additional important aspect brought out by our analysis pertains to how this excess heat is distributed within the year. This can be observed in Figure \ref{fig:raw.and.mean}, wherein one observes the  darker regions of the post segments are more pronounced mid year. This is more clearly visualized in Figure \ref{fig:cross.sectional} which plots a further cross section comparing one year from the pre 2013 and one from the post 2013 segment. At all three locations one observes a similar trend, specifically, excess heat explained largely by an early onset of summer and a more intense peak summer. The winters and early spring appear to remain somewhat unaffected. Precise mean values over these segments may also be gleaned from this analysis. 

\begin{figure}[htbp]
\centering
	\begin{minipage}{0.31\textwidth}
	\includegraphics[width=\textwidth]{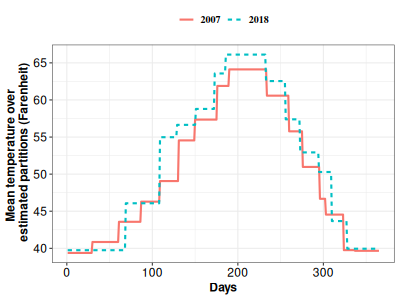}
	\end{minipage}
	\hspace{3mm}
	\begin{minipage}{0.31\textwidth}
	\includegraphics[width=\textwidth]{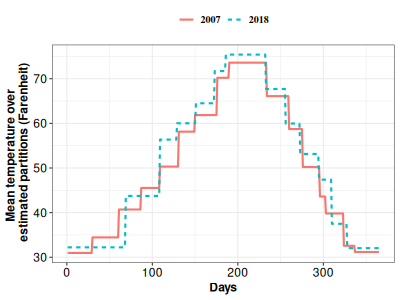}
	\end{minipage}
	\hspace{1mm}
	\begin{minipage}{0.31\textwidth}
	\includegraphics[width=\textwidth]{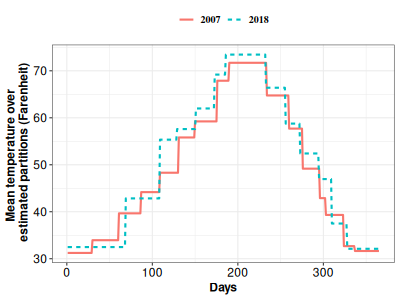}
	\end{minipage}
	\caption{\footnotesize{{\it Means over estimated partitions for Seattle (left), Spokane (center) and Pullman (right) for the years $2007$ (solid orange) and $2018$ (dashed, blue). Observe an uneven distribution of excess heat largely indicated by an early onset of summer.}}}  
	\label{fig:cross.sectional}
\end{figure}

Shifts in seasonal patterns have implications on agricultural practices. Food grains have specific tilling, sowing, filling and harvesting schedules with prescribed temperature ranges. E.g., per the Whole Grain Council (\url{https://wholegrainscouncil.org}), spring wheat (\textit{Triticum aestivum} L.), is sown March--May and harvested Aug--Sept  Tillering (Day of year 110--150) requires 12--25~°C, while grain filling (Day of year  180--220) requires 18--22°C \citep{Porter1999,  FAO2000}. Our analysis may be utilized to obtain times of the year when one may expect such conditions in a dynamic manner. Moreover, on also how long such conditions may last over the year.

\bibliographystyle{abbrvnat}
\bibliography{bibliography}

\appendix

\newpage
{\Large {\sc Supplementary Materials: High Dimensional Change Point Models for Two-Directional Data}}

\section{Proofs of results in Section \ref{sec:main}}
To present the arguments of this section we require additional notation. In all to follow let  $\h\eta_{(j)}$ represent estimates of the jump parameters $\eta^0_{(j)},$ $j=1,2,3,4$ respectively. We define,
\benr\label{def:cU}
&&\cU_w(\tau_w,\tau_h,\theta)=T_wT_h\Big(\cL(\tau_w,\tau_h,\theta)-\cL(\tau_w^0,\tau_h,\theta)\Big),\,\,{\rm and}\nn\\
&&\cU_h(\tau_w,\tau_h,\theta)=T_wT_h\Big(\cL(\tau_w,\tau_h,\theta)-\cL(\tau_w,\tau_h^0,\theta)\Big),\nn
\eenr
where $\cL(\cdotp,\cdotp,\cdotp)$ is the squared loss defined in (\ref{def:sq.loss}). Clearly, the plug-in estimates $\tilde\tau_w(\h\tau_h,\h\theta)$ and $\tilde\tau_h(\h\tau_w,\h\theta)$ of (\ref{est:optimal}) can then equivalently be written as,
\benr\label{est:optimal.cU}
\tilde\tau_w(\h\tau_h,\h\theta)=\argmin_{1\le\tau_w< T_w} \cU_w(\tau_w,\h\tau_h,\h\theta),\quad{\rm and}\quad \tilde\tau_h(\h\tau_w,\h\theta)=\argmin_{1\le\tau_h< T_h} \cU_h(\h\tau_w,\tau_h,\h\theta)
\eenr
The change of representation of estimates to (\ref{est:optimal.cU}) is made solely for notational convenience in the proofs to follow. 

\begin{lemma}\label{lem:unif.lb} Suppose the model (\ref{model:rvmcp}) and assume that Condition A, B, C(i)(a) and C(ii)(a,b) hold. Then, we have, 
	\benr\label{eq:3a}
	\cU_w(\tau_w,\h\tau_h,\h\theta)\ge \frac{T_h\xi^2_w}{2}\Big[|\tau_w-\tau_w^0|-c_u\log(p\vee T_wT_h)\frac{\si}{\xi_w}\Big\{\frac{s|\tau_w-\tau_w^0|}{T_wT_h}\Big\}^{\frac{1}{2}}\Big]
	\eenr
	with probability at least $1-2\exp\{-c_1\log(p\vee T)\}-\pi_T,$ uniformly over all possible $\tau_w.$ Here the constant $c_1>0$ that does not depend on any model parameters.
\end{lemma}

\begin{proof}[{\bf Proof of Lemma \ref{lem:unif.lb}}] The proof of this result is broken up into the following three steps.
	\begin{enumerate}[itemsep=0pt, labelindent=!]
		\item[{\bf Step 1}] Utilize Condition B, Condition C and results of Appendix \ref{app:deviation} to obtain upper and lower bounds on some stochastic quantities of interest.
		\item[{\bf Step 2}] Perform an algebraic decomposition of $\cU_w(\tau_w,\h\tau_h,\h\theta)$ in terms of jump sizes (\ref{def:jump.size.quad}) and additional noise terms.
		\item [{\bf Step 3}] Apply bounds of Step 1 to expression of Step 2 to obtain the desired uniform lower bound of the statement of this lemma.
	\end{enumerate}	
	
	\noi We begin with {\noi\bf Step 1} that provides a few observations that shall be required to obtain the desired lower bound. Using Condition C(ii)(a,b) we have the following relations,	
	\benr\label{eq:4}
	\|\h\eta_{(1)}-\eta^0_{(1)}\|_2&\le& \|\h\theta_{(1)}-\theta_{(1)}^0\|_2+\|\h\theta_{(2)}-\theta_{(2)}^0\|_2\le 2c_{u1}\xi_{\min}\quad{\rm and\,\, similarly},\nn\\
	\|\h\eta_{(1)}-\eta^0_{(1)}\|_1&\le& 4\surd{s}\|\h\theta_{(1)}-\theta_{(1)}^0\|_2+4\surd{s}\|\h\theta_{(2)}-\theta_{(2)}^0\|_2\le 8c_{u1}\surd{s}\xi_{\min}
	\eenr
	with probability at least $1-\pi_T.$ Here the third inequality follows from Condition C(ii)(a). Next, consider,
	\benr\label{eq:5}
	\|\h\eta_{(1)}\|_2&\le& \|\h\eta_{(1)}-\eta^0_{(1)}\|_2+\|\eta^0_{(1)}\|_2\le \xi_1 +2_{u1}\xi_{\min}\le c_u\xi_{\min},\quad{\rm and\,\, similarly},\nn\\
	\|\h\eta_{(1)}\|_1&\le&\|\h\eta_{(1)}-\eta^0\|_1+\|\eta^0_{(1)}\|_1\le 8c_{u1}\surd{s}\xi_{\min}+\surd{s}\xi_1\le c_u\surd{s}\xi_{\min},
	\eenr
	which holds with probability at least $1-\pi_T.$ Here the second inequality for the $\ell_2$ bound follows from (\ref{eq:4}) and the third follows from Condition B(ii). The $\ell_1$ bound follows analogously. Expression (\ref{eq:5}) provides an upper bound for $\|\h\eta_{(1)}\|_2$ that holds with probability at least $1-\pi_T,$ below we show that this quantity can also be bounded from below. Consider,
	\benr\label{eq:7}
	\|\h\eta_{(1)}\|_2^2&=&\|\eta^0_{(1)}+(\h\eta_{(1)}-\eta^0_{(1)})\|_2^2\ge \|\eta^0_{(1)}\|_2^2+2(\h\eta_{(1)}-\eta^0_{(1)})^T\eta^0_{(1)}\nn\\
	&\ge& \xi_1^2-2\|\h\eta_{(1)}-\eta^0_{(1)}\|_2\xi_1\ge \xi_1^2-2c_{u1}\xi_{\min}\overline\xi
	\eenr
	with probability at least $1-\pi_T.$ Here the second inequality follows by an application of the Cauchy-Schwarz inequality. The final inequality follows from (\ref{eq:4}). Analogous arguments can be utilized to obtain versions of the above bounds for $\|\h\eta_{(3)}\|_2,$ specifically,
	\benr\label{eq:8}
	&\|\h\eta_{(3)}-\eta^0_{(3)}\|_2\le  2c_{u1}\xi_{\min},\qquad \|\h\eta_{(3)}-\eta^0_{(3)}\|_1\le 8c_{u1}\surd{s}\xi_{\min},\\
	&\|\h\eta_{(3)}\|_2\le c_u\xi_{\min},\qquad \|\h\eta_{(3)}\|_1\le c_u\surd{s}\xi_{\min},\quad{\rm and}\quad \|\h\eta_{(3)}\|_2^2\ge \xi_3^2-2c_{u1}\xi_{\min}\overline\xi\nn
	\eenr
	with probability $1-\pi_T.$ Additional residual terms that shall require control are as follows,
	\benr\label{eq:9}
	2\om_h(\h\theta_{(1)}-\theta_{(1)}^0)^T\h\eta_{(1)}+2(1-\om_h)(\h\theta_{(4)}-\theta_{(4)}^0)^T\h\eta_{(3)}\hspace{4cm}\nn\\
	-2\frac{(\h\tau_h-\tau_h^0)}{T_h}(\h\theta_{(1)}-\theta_{(1)}^0)^T\h\eta_{(1)}
	+2\frac{(\h\tau_h-\tau_h^0)}{T_h}(\h\theta_{(4)}-\theta_{(1)}^0)^T\h\eta_{(3)}\nn\\
	+\frac{(\h\tau_h-\tau_h^0)}{T_h}\|\h\eta_{(3)}\|_2^2-\frac{(\h\tau_h-\tau_h^0)}{T_h}\|\h\eta_{(1)}\|_2^2\hspace{3.5cm}\nn\\
	\le 2\|\h\theta_{(1)}-\theta_{(1)}^0\|_2\|\h\eta_{(1)}\|_2+2\|\h\theta_{(4)}-\theta_{(4)}^0\|_2\|\h\eta_{(3)}\|_2\hspace{3cm}\nn\\
	+2\frac{|\h\tau_h-\tau_h^0|}{T_h}\|\h\theta_{(1)}-\theta_{(1)}^0\|_2\|\h\eta_{(1)}\|_2+2\frac{|\h\tau_h-\tau_h^0|}{T_h}\|\h\theta_{(4)}-\theta_{(1)}^0\|_2\|\h\eta_{(3)}\|_2\nn\\
	+\frac{|\h\tau_h-\tau_h^0|}{T_h}\|\h\eta_{(3)}\|_2^2+\frac{|\h\tau_h-\tau_h^0|}{T_h}\|\h\eta_{(1)}\|_2^2\hspace{4.65cm}\nn\\
	\le c_{u1}\xi_{\min}^2+c_{u1}\overline\xi\xi_{\min}\hspace{2.9in}
	\eenr
	with probability at least $1-\pi_T.$ Here the first inequality follows from several applications of the Cauchy-Schwarz inequality. The second inequality follows by utilizing Condition C(i)(a), C(ii)(b) as well as (\ref{eq:5}). Here we have also utilized the triangle inequality $\|\h\theta_{(4)}-\theta_{(1)}^0\|\le \|\h\theta_{(4)}-\theta_{(4)}^0\|+\|\eta^0_{(4)}\|.$ The above inequalities provide bounds on terms where the randomness is induced solely due to the plug-in preliminary estimates $\h\tau_h$ and $\h\theta.$ 
	
	The following provides upper bounds on stochastic terms where randomness is induced via the noise terms $\vep_{(w,h)}'$s. These shall follow mainly as a consequence of results of Appendix \ref{app:deviation}. Consider,
	\benr\label{eq:10}
	\Big|\sum_{w=\tau_w^0+1}^{\tau_w}\sum_{h=\tau_h^0+1}^{T_h}\vep_{(w,h)}^T\h\eta_{(1)}\Big|&\le& \Big\|\sum_{w=\tau_w^0+1}^{\tau_w}\sum_{h=\tau_h^0+1}^{T_h}\vep_{(w,h)}\Big\|_{\iny}\|\h\eta_{(1)}\|_1\nn\\
	&\le& c_u\xi_{\min}\si\log(p\vee T_wT_h)\surd\big(T_h(\tau_w-\tau^0_w)\big)\surd{s}
	\eenr
	with probability at least $1-2\exp\{c_1\log(p\vee T_wT_h)\}-\pi_T,$ for some $c_1>0.$ Here the second inequality follows from Lemma \ref{lem:nearoptimalcross.subE.special} and (\ref{eq:5}). The same argument also yields the same uniform bounds on other similar residual terms,
	\benr\label{eq:11}
	&&\Big|\sum_{w=\tau_w^0+1}^{\tau_w}\sum_{h=1}^{\tau_h^0}\vep_{(w,h)}^T\h\eta_{(3)}\Big|\le c_u\xi_{\min}\si\log(p\vee T_wT_h)\surd\big(T_h(\tau_w-\tau^0_w)\big)\surd{s}\nn\\
	&&\Big|\sum_{w=\tau_w^0+1}^{\tau_w}\sum_{h=\tau^0_h+1}^{\h\tau_h}\vep_{(w,h)}^T\h\eta_{(1)}\Big|\le c_u\xi_{\min}\si\log(p\vee T_wT_h)\surd\big(T_h(\tau_w-\tau^0_w)\big)\surd{s}\nn\\
	&&\Big|\sum_{w=\tau_w^0+1}^{\tau_w}\sum_{h=\tau^0_h+1}^{\h\tau_h}\vep_{(w,h)}^T\h\eta_{(3)}\Big|\le c_u\xi_{\min}\si\log(p\vee T_wT_h)\surd\big(T_h(\tau_w-\tau^0_w)\big)\surd{s}
	\eenr
	with probability at least $1-2\exp\{c_1\log(p\vee T_wT_h)\}-\pi_T.$ Here we have also utilized Condition C(i)(a) which provides $|\h\tau_h-\tau_h^0|\le c_{u1}T_h,$ w.p. $1-\pi_T.$ These bounds complete {\bf Step 1} of our argument and provide the necessary groundwork to proceed to the next step.
	
	\vspace{2mm}
	\noi{\bf Step 2:} We shall prove the statement of this lemma for the case $\tau_w\ge \tau_w^0$ and $\h\tau_h\ge \tau_h^0.$ The remaining three permutations of the ordering of $\tau_w\le \tau_w^0$ and $\h\tau_h\ge \tau_h^0,$ or  $\tau_w\ge \tau_w^0$ and $\h\tau_h\le \tau_h^0,$ or  $\tau_w\le \tau_w^0$ and $\h\tau_h\le \tau_h^0$ shall follow symmetrical arguments. Consider,
	\benr\label{eq:1a}
	\cU_w(\tau_w,\h\tau_h,\h\theta)&=&\cL(\tau_w,\h\tau_h,\h\theta)-\cL(\tau_w^0,\h\tau_h,\h\theta)\nn\\
	&=&\sum_{w=\tau_w+1}^{T_w}\sum_{h=\h\tau_h+1}^{T_h}\|x_{(w,h)}-\h\theta_{(1)}\|_2^2+\sum_{w=1}^{\tau_w}\sum_{h=\h\tau_h+1}^{T_h}\|x_{(w,h)}-\h\theta_{(2)}\|_2^2\nn\\
	&&+\sum_{w=1}^{\tau_w}\sum_{h=1}^{\h\tau_h}\|x_{(w,h)}-\h\theta_{(3)}\|_2^2+\sum_{w=\tau_w+1}^{T_w}\sum_{h=1}^{\h\tau_h}\|x_{(w,h)}-\h\theta_{(4)}\|_2^2\nn\\
	&&-\sum_{w=\tau_w^0+1}^{T_w}\sum_{h=\h\tau_h+1}^{T_h}\|x_{(w,h)}-\h\theta_{(1)}\|_2^2-\sum_{w=1}^{\tau_w^0}\sum_{h=\h\tau_h+1}^{T}\|x_{(w,h)}-\h\theta_{(2)}\|_2^2\nn\\
	&&-\sum_{w=1}^{\tau_w^0}\sum_{h=1}^{\h\tau_h}\|x_{(w,h)}-\h\theta_{(3)}\|_2^2-\sum_{w=\tau_w^0+1}^{T}\sum_{h=1}^{\h\tau_h}\|x_{(w,h)}-\h\theta_{(4)}\|_2^2
	\eenr
	
	An algebraic manipulation of (\ref{eq:1a}) yields the following expression.
	\benr\label{eq:2}
	\cU_w(\tau_w,\h\tau_h,\h\theta)&=&(\tau_w-\tau_w^0)\Bigg[(T_h-\tau_h^0)\Big\{\|\h\eta_{(1)}\|_2^2+2(\h\theta_{(1)}-\theta_{(1)}^0)^T\h\eta_{(1)}\Big\}\nn\\
	&&+\tau_h^0\Big\{\|\h\eta_{(3)}\|_2^2+2(\h\theta_{(4)}-\theta_{(4)}^0)^T\h\eta_{(3)}\Big\}\Bigg]\nn\\
	&&+(\tau_w-\tau_w^0)(\h\tau_h-\tau_h^0)\Bigg[\|\h\eta_{(3)}\|_2^2-\|\h\eta_{(1)}\|_2^2\nn\\
	&&-2(\h\theta_{(1)}-\theta_{(1)}^0)^T\h\eta_{(1)}+2(\h\theta_{(4)}-\theta_{(1)}^0)^T\h\eta_{(3)}\Bigg]\nn\\
	&&-2\sum_{w=\tau_w^0+1}^{\tau_w}\sum_{h=\tau_h^0+1}^{T_h}\vep_{(w,h)}^T\h\eta_{(1)}-2\sum_{w=\tau_w^0+1}^{\tau_w}\sum_{h=1}^{\tau_h^0}\vep_{(w,h)}^T\h\eta_{(3)}\nn\\
	&&+2\sum_{w=\tau_w^0+1}^{\tau_w}\sum_{h=\tau_h^0+1}^{\h\tau_h}\vep_{(w,h)}^T\h\eta_{(1)}-2\sum_{w=\tau_w^0+1}^{\tau_w}\sum_{h=\tau_h^0+1}^{\h\tau_h}\vep_{(w,h)}^T\h\eta_{(3)}
	\eenr
	The calculations yielding (\ref{eq:2}) from the defining equality (\ref{eq:1a}) are fairly tedious and in order to maintain continuity of the main argument of this lemma, these algebraic manipulations are presented as Remark \ref{rem:lem1.algebra} after the proof of this lemma.  We can now proceed to the final step of the argument where we shall apply the bounds obtained in Step 1 to the expression (\ref{eq:2}) in order to obtain the desired uniform lower bound.
	
	\vspace{2mm}
	\noi{\bf Step 3:} Recall that all constants in Step 1 above represented by $c_{u1}$ arise from the Condition C, where it is assumed as chosen to be suitably small enough. Now applying the bounds (\ref{eq:5}), (\ref{eq:8}), (\ref{eq:9}), (\ref{eq:10}) and (\ref{eq:11}) to the expression (\ref{eq:2}) we obtain,
	
	\benr\label{eq:12}
	\cU_w(\tau_w,\h\tau_h,\h\theta)&\ge& T_h(\tau_w-\tau^0_w)\Big[\om_h\xi_1^2+(1-\om_h)\xi_3^2-c_{u1}\xi_{\min}\overline\xi-c_{u1}\xi_{\min}^2\Big]\nn\\
	&&-c_u\xi_{\min}\si\log(p\vee T_wT_h)\surd\big(T_h(\tau_w-\tau_w^0)\big)\surd{s}\nn\\
	&\ge& \frac{T_h\xi_{w}^2}{2}\Big[(\tau_w-\tau_w^0)-c_u\log(p\vee T_wT_h)\frac{\si}{\xi_w}\Big(\frac{s((\tau_w-\tau_w^0))}{T_wT_h}\Big)^{\frac{1}{2}}\Big]
	\eenr
	with probability at least $1-2\exp\{c_1\log(p\vee T_wT_h)\}-\pi_T.$ To obtain the second inequality  we have used that by definition $\om_h\xi_1^2+(1-\om_h)\xi_3^2=\xi_w^2,$ and $\xi_{\min}\le \xi_{w},$ additionally from Condition B(ii) we have $\overline\xi\le c_u\xi_{\min}$ and that the constant $c_{u1}$ arises from Condition C where it is chosen to be suitable small enough. Repeating symmetrical arguments for the mirroring permutations of the ordering of $\tau_w,\tau_h$ with respect to $\tau_w^0,\h\tau_h$ shall yield the same uniform lower bound (\ref{eq:12}). This completes the proof of this lemma.
\end{proof}

\begin{remark}\label{rem:lem1.algebra} \textnormal{This remark provides the decomposition (\ref{eq:2}) of (\ref{eq:1a}) described in {\bf Step 2} of the proof of Lemma \ref{lem:unif.lb}. A note of interest here is that the calculations below become much more intuitive when viewed w.r.t. to a $2d$-visualization such as that illustrated in (\ref{model:rvmcp}). Let,}
	\benr
	\h\vep_{(w,h)}=\begin{cases} x_{(w,h)}-\h\theta_{(1)} & \tau_w^0< w\le \tau_w,\,\, \h\tau_h<h\le T\\
		x_{(w,h)}-\h\theta_{(4)} & \tau_w^0< w\le \tau_w,\,\, 1\le h<\h\tau_h\nn
	\end{cases}
	\eenr
	\textnormal{Then picking up from the expression (\ref{eq:1a}), a simplification now yields,}
	\benr\label{eq:16}
	\cU(\tau_w,\h\tau_h,\h\theta)&=&-\sum_{w=\tau_w^0+1}^{\tau_w}\sum_{h=\h\tau_h+1}^{T}\|x_{(w,h)}-\h\theta_{(1)}\|_2^2
	+\sum_{w=\tau_w^0+1}^{\tau_h}\sum_{\h\tau_h+1}^{T}\|x_{(w,h)}-\h\theta_{(2)}\|_2^2\nn\\
	&&+\sum_{w=\tau_w^0+1}^{\tau_w}\sum_{h=1}^{\h\tau_h}\|x_{(w,h)}-\h\theta_{(3)}\|_2^2
	-\sum_{w=\tau_w^0+1}^{\tau_w}\sum_{h=1}^{\h\tau_h}\|x_{(w,h)}-\h\theta_{(4)}\|_2^2\nn\\
	&=&\sum_{w=\tau_w^0+1}^{\tau_w}\sum_{h=\h\tau_h+1}^{T}\|\h\theta_{(2)}-\h\theta_{(1)}\|_2^2-2\sum_{w=\tau_w^0+1}^{\tau_w}\sum_{h=\h\tau_h+1}^{T}\h\vep_{(w,h)}^T(\h\theta_{(2)}-\h\theta_{(1)})\nn\\
	&&+\sum_{w=\tau_w^0+1}^{\tau_w}\sum_{h=1}^{\h\tau_h}\|\h\theta_{(3)}-\h\theta_{(4)}\|_2^2-2\sum_{w=\tau_w^0+1}^{\tau_w}\sum_{h=1}^{\h\tau_h}\h\vep_{(w,h)}^T(\h\theta_{(3)}-\h\theta_{(4)})\nn\\
	&=&R1-R2+R3-R4.
	\eenr
	\textnormal{Next we consider these remainder terms $R1,$ and $R3,$ where we have,}
	\benr\label{eq:13}
	R1+R3&=&\sum_{w=\tau_w^0+1}^{\tau_w}\sum_{h=\h\tau_h+1}^{T}\|\h\theta_{(2)}-\h\theta_{(1)}\|_2^2+\sum_{w=\tau_w^0+1}^{\tau_w}\sum_{h=1}^{\h\tau_h}\|\h\theta_{(3)}-\h\theta_{(4)}\|_2^2\nn\\
	&=&\sum_{w=\tau_w^0+1}^{\tau_w}\sum_{h=\tau_h^0+1}^{T}\|\h\theta_{(2)}-\h\theta_{(1)}\|_2^2+\sum_{w=\tau_w^0+1}^{\tau_w}\sum_{h=1}^{\tau_h^0}\|\h\theta_{(3)}-\h\theta_{(4)}\|_2^2\nn\\
	&&-\sum_{w=\tau_w^0+1}^{\tau_w}\sum_{h=\tau_h^0+1}^{\h\tau_h}\|\h\theta_{(2)}-\h\theta_{(1)}\|_2^2 +\sum_{w=\tau_w^0+1}^{\tau_w}\sum_{h=\tau_h^0+1}^{\h\tau_h}\|\h\theta_{(3)}-\h\theta_{(4)}\|_2^2\nn\\
	&=& (\tau_w-\tau_w^0)\Big[(T-\tau_h^0)\|\h\theta_{(2)}-\h\theta_{(1)}\|_2^2+\tau_h^0\|\h\theta_{(3)}-\h\theta_{(4)}\|_2^2\Big]\nn\\
	&&+(\tau_w-\tau_w^0)(\h\tau_h-\h\tau_h^0)\Big[\|\h\theta_{(3)}-\h\theta_{(4)}\|_2^2-\|\h\theta_{(2)}-\h\theta_{(1)}\|_2^2\Big]
	\eenr
	
	\textnormal{In order to simplify the terms $R2$ and $R4,$ the double sums under consideration are split at the underlying change point parameters $\tau_w^0,\tau_h^0,$ leading to the following decompositions.}
	\benr\label{eq:14}
	R2&=&2\sum_{w=\tau_w^0+1}^{\tau_w}\sum_{h=\h\tau_h+1}^{T}\h\vep_{(w,h)}^T(\h\theta_{(2)}-\h\theta_{(1)})\nn\\
	&=&2\sum_{w=\tau_w^0+1}^{\tau_w}\sum_{h=\tau_h^0+1}^{T}\h\vep_{(w,h)}^T(\h\theta_{(2)}-\h\theta_{(1)})-2\sum_{w=\tau_w^0+1}^{\tau_w}\sum_{h=\tau_h^0+1}^{\h\tau_h}\h\vep_{(w,h)}^T(\h\theta_{(2)}-\h\theta_{(1)})\nn\\
	&=&2\sum_{w=\tau_w^0+1}^{\tau_w}\sum_{h=\tau_h^0+1}^{T}\vep_{(w,h)}^T(\h\theta_{(2)}-\h\theta_{(1)})-2\sum_{w=\tau_w^0+1}^{\tau_w}\sum_{h=\tau_h^0+1}^{T}(\h\theta_{(1)}-\theta_{(1)}^0)(\h\theta_{(2)}-\h\theta_{(1)})\nn\\
	&&-2\sum_{w=\tau_w^0+1}^{\tau_w}\sum_{h=\tau_h^0+1}^{\h\tau_h}\vep_{(w,h)}^T(\h\theta_{(2)}-\h\theta_{(1)})\nn\\
	&&+2\sum_{w=\tau_w^0+1}^{\tau_w}\sum_{h=\tau_h^0+1}^{\h\tau_h}(\h\theta_{(1)}-\theta_{(1)}^0)(\h\theta_{(2)}-\h\theta_{(1)})
	\eenr
	\textnormal{Similarly, we have,}
	\benr\label{eq:15}
	R4&=&2\sum_{w=\tau_w^0+1}^{\tau_w}\sum_{h=1}^{\h\tau_h}\h\vep_{(w,h)}^T(\h\theta_{(3)}-\h\theta_{(4)})\nn\\
	&=&2\sum_{w=\tau_w^0+1}^{\tau_w}\sum_{h=1}^{\tau_h^0}\h\vep_{(w,h)}^T(\h\theta_{(3)}-\h\theta_{(4)})+2\sum_{w=\tau_w^0+1}^{\tau_w}\sum_{h=\tau_h^0+1}^{\h\tau_h}\h\vep_{(w,h)}^T(\h\theta_{(3)}-\h\theta_{(4)})\nn\\
	&=&2\sum_{w=\tau_w^0+1}^{\tau_w}\sum_{h=1}^{\tau_h^0}\vep_{(w,h)}^T(\h\theta_{(3)}-\h\theta_{(4)})-2\sum_{w=\tau_w^0+1}^{\tau_w}\sum_{h=1}^{\tau_h^0}(\theta_{(4)}-\h\theta_{(4)})(\h\theta_{(3)}-\h\theta_{(4)})\nn\\
	&& +2\sum_{w=\tau_w^0+1}^{\tau_w}\sum_{h=\tau_h^0+1}^{\h\tau_h}\vep_{(w,h)}^T(\h\theta_{(3)}-\h\theta_{(4)})\nn\\
	&&-2\sum_{w=\tau_w^0+1}^{\tau_w}\sum_{h=\tau_h^0+1}^{\h\tau_h}(\h\theta_{(4)}-\theta_{(1)}^0)(\h\theta_{(3)}-\h\theta_{(4)})
	\eenr
	
	\textnormal{Substituting (\ref{eq:13}), (\ref{eq:14}) and (\ref{eq:15}) in (\ref{eq:16}) we obtain,}
	\benr\label{eq:17}
	\cU(\tau_w,\h\tau_h,\h\theta)&=&R1-R2+R3-R4\nn\\
	&=&(\tau_w-\tau_w^0)\Bigg[(T-\tau_h^0)\Big\{\|\h\theta_{(2)}-\h\theta_{(1)}\|_2^2+2(\h\theta_{(1)}-\theta_{(1)}^0)^T(\h\theta_{(2)}-\h\theta_{(1)})\Big\}\nn\\
	&&+\tau_h^0\Big\{\|\h\theta_{(3)}-\h\theta_{(4)}\|_2^2+2(\h\theta_{(4)}-\theta_{(4)}^0)^T(\h\theta_{(3)}-\h\theta_{(4)})\Big\}\Bigg]\nn\\
	&&+(\tau_w-\tau_w^0)(\h\tau_h-\h\tau_h^0)\Bigg[\|\h\theta_{(3)}-\h\theta_{(4)}\|_2^2-\|\h\theta_{(2)}-\h\theta_{(1)}\|_2^2\nn\\
	&&-2(\h\theta_{(1)}-\theta_{(1)}^0)^T(\h\theta_{(2)}-\h\theta_{(1)})+2(\h\theta_{(4)}-\theta_{(1)}^0)^T(\h\theta_{(3)}-\h\theta_{(4)})\Bigg]\nn\\
	&&-2\sum_{w=\tau_w^0+1}^{\tau_w}\sum_{h=\tau_h^0+1}^T\vep_{(w,h)}^T(\h\theta_{(2)}-\h\theta_{(1)})-2\sum_{w=\tau_w^0+1}^{\tau_w}\sum_{h=1}^{\tau_h^0}\vep_{(w,h)}^T(\h\theta_{(3)}-\h\theta_{(4)})\nn\\
	&&+2\sum_{w=\tau_w^0+1}^{\tau_w}\sum_{h=\tau_h^0+1}^{\h\tau_h}\vep_{(w,h)}^T(\h\theta_{(2)}-\h\theta_{(1)})\nn\\
	&&-2\sum_{w=\tau_w^0+1}^{\tau_w}\sum_{h=\tau_h^0+1}^{\h\tau_h}\vep_{(w,h)}^T(\h\theta_{(3)}-\h\theta_{(4)})
	\eenr
	\textnormal{The expression (\ref{eq:2}) now follows from (\ref{eq:17}) with notational changes. This completes this algebraic part of the proof of Lemma \ref{lem:unif.lb}.}
\end{remark}

%
%
%
%
%

\begin{proof}[{\bf Proof of Theorem \ref{thm:cpoptimal}}] This proof is divided into two parts. We begin with {\bf Parts (1a)} the proof of Part (1b) follows by symmetrical arguments and shall thus be omitted. 
	
	As a consequence of Lemma \ref{lem:unif.lb}, we have, 
	\benr\label{eq:new11}
	\cU_w(\tau_w,\h\tau_h,\h\theta)\ge \frac{T_h\xi^2_w}{2}\Big[(\tau_w-\tau_w^0)-c_u\log(p\vee T_wT_h)\frac{\si}{\xi_w}\Big\{\frac{s(\tau_w-\tau_w^0)}{T_h}\Big\}^{\frac{1}{2}}\Big]
	\eenr
	with probability at least $1-2\exp\{-c_{1}\log (p\vee T_wT_h)\}-\pi_T.$ Upon choosing,
	\benrr
	(\tau_w-\tau_w^0)\ge  T_w\Big\{c_u\log(p\vee T_wT_h)\frac{\si}{\xi_w}\Big\}^{2}\Big\{\frac{s}{T_wT_h}\Big\},
	\eenrr
	we have from (\ref{eq:new11}) that $\cU_w(\tau_w,\h\tau_h,\h\theta)>0,$ for all possible choices of $\tau_w,$ with probability  at least $1-2\exp\{-c_{1}\log (p\vee T_wT_h)\}-\pi_T.$ This together with the definition of $\h\tau_w$ implies that $\h\tau_w$ cannot belong in this exclusion set. Consequently, we must have, 
	\benr
	(\h\tau_w-\tau_w^0)\le  T_w\Big\{c_u\log(p\vee T_wT_h)\frac{\si}{\xi_w}\Big\}^{2}\Big\{\frac{s}{T_wT_h}\Big\},
	\eenr 
	with the same probability. 
	
	\vspace{3mm}
	Next we prove {\bf Part (2a)}.  The proof of Part (2b) follows by symmetrical arguments and shall thus be omitted. The broad structure is similar as it shall also follow by a contradiction argument. One distinction being the availability of a sharper assumption of Condition C(ii)(c) on the preliminary mean estimates.
	
	Recall the inequalities of {\bf Step 1} of Lemma \ref{lem:unif.lb} and note that analogous to (\ref{eq:4}), one may obtain that under Condition C(ii)(c) that,
	\benr\label{eq:25}
	\|\h\eta_{(1)}-\eta^0_{(1)}\|_1\le \frac{8c_{u1}\xi_{\min}}{\log(p\vee T_wT_h)},\quad{\rm and}\quad
	\|\h\eta_{(3)}-\eta^0_{(3)}\|_1\le \frac{8c_{u1}\xi_{\min}}{\log(p\vee T_wT_h)},
	\eenr
	with probability at least $1-\pi_T.$ Moreover, also observe here that since Condition C(ii)(c) assumed here is sharper than C(ii)(b) assumed in Lemma \ref{lem:unif.lb} consequently, the bounds  (\ref{eq:5}), (\ref{eq:7}), (\ref{eq:8}) and (\ref{eq:9}) remain valid here as well, with the same probability.
	
	Next we examine the stochastic terms considered in (\ref{eq:10}) and (\ref{eq:11}) more closely. Applying Lemma \ref{lem:optimalcross}, for any $0<a<1$ we have a $c_a>0,$ such that,
	\benr\label{eq:24}
	\sup_{(\tau_w-\tau_w^0)\ge c_aT_h^{-1}\phi^2\xi_w^{-2}}\frac{1}{(\tau_w-\tau_w^0)}\Big|\sum_{w=\tau^0_w+1}^{\tau_w}\Big(\sum_{h=\tau^0_h+1}^{T_h}\vep_{(w,h)}^T\eta^0_{(1)}+\sum_{h=1}^{\tau^0_h}\vep_{w,h}^T\eta^0_{(3)}\Big)\Big|\le \frac{T_h\xi_w^2}{2}
	\eenr
	with probability at least $1-a.$ Additionally, we also have that,
	\benr\label{eq:29}
	\Big|\sum_{w=\tau^0_w+1}^{\tau_w}\sum_{h=\tau^0_h+1}^{T_h}\vep_{(w,h)}^T(\h\eta_{(1)}-\eta^0_{(1)})\Big|&\le& \sup_{\substack{\tau_w\in\cG_w(u_{T_w},v_{T_w});\\ \tau_w\ge\tau^0_w}}\Big\|\sum_{w=\tau^0_w+1}^{\tau_w}\sum_{h=\tau^0_h+1}^{T_h}\vep_{(w,h)}\Big\|_{\iny}\big\|\h\eta_{(1)}-\eta^0_{(1)}\big\|_1\nn\\
	&\le& c_uc_{u1}\xi_{\min}\si\surd{(T_h(\tau_w-\tau_w^0))}
	\eenr	
	with probability at least $1-o(1)-\pi_T,$ uniformly over $\tau_w.$ Here we have utilized Lemma \ref{lem:nearoptimalcross.subE.special} as well as (\ref{eq:25}) to obtain the final inequality.

	Similarly, we can also obtain the following bounds,
	\benr\label{eq:26}
	&&\Big|\sum_{w=\tau^0_w+1}^{\tau_w}\sum_{h=1}^{\tau^0_h}\vep_{(w,h)}^T(\h\eta_{(3)}-\eta^0_{(3)})\Big|
	\le c_uc_{u1}\xi_{\min}\si\surd{(T_h(\tau_w-\tau_w^0))},\nn\\
	&&\Big|\sum_{w=\tau^0_w+1}^{\tau_w}\sum_{h=\tau_h^0+1}^{\h\tau_h}\vep_{(w,h)}^T\h\eta_{(1)}\Big| \le  c_uc_{u1}\xi_{\min}\si\surd{(T_h(\tau_w-\tau_w^0))},\nn\\
	&&\Big|\sum_{w=\tau^0_w+1}^{\tau_w}\sum_{h=\tau_h^0+1}^{\h\tau_h}\vep_{(w,h)}^T\h\eta_{(3)}\Big| \le  c_uc_{u1}\xi_{\min}\si\surd{(T_h(\tau_w-\tau_w^0))},
	\eenr	
	with probability at least $1-o(1)-\pi_T,$ uniformly over $\tau_w.$ In order to obtain the second and third inequalities of (\ref{eq:26}), we have utilized the $\ell_1$ bounds of  (\ref{eq:5}), (\ref{eq:8}). Additionally, towards these bounds we have also utilized Condition C(i)(b) together with Lemma \ref{lem:nearoptimalcross.subE.special}, in particular from Condition C(i)(b) we have that $|\h\tau_h-\tau_h^0|\le T_hu_{T_h},$ where $u_{T_h}=c_{u1}\big/\{s\log^2(p\vee T_wT_h)\},$ with probability $1-\pi_T,$ Lemma \ref{lem:nearoptimalcross.subE.special} is then applied with this choice of $u_{T_h}.$
	
	Next we consider {\bf Step 2} as described in the proof of Lemma \ref{lem:unif.lb}. Consider the decomposition (\ref{eq:2}) and note that it can be further manipulated as the following,
	\benr\label{eq:27}
	\cU_w(\tau_w,\h\tau_h,\h\theta)&=&(\tau_w-\tau_w^0)\Bigg[(T_h-\tau_h^0)\Big\{\|\h\eta_{(1)}\|_2^2+2(\h\theta_{(1)}-\theta_{(1)}^0)^T\h\eta_{(1)}\Big\}\nn\\
	&&+\tau_h^0\Big\{\|\h\eta_{(3)}\|_2^2+2(\h\theta_{(4)}-\theta_{(4)}^0)^T\h\eta_{(3)}\Big\}\Bigg]\nn\\
	&&+(\tau_w-\tau_w^0)(\h\tau_h-\tau_h^0)\Bigg[\|\h\eta_{(3)}\|_2^2-\|\h\eta_{(1)}\|_2^2\nn\\
	&&-2(\h\theta_{(1)}-\theta_{(1)}^0)^T\h\eta_{(1)}+2(\h\theta_{(4)}-\theta_{(1)}^0)^T\h\eta_{(3)}\Bigg]\nn\\
	&&-2\sum_{w=\tau^0_w+1}^{\tau_w}\Big(\sum_{h=\tau^0_h+1}^{T_h}\vep_{(w,h)}^T\eta^0_{(1)}+\sum_{h=1}^{\tau^0_h}\vep_{w,h}^T\eta^0_{(3)}\Big)\nn\\
	&&-2\sum_{w=\tau_w^0+1}^{\tau_w}\sum_{h=\tau_h^0+1}^{T_h}\vep_{(w,h)}^T(\h\eta_{(1)}-\eta_{(1)}^0)-2\sum_{w=\tau_w^0+1}^{\tau_w}\sum_{h=1}^{\tau_h^0}\vep_{(w,h)}^T(\h\eta_{(3)}-\eta^0_{(3)})\nn\\
	&&+2\sum_{w=\tau_w^0+1}^{\tau_w}\sum_{h=\tau_h^0+1}^{\h\tau_h}\vep_{(w,h)}^T\h\eta_{(1)}-2\sum_{w=\tau_w^0+1}^{\tau_w}\sum_{h=\tau_h^0+1}^{\h\tau_h}\vep_{(w,h)}^T\h\eta_{(3)}
	\eenr
	
	Finally, utilizing the $\ell_2$ lower bounds of  (\ref{eq:7}) and (\ref{eq:8}), and the upper bounds of (\ref{eq:9}) (\ref{eq:24}), (\ref{eq:29}) and (\ref{eq:26}) to the expression (\ref{eq:27}) yields that on the collection $(\tau_w-\tau_w^0)\ge c_aT_h^{-1}\phi^2\xi_w^{-2}$, we have,
	\benr\label{eq:new12}
	\cU_w(\tau_w,\h\tau_h,\h\theta)&\ge& T_h(\tau_w-\tau_w^0)\Big[\om_h\xi_1^2+(1-\om_h)\xi_3^2-c_{u1}\xi_{\min}\overline\xi-c_{u1}\xi_{\min}^2\Big]\nn\\
	&&-\frac{T_h\xi_w^2(\tau_w-\tau_w^0)}{2}-c_uc_a\xi_{\min}\si\surd(T_h(\tau_w-\tau_w^0))>0
	\eenr
	with probability at least $1-a-o(1)-\pi_T,$ uniformly over $(\tau_w-\tau_w^0)\ge c_aT_h^{-1}\phi^2\xi_w^{-2}.$ To obtain the second inequality  we have used that by definition $\om_h\xi_1^2+(1-\om_h)\xi_3^2=\xi_w^2,$ and $\xi_{\min}\le \xi_{w},$ additionally from Condition B(ii) we have $\overline\xi\le c_u\xi_{\min}$ and that the constant $c_{u1}$ arises from Condition C where it is chosen to be suitable small enough. Now Inequality (\ref{eq:new12}) directly implies that $\tilde\tau_w$ cannot be in the set $(\tau_w-\tau_w^0)\ge c_aT_h^{-1}\phi^2\xi_w^{-2},$ and must lie in its complement set, with the same probability. This completes the proof of this result.	
\end{proof}

A change of notation has been carried out for Theorem \ref{thm:wc.vanishing} and Theorem \ref{thm:wc.non.vanishing}. These are presented in more conventional {\it argmax} instead of {\it argmin} notation. Let  $\cU_w(\tau_w,\tau_h,\theta)$ and $\cU_h(\tau_w,\tau_h,\theta)$ be as in (\ref{def:cU}) and consider,
\benr\label{def:cC}
\cC_w(\tau_w,\tau_h,\theta)=-\cU_w(\tau_w,\tau_h,\theta)\quad{\rm and,}\quad\cC_h(\tau_w,\tau_h,\theta)=-\cU_h(\tau_w,\tau_h,\theta)
\eenr
Then, we can re-express the change point estimators $\tilde\tau_w(\tau_h,\theta)$ and $\tilde\tau_h(\tau_w,\theta)$ as,
\benr
\tilde\tau_w(\tau_h,\theta)=\argmax_{1\le \tau_w< T_w}\cC_w(\tau_w,\tau_h,\theta),\quad{\rm and}\quad \tilde\tau_h(\tau_w,\theta)=\argmax_{1\le \tau_h< T_h}\cC_h(\tau_w,\tau_h,\theta)\nn
\eenr

The proofs of Theorem \ref{thm:wc.vanishing} and Theorem \ref{thm:wc.non.vanishing} below are applications of the Argmax Theorem (reproduced as Theorem \ref{thm:argmax} in Appendix \ref{app:auxiliary}). The arguments here are largely an exercise in verification of requirements of this theorem.


\begin{proof}[{\bf Proof of Theorem \ref{thm:wc.vanishing}}]

	Here we require the limiting distribution of the sequence $T_h\xi_w^{2}(\tilde\tau_w-\tau^0_w),$ consequently the underlying indexing metric space here is $\R$\footnote{Although $\tilde\tau_w$ is a discrete r.v., however $T_h\xi_w^{2}\tilde\tau_w\in\R$}. Now, following is list of requirement of the Argmax theorem that require verification for this case (see, page 288 of \cite{vaart1996weak}).
	
	\begin{enumerate}
		\item The sequence $T_h\xi_w^{2}(\tilde\tau_w-\tau^0_w)$ is uniformly tight (see, Definition \ref{def:utight}).
		\item $\big\{2\si_{(w,\iny)}W_w(\z)-|\z|\}$ satisfies suitable regularity conditions\footnotemark.
		\item For any $\z\in [-c_u,c_u]$ we have
		\benr
		\cC_w(\tau_w^0+\z T_h^{-1}\xi^{-2}_w,\h\tau_h,\h\theta)
		\Rightarrow \big\{2\si_{(w,\iny)}W_w(\z)-|\z|\}.\nn
		\eenr
	\end{enumerate}
	\footnotetext{Almost all sample paths $\z\to \big\{2\si_{(w,\iny)}W_w(\z)-|\z|\}$ are upper semicontinuous and posses a unique maximum at a (random) point $\argmax_{\z\in\R}\big\{2\si_{(w,\iny)}W_w(\z)-|\z|\},$ which as a random map in the indexing metric space is tight.}
	
	Part (i) follows from the result of Theorem \ref{thm:cpoptimal} under the assumed Condition C(i)(b) and C(ii)(a,c) on the nuisance estimates $\h\tau_h$ and $\h\theta.$  The second requirement follows from well known properties of Brownian motion's. The only remaining requirement is (3). To prove this property we shall show the following two results: 
	\benr\label{eq:36}
	&(i)&	\cC_w(\tau_w^0+\z T_h^{-1}\xi^{-2}_w,\h\tau_h,\h\theta)
	\Rightarrow \big\{2\si_{(w,\iny)}W_w(\z)-|\z|\}, \quad{\textnormal{and}},\nn\\
	&(ii)& \sup_{|\tau_w-\tau_w^0|\le c_uT^{-1}_h\xi^{-2}_w}|\cC_w(\tau_w,\h\tau_h,\h\theta)-\cC_w(\tau_w,\tau_h^0,\theta^0)|=o_p(1).
	\eenr
	
	Part (ii) of (\ref{eq:36}) is established by Lemma \ref{lem:Capprox}. Part (i) of (\ref{eq:36}) is provided in the following. Begin with a couple of observations that shall be useful subsequently. For any given $w=1,...,T_w,$ define r.v.'s,
	\benr\label{eq:34a}
	\psi_w=\frac{1}{\xi_w\surd(T_h)}\Big[\sum_{h=\tau^0_h+1}^{T_h}\vep_{(w,h)}^T\eta^0_{(1)}+\sum_{h=1}^{\tau_h^0}\vep_{(w,h)}^T\eta^0_{(3)}\Big],
	\eenr
	and note that we have,
	\benr
	{\rm var}(\psi_w)=\frac{1}{\xi_w^2T_h}\Big[T_h\om_h\eta^0_{(1)}\Si\eta^0_{(1)}+T_h(1-\om_h)\eta^0_{(3)}\Si\eta^0_{(3)}\Big]\to \si^2_{(w,\iny)},\nn
	\eenr
	where the convergence follows from Condition D. Now let $\z> 0,$ w.l.o.g. assume $\z T_h^{-1}\xi^{-2}_w$ is integer valued and let $\tau^*_w=\tau^0_w+\z T_h^{-1}\xi^{-2}_w>\tau^0_w$ and consider,
	\benr\label{eq:39}
	\cC_w(\tau_w^*,\tau_h^0,\theta^0)
	&=&\sum_{w=\tau_w^0+1}^{T_w}\sum_{h=\tau^0_h+1}^{T_h}\|x_{(w,h)}-\theta_{(1)}^0\|_2^2+\sum_{w=1}^{\tau_w^0}\sum_{h=\tau_h^0+1}^{T_h}\|x_{(w,h)}-\theta_{(2)}^0\|_2^2\nn\\
	&&+\sum_{w=1}^{\tau_w^0}\sum_{h=1}^{\tau_h^0}\|x_{(w,h)}-\theta_{(3)}^0\|_2^2+\sum_{w=\tau_w^0+1}^{T_w}\sum_{h=1}^{\tau_h^0}\|x_{(w,h)}-\theta_{(4)}^0\|_2^2\nn\\
	&&-\sum_{w=\tau_w^*+1}^{T_w}\sum_{h=\tau_h^0+1}^{T_h}\|x_{(w,h)}-\theta_{(1)}^0\|_2^2-\sum_{w=1}^{\tau_w^*}\sum_{h=\tau_h^0+1}^{T_h}\|x_{(w,h)}-\theta_{(2)}^0\|_2^2\nn\\
	&&-\sum_{w=1}^{\tau_w^*}\sum_{h=1}^{\tau_h^0}\|x_{(w,h)}-\theta_{(3)}^0\|_2^2-\sum_{w=\tau_w^*+1}^{T_w}\sum_{h=1}^{\tau_h^0}\|x_{(w,h)}-\theta_{(4)}^0\|_2^2\nn\\
	&=& \sum_{w=\tau_w^0+1}^{\tau_w^*}\sum_{h=\tau^0_h+1}^{T_h}\Big[\|x_{(w,h)}-\theta_{(1)}^0\|_2^2-\|x_{(w,h)}-\theta_{(2)}^0\|_2^2\Big]\nn\\
	&&+\sum_{w=\tau_w^0+1}^{\tau_w^*}\sum_{h=1}^{\tau_h^0}\Big[\|x_{(w,h)}-\theta_{(4)}^0\|_2^2-\|x_{(w,h)}-\theta_{(3)}^0\|_2^2\Big]\nn\\
	&=&2\sum_{w=\tau_w^0+1}^{\tau_w^*}\Big[\sum_{h=\tau^0_h+1}^{T_h}\vep_{(w,h)}^T\eta^0_{(1)}+\sum_{h=1}^{\tau_h^0}\vep_{(w,h)}^T\eta^0_{(3)}\Big]-(\tau^*_w-\tau^0_w)T_h\xi_w^2\nn\\
	&=&2\xi_w \surd(T_h)\sum_{w=\tau_w^0+1}^{\tau_w^0+\z T_h^{-1}\xi_w^{-2}} \psi_w -\z\Rightarrow 2\si_{(w,\iny)}W_{w1}(\z)-\z
	\eenr
	The final equality obtained by substituting the defining expressions of $\tau^*_w=\tau^0_w+\z T_h^{-1}\xi^{-2}_w>\tau^0_w$ as well as that of $\psi_w$ from (\ref{eq:34a}). The weak convergence here now follows from the functional central limit theorem. Repeating the same argument with $\z\in[-c_u,0),$ yields  $\cC(\tau^0+\z T_h^{-1}\xi^{-2}_w,\tau_h^0,\theta^0)\Rightarrow 2\si_{(w,\iny)}W_{w2}(-\z)-|\z|.$ 
	
	This completes the proof of requirement (3) for the Argmax theorem and consequently an application of its results yields $T_h^{-1}\xi^{2}_w(\tilde\tau_w-\tau^0_w)\Rightarrow \argmax_{\z\in\R}\big\{2\si_{(w,\iny)}W_w(\z)-|\z|\},$ which completes the proof. The second limiting distribution result can be proved by proceeding with symmetrical arguments.
\end{proof}

\begin{lemma}\label{lem:Capprox} Let $\cC_w(\tau_w,\tau_h,\theta)$ and $\cC_h(\tau_w,\tau_h,\theta)$ be as defined in (\ref{def:cC}) and suppose Condition A and B hold. Additionally assume that Condition C(i)(b) and Condition C(ii)(a,c) holds with the sequence $r_T=\big\{o(1)\big/s^{1/2}\log(p\vee T_wT_h)\big\}.$ Then, for any $c_u>0,$ we have,
	\benr
	&(i)&\,\,\sup_{|\tau_w-\tau_w^0|\le c_uT^{-1}_h\xi^{-2}_w}\big|\cC_w(\tau_w,\h\tau_h,\h\theta)-\cC_w(\tau_w,\tau_h^0,\theta^0)\big|=o_p(1),\quad{\rm and}\nn\\
	&(ii)&\,\,\sup_{|\tau_h-\tau_h^0|\le c_uT^{-1}_w\xi^{-2}_h}\big|\cC_h(\h\tau_w,\tau_h,\h\theta)-\cC_w(\tau_w^0,\tau_h,\theta^0)\big|=o_p(1),\nn
	\eenr
	where the orders of (i) and (ii) are w.r.t. $T_w$ and $T_h,$ respectively.
\end{lemma}

\begin{proof}[{\bf Proof of Lemma \ref{lem:Capprox}}] We only prove Part (i) below, the proof of Part (ii) follows symmetrically. This proof relies on an  algebraic decomposition of the difference of interest and rates of residual terms. To this end, we begin with a few bounds and residual terms that shall be required. Observing that from Condition C(i)(b) with the assumed choice of $r_T,$ we have that,
	\benr\label{eq:43}
	|\h\tau_h-\tau_h^0|\le T_hr_T^2,\quad{\rm with}\quad r_T=\frac{o(1)}{s^{1/2}\log(p\vee T_wT_h)},
	\eenr
	with probability at least $1-\pi_T.$ Also, by proceeding similar to (\ref{eq:4}) and (\ref{eq:8}), we have under Condition C(ii)(a,c) with the assumed choice of $r_T,$ that,
	\benr\label{eq:44}
	\max_{1\le j\le 4}\|\h\theta_{(j)}-\theta_{(j)}^0\|_1\le \max_{1\le j\le 4}\surd{s}\|\h\theta_{(j)}-\theta_{(j)}^0\|_2\le \frac{o(1)\xi_{\min}}{\log (p\vee T)}.
	\eenr
	with probability at least $1-\pi_T.$ Consequently, we also have,
	\benr\label{eq:47}
	\|\h\eta_{(1)}-\eta^0_{(1)}\|_1\le \frac{o(1)\xi_{\min}}{\log (p\vee T)}\,\,\,{\rm and}\,\,\, \|\h\eta_{(3)}-\eta^0_{(3)}\|_1\le \frac{o(1)\xi_{\min}}{\log (p\vee T_wT_h)}.
	\eenr
	with probability at least $1-\pi_T.$ Finally, following (\ref{eq:5}) and (\ref{eq:8}) we have that,
	\benr\label{eq:48}
	\|\h\eta_{(1)}\|_1\le c_u\surd{s}\xi_{\min}\quad{\rm and} \|\h\eta_{(3)}\|_1\le c_u\surd{s}\xi_{\min}
	\eenr
	with probability at least $1-\pi_T.$  Now consider the case where $\tau_w\ge \tau^0_w$ and  $\h\tau_h\ge\tau_h^0,$ and define the following residual terms,
	\benr
	R1&=&2\sum_{w=\tau_w^0+1}^{\tau_w}\sum_{h=\tau_h^0+1}^{T_h}\vep_{(w,h)}^T(\h\eta_{(1)}-\eta_{(1)}^0)\nn\\
	R2&=&2\sum_{w=\tau_w^0+1}^{\tau_w}\sum_{h=1}^{\tau_h^0}\vep_{(w,h)}^T(\h\eta_{(3)}-\eta^0_{(3)})\nn\\
	R3&=&2\sum_{w=\tau_w^0+1}^{\tau_w}\sum_{h=\tau_h^0+1}^{\h\tau_h}\vep_{(w,h)}^T\h\eta_{(1)}-2\sum_{w=\tau_w^0+1}^{\tau_w}\sum_{h=\tau_h^0+1}^{\h\tau_h}\vep_{(w,h)}^T\h\eta_{(3)}\nn\\
	R4&=&\om_h\big(\|\h\eta_{(1)}\|_2^2-\|\eta_{(1)}^0\|_2^2\big)+(1-\om_h)\big(\|\h\eta_{(3)}\|_2^2-\|\eta_{(3)}^0\|_2^2\big)\nn\\
	R5&=&2\om_h(\h\theta_{(1)}-\theta_{(1)}^0)^T\h\eta_{(1)}+2(1-\om_h)(\h\theta_{(4)}-\theta_{(4)}^0)^T\h\eta_{(3)}\nn\\
	&&-2\frac{(\h\tau_h-\tau_h^0)}{T_h}(\h\theta_{(1)}-\theta_{(1)}^0)^T\h\eta_{(1)}
	+2\frac{(\h\tau_h-\tau_h^0)}{T_h}(\h\theta_{(4)}-\theta_{(1)}^0)^T\h\eta_{(3)}\nn\\
	&&+\frac{(\h\tau_h-\tau_h^0)}{T_h}\|\h\eta_{(3)}\|_2^2-\frac{(\h\tau_h-\tau_h^0)}{T_h}\|\h\eta_{(1)}\|_2^2\nn
	\eenr
	Then under the considered orientation $\tau_w\ge \tau^0_w$ and  $\h\tau_h\ge\tau_h^0,$ we have the following algebraic expansion,	
	\benr\label{eq:46}
	\cC_w(\tau_w,\h\tau_h,\h\theta)-\cC_w(\tau_w,\tau_h^0,\theta^0)&=&\cU_w(\tau_w,\tau_h^0,\theta^0)-\cU_w(\tau_w,\h\tau_h,\h\theta)\nn\\
	&=&-R1+R2-R3-(\tau_w-\tau_w^0)T_h (R4+R5)
	\eenr
	We now examine each of the residual terms in (\ref{eq:46}) individually. Applying Lemma \ref{lem:nearoptimalcross.subE.special} we have that,
	\benr
	\sup_{(\tau_w-\tau_w^0)\le c_uT^{-1}_h\xi^{-2}_w}|R1|\le \frac{c_u\si}{\xi_w}\log(p\vee T_wT_h)\|\h\eta_{(1)}-\eta^0_{(1)}\|_1=o(1)\nn
	\eenr
	w.p. $1-o(1).$ Here the equality follows from (\ref{eq:47}). An analogous argument yields,
	\benr
	\sup_{(\tau_w-\tau_w^0)\le c_uT^{-1}_h\xi^{-2}_w}|R2|=o_p(1).\nn
	\eenr
	Applying Lemma \ref{lem:nearoptimalcross.subE.special} together with the bounds of (\ref{eq:43}) and (\ref{eq:47}) yields,
	\benr
	\sup_{(\tau_w-\tau_w^0)\le c_uT^{-1}_h\xi^{-2}_w}|R3|\le \frac{c_u\si\xi_{\min}\surd{s}}{\xi_w}r_T\log(p\vee T_wT_h)=o(1)
	\eenr
	with probability $1-o(1).$ To bound the remaining two terms $R4$ and $R5,$  consider,
	\benr
	\sup_{(\tau_w-\tau_w^0)\le c_uT^{-1}_h\xi^{-2}_w}|(\tau_w-\tau_w^0)T_hR4|\le c_u\xi_{w}^{-2}|\big(\|\h\eta_{(1)}\|_2^2-\|\eta_{(1)}^0\|_2^2|+\|\h\eta_{(3)}\|_2^2-\|\eta_{(3)}^0\|_2^2\big)\nn\\
	\le  c_u\xi^{-2}_w\big|\|\h\eta_{(1)}-\eta^0_{(1)}\|_2^2+2(\h\eta_{(1)}-\eta^0_{(1)})^T\eta^0_{(1)}\big|\hspace{2cm}\nn\\
	+c_u\xi^{-2}_w\big|\|\h\eta_{(3)}-\eta^0_{(3)}\|_2^2+2(\h\eta_{(3)}-\eta^0_{(3)})^T\eta^0_{(3)}\big|\hspace{1.75cm}\nn\\
	\le  c_u\xi^{-2}_w \|\h\eta_{(1)}-\eta^0_{(1)}\|_2^2+ c_u\xi^{-1}_w \|\h\eta_{(1)}-\eta^0_{(1)}\|_2\hspace{2cm}\nn\\
	+ c_u\xi^{-2}_w \|\h\eta_{(3)}-\eta^0_{(3)}\|_2^2+ c_u\xi^{-1}_w \|\h\eta_{(3)}-\eta^0_{(3)}\|_2
	= o_p(1).\hspace{0.35cm}\nn
	\eenr
	The third inequality follows from the Cauchy Schwarz inequality and the equality follows from the bounds in (\ref{eq:47}). The only remaining residual term now is $R5$ which is examined below.
	\benr
	\sup_{(\tau_w-\tau_w^0)\le c_uT^{-1}_h\xi^{-2}_w}|(\tau_w-\tau_w^0)T_hR5|\le  c_u\xi_w^{-2}\Big[\|\h\theta_{(1)}-\theta_{(1)}^0\|_2\|\h\eta_{(1)}\|_2+\|\h\theta_{(4)}-\theta_{(4)}^0\|_2\|\h\eta_{(3)}\|_2\Big]\nn\\
	+c_u\xi_w^{-2}r_T^2\|\h\theta_{(1)}-\theta_{(1)}^0\|_2\|\h\eta_{(1)}\|_2
	+c_u\xi_w^{-2}r_T^2\|\h\theta_{(4)}-\theta_{(1)}^0\|_2\|\h\eta_{(3)}\|_2\nn\\
	+c_ur_T^2\|\h\eta_{(3)}\|_2^2+c_ur_T^2\|\h\eta_{(1)}\|_2^2=o_p(1)\hspace{4.15cm}\nn
	\eenr
	The inequality here follows from several applications of the Cauchy Schwarz inequality and the equality follows by substituting the choice of $r_T$ from (\ref{eq:43}), as well as the available bounds for the mean estimates. Substituting the uniform bounds for $R1,R2,R3,R4$ and $R5$ obtained above into the expression (\ref{eq:46}) yields,
	\benr
	\sup_{(\tau_w-\tau_w^0)\le c_uT^{-1}_h\xi^{-2}_w}\Big|\cC_w(\tau_w,\h\tau_h,\h\theta)-\cC_w(\tau_w,\tau_h^0,\theta^0)\Big|=o_p(1)
	\eenr
	Repeating symmetrical arguments on the remaining three orientations of the ordering of $(\tau_w,\h\tau_h)$ w.r.t $(\tau_w^0,\tau_h^0),$ in particular, $\tau_w\le\tau_w^0,$ $\h\tau_h\ge\tau_h^0,$ and $\tau_w\le\tau_w^0,$ $\h\tau_h\le\tau_h^0,$ and $\tau_w\ge\tau_w^0,$ $\h\tau_h\le\tau_h^0,$ shall yield the same $o_p(1)$ approximation. This completes the proof of the lemma.
\end{proof}

\begin{proof}[{\bf Proof of Theorem \ref{thm:wc.non.vanishing}}] The proof of this theorem follows a similar structure as that of Theorem \ref{thm:wc.vanishing} in that it is also an application of the Argmax theorem. The distinction here is in the limiting distributional structure that is induced by the change of regime of the jump size.
	
	The requirements to be verified here are as follows.
	\begin{enumerate}
		\item The sequence $(\tilde\tau_w-\tau^0_w)$ is uniformly tight.
		\item $\cC_{(w,\iny)}(\z)$ satisfies suitable regularity conditions.
		\item For any $\z\in \{-c_u,-c_u+1,...,-1,0,1,...c_u\},$ we have
		\benr
		\cC_w(\tau_w^0+\z,\h\tau_h,\h\theta)
		\Rightarrow \cC_{(w,\iny)}(\z).\nn
		\eenr
	\end{enumerate}
	As in the proof of Theorem \ref{thm:wc.vanishing}, requirement (1) follows directly from the result of Theorem \ref{thm:cpoptimal}.  Requirement (2) of regularity of the {\it argmax} of two sided negative drift random walk $\cC_{(w,\iny)}(\z)$ has been proved earlier in Lemma A.3 of the supplement of \cite{Kaul2021single}.  As before, the requirement (3) is verified by the following two results:
	\benr\label{eq:36nv}
	&(i)&	\cC_w(\tau_w^0+\z,\tau_h^0,\h\theta^0)
	\Rightarrow \cC_{(w,\iny)}(\z)., \quad{\textnormal{and}},\nn\\
	&(ii)& \sup_{|\tau_w-\tau_w^0|\le c_u}|\cC_w(\tau_w,\h\tau_h,\h\theta)-\cC_w(\tau_w,\tau_h^0,\theta^0)|=o_p(1).
	\eenr
	
	Part (ii) of (\ref{eq:36nv}) is established by Lemma \ref{lem:Capprox}. The following establishes Part (i) of (\ref{eq:36nv}). Let $\psi_w$ be as defined in (\ref{eq:34a}), then we begin by noting that under this non-vanishing regime $\surd(T_h)\xi_w\to\xi_{(w,\iny)},$ we have,
	\benr
	{\rm var}(\psi_w)=\frac{1}{\xi_w^2T_h}\Big[T_h\om_h\eta^0_{(1)}\Si\eta^0_{(1)}+T_h(1-\om_h)\eta^0_{(3)}\Si\eta^0_{(3)}\Big]\to \xi_{(w,\iny)}\si^2_{(w,\iny)},\nn
	\eenr
	where the convergence follows from Condition D and the regime under consideration. Now for any $\z\in\{1,2,...,c_u\},$  let $\tau^*_w=\tau^0_w+\z T_h^{-1}\xi^{-2}_w>\tau^0_w$ and note that,
	\benr\label{eq:new2} 
	\cC_w(\tau_w^*,\tau_h^0,\theta^0)
	&=&2\sum_{w=\tau_w^0+1}^{\tau_w^*}\Big[\sum_{h=\tau^0_h+1}^{T_h}\vep_{(w,h)}^T\eta^0_{(1)}+\sum_{h=1}^{\tau_h^0}\vep_{(w,h)}^T\eta^0_{(3)}\Big]-(\tau^*_w-\tau^0_w)T_h\xi_w^2\nn\\
	&=&2\sum_{w=\tau_w^0+1}^{\tau_w^0+\z} \psi_w -\z\xi_w^2\Rightarrow \sum_{w=1}^{\z}\cP\big(-\xi_{(w,\iny)}^2,\,\,4\xi_{(w,\iny)}^2\si^2_{(w,\iny)}\big).
	\eenr
	The equalities here follow by performing a algebraic decomposition as provided in (\ref{eq:39}). The weak convergence follows from Condition A$'$. Repeating the same argument with $\z\in\{-c_u,-c_u+1,...,-1\},$ yields  	$	\cC_w(\tau_w^0+\z,\tau_h^0,\theta^0)\Rightarrow \sum_{t=1}^{-\z}\cP(-\xi_{(w,\iny)}^2,4\xi_{(w,\iny)}^2\si^2_{(w,\iny)}).$  The Argmax theorem now yields $(\tilde\tau^*_w-\tau^0_w)\Rightarrow \argmax_{\z\in\Z}\cC_{(w,\iny)}(\z),$ which completes the proof.

\end{proof}

\begin{theorem}\label{thm:unifmean} 	Suppose Condition A and B holds and that $c_uT_wT_h\underline\om\ge \log (p\vee T_wT_h).$ Let $\psi=\max_j\|\eta^0_{(j)}\|_{\iny}$ and for any $\tau=(\tau_w,\tau_h)$ and each $j=1,2,3,4,$ let,
	\benr\label{eq:la}
	\la:=\la_j=8\max\Big\{\si\Big\{\frac{2c_{u1}\log(p\vee T_wT_h)}{c_uT_wT_h\underline\om}\Big\}^{\frac{1}{2}},\,\,\frac{3\psi}{c_u\underline\om}(u_{T_w}\vee u_{T_h})\Big\}.
	\eenr
	Where $|\tau_w-\tau_w^0|\le Tu_{T_w}$ and $|\tau_h-\tau_h^0|\le Tu_{T_{h}}$ and $c_{u1}>0$ is a constant. Then, $\h\theta_{(j)}(\tau),$ of (\ref{est:softthresh}) satisfies the following two results with probability at least $1-\pi_T.$\\~
	(i) For $j=1,2,3,4,$ with $|Q_j(\tau)|\ge c_uT_wT_h\underline\om,$ we have $\big\|\big(\h\theta_{(j)}(\tau)\big)_{S_j^c}\big\|_1\le 3\big\|\big(\h\theta_{(j)}(\tau)-\theta_{(j)}^0\big)_{S_j}\big\|_1.$\\~
	(ii) The following bound is satisfied,
	\benr
	\max_{1\le j\le 4} \|\h\theta_{(j)}(\tau)-\theta_{(j)}^0\|_2\le 6\surd{s}\la,\nn
	\eenr
	uniformly for all possible choices of $\tau=(\tau_w,\tau_h).$ Here  $\pi_T=8\exp\big\{-\big(c_{u2}-2\big)\log (p\vee T_wT_h)\big\},$ where $c_{u2}=c_{u1}\wedge \surd(c_uc_{u1}/2).$	
\end{theorem}

\begin{proof}[{\bf Proof of Theorem \ref{thm:unifmean}}]
	Consider $j=1$ and any $(\tau_w,\tau_h)^T$ such that $|Q_1(\tau)|\ge c_uT_wT_h\underline\om.$ Without loss of generality assume $\tau_w\le\tau^0_w,$ $\tau_h\le \tau_h^0$ The remaining permutations of the ordering of $\tau$ w.r.t. $\tau^0$ can be proved using symmetrical arguments.
	
	An algebraic rearrangement of the elementary inequality $\big\|\bar x_{(1)}(\tau)-\h\theta_{(1)}(\tau)\big\|^2_2+\la_1\|\h\theta_{(1)}(\tau)\|_1\le \big\|\bar x_{(1)}(\tau)-\theta_{(1)}^0\big\|^2_2+\la_1\|\theta_{(1)}^0\|_1$ yields,
	\benr\label{eq:20}
	\big\|\h\theta_{(1)}(\tau)-\theta_{(1)}^0\big\|_2^2+\la_1\big\|\tilde\theta_{(1)}(\tau)\big\|_1&\le& \la_1\big\|\theta^0_{(1)}\big\|_1\nn\\
	&&+ \frac{2}{|Q_1(\tau)|}\sum_{w=\tau_w+1}^{T_w}\sum_{h=\tau_h+1}^{T_h}\tilde\vep_{(w,h)}^T(\h\theta_{(1)}(\tau)-\theta_{(1)}^0).\nn\\
	&=&\la_1\big\|\theta^0_{(1)}\big\|_1+\frac{2}{|Q_1(\tau)|}\sum_{w=\tau_w+1}^{T_w}\sum_{h=\tau_h+1}^{T_h}\vep_{(w,h)}^T(\h\theta_{(1)}(\tau)-\theta_{(1)}^0)\nn\\
	&&-\frac{2}{|Q_1(\tau)|}(\tau_w^0-\tau_w)(\tau_h^0-\tau_h)(\theta_{(1)}^0-\theta_{(3)}^0)^T(\h\theta_{(1)}(\tau)-\theta_{(1)}^0)\nn\\
	&&-\frac{2}{|Q_1(\tau)|}(\tau_w^0-\tau_w)(T_h-\tau_h^0)(\theta_{(1)}^0-\theta_{(2)}^0)^T(\h\theta_{(1)}(\tau)-\theta_{(1)}^0)\nn\\
	&&-\frac{2}{|Q_1(\tau)|}(T_w-\tau_w^0)(\tau_h^0-\tau_h)(\theta_{(1)}^0-\theta_{(3)}^0)^T(\h\theta_{(1)}(\tau)-\theta_{(1)}^0).\nn\\
	&=&\la_1\big\|\theta^0_{(1)}\big\|_1+2R1-2R2-2R3-2R4
	\eenr
	Here $\tilde\vep_{(w,h)}=\big(x_{(w,h)}-\h\theta_{(1)}(\tau)\big).$ Next we consider the residual terms $R1,R2,R3,R4$ on the r.h.s of (\ref{eq:20}). For this purpose, first note from Lemma \ref{lem:1.to.tau.bound} we have,
	\benr
	\frac{2}{|Q_1(\tau)|}\Big\|\sum_{w=\tau_w+1}^{T_w}\sum_{h=\tau_h+1}^{T_h}\vep_{(w,h)}\Big\|_{\iny}\le 2\si\Big\{\frac{2c_{u1}\log(p\vee T_wT_h)}{c_uT_wT_h\underline\om}\Big\}^{\frac{1}{2}},
	\eenr
	with probability at least $1-8\exp\{-(c_{u2}-2)\log(p\vee T_wT_h)\}.$ Additionally recall we have $\psi=\max_j\|\eta^0_{(j)}\|,$ and $|Q_1(\tau)|\ge c_uT_wT_h\underline\om$ thus,
	\benr
	2|R2+R3+R4|\le \frac{6\psi}{c_u\underline\om}(u_{T_w}\vee u_{T_h})\|\h\theta_{(1)}(\tau)-\theta_{(1)}^0\|_1.
	\eenr
	Consequently, upon choosing,
	\benr
	\la^*=\max\Big\{4\si\Big\{\frac{2c_{u1}\log(p\vee T_wT_h)}{c_uT_wT_h\underline\om}\Big\}^{\frac{1}{2}},\,\,\frac{12\psi}{c_u\underline\om}(u_{T_w}\vee u_{T_h})\Big\},\nn
	\eenr
	and substituting these bounds in (\ref{eq:20}) we obtain,
	\benr\label{eq:21}
	\big\|\h\theta_{(1)}(\tau)-\theta_{(1)}^0\big\|_2^2+\la_1\big\|\h\theta_{(1)}(\tau)\big\|_1\le\la_1\big\|\theta^0_{(1)}\big\|_1+\la^*\big\|\h\theta_{(1)}(\tau)-\theta_{(1)}^0\big\|_1,
	\eenr
	with probability at least $1-8\exp\{-(c_{u2}-2)\log(p\vee T_wT_h)\}.$  Now choosing $\la_1= 2\la^*,$ leads to $\|\big(\h\theta_{(1)}(\tau)\big)_{S_1^c}\|_1\le 3\|\big(\h\theta_{(1)}(\tau)-\theta_{(1)}^0\big)_{S_1}\|_1,$ which proves part (i) of this theorem for $j=1.$ From inequality (\ref{eq:21}) we also have that,
	\benr
	\|\h\theta_{(1)}(\tau)-\theta_{(1)}^0\|_2^2\le \frac{3}{2}\la_1\|\h\theta_{(1)}(\tau)-\theta_{(1)}^0\|_1\le 6\la_1\surd{s}\|\h\theta_{(1)}(\tau)-\theta_{(1)}^0\|_2
	\eenr
	This directly implies that  $\|\h\theta_{(1)}(\tau)-\theta_{(1)}^0\|_2\le 6\la_1\surd{s},$ where we have used $\|\h\theta_{(1)}(\tau)-\theta_{(1)}^0\|_1\le 4\sqrt{s}\|\h\theta_{(1)}(\tau)-\theta_{(1)}^0\|_2,$ which follows in turn from $\|\big(\h\theta_{(1)}(\tau)\big)_{S_1^c}\|_1\le 3\|\big(\h\theta_{(1)}(\tau)-\theta_{(1)}^0\big)_{S_1}\|_1.$ To supply uniformity over $\tau,$ recall that the only stochastic bound used here is Lemma \ref{lem:1.to.tau.bound} which holds uniformly over $\tau,$ consequently the final bound also holds uniformly over the given collection. A symmetrical argument can be replicated for each $j=2,3,4$ and recalling that Lemma \ref{lem:1.to.tau.bound} also holds uniformly over these $j$'s. This finishes the proof of the Theorem. This result can alternatively be proved using the properties of the soft-thresholding operator $k_{\la}(\cdotp),$ by building uniform versions of arguments such as those in \cite{kaul2017structural}.
\end{proof}

\begin{lemma}\label{lem:step1mean} Assume Condition A, B and F holds and that the model dimensions together with the least jump size are restricted by the following condition,
	\benr\label{eq:23}
	\frac{c_u\si}{\xi_{\min}}\Big\{\frac{s\log (p\vee T_wT_h)}{T_wT_h\underline\om}\Big\}^{\frac{1}{2}}\le c_{u1},
	\eenr	
	for an appropriately chosen small enough constant $c_{u1}>0.$ Additionally assume that  $c_uT_wT_h\underline\om\ge \log (p\vee T_wT_h).$ Then with a suitably chosen regularizer $\la,$ the Step 1 mean estimates of Algorithm 1, $\check\theta_{(j)}=\h\theta_{(j)}(\check\tau),$ $j=1,2,3,4,$ satisfy the following, w.p. $1-o(1).$\\~
	(i)  $\big\|\big(\check\theta_{(j)}\big)_{S_j^c}\big\|_1\le 3\big\|\big(\check\theta_{(j)}-\theta_{(j)}^0\big)_{S_j}\big\|_1,$ for any $j=1,2,3,4.$ \\~
	(ii) The following bound is satisfied,
	\benr
	\max_{1\le j\le 4}\|\check\theta_{(1)}-\theta_{(1)}^0\|_2\le c_{u1}\xi_{\min}.\nn
	\eenr
	Consequently, the mean estimates $\check\theta_{(j)},$ $j=1,2,3,4$ satisfy Condition C(ii)(a,b).
\end{lemma}

\begin{proof}[{\bf Proof of Lemma \ref{lem:step1mean}}] From Condition F, the initializer  $\check\tau=(\check\tau_w,\check\tau_h)^T$ 	is assumed to satisfy, $(i)\,\,|\check\tau_w-\tau^0_w|\le T_wu_{T_w},$ $(ii)\,\,|\check\tau_h-\tau_w^0|\le T_hu_{T_h}$ and $(iii)\,\,\min_{1\le j\le 4}|Q_j(\check\tau)|\ge c_{u}T_wT_h\underline\om,$ where,
	\benr\label{eq:28}
	u_{T_w}=u_{T_h}=c_{u1}\underline\om\xi_{\min}\big/(\surd s\psi).
	\eenr	
	Now applying Theorem \ref{thm:unifmean} while choosing,
	\benr\label{eq:la.step1.choice}
	\la\,\,{\textrm {as prescribed in (\ref{eq:la}) with}} \,\,u_{T_w},\,u_{T_h}\,\,{\textrm {as given in (\ref{eq:28})}},
	\eenr
	we obtain the following two results that hold with probability $1-o(1).$  First,  $\big\|\big(\check\theta_{(j)}\big)_{S_j^c}\big\|_1\le 3\big\|\big(\check\theta_{(j)}-\theta_{(j)}^0\big)_{S_j}\big\|_1,$ for any $j=1,2,3,4.$ Second,
	\benr\label{eq:22}
	\max_{1\le j\le 4}\|\check\theta_{(j)}-\theta_{(j)}^0\|_2&\le& \max\Big[c_u\si\Big\{\frac{s\log(p\vee T_wT_h)}{T_wT_h\underline\om}\Big\}^{\frac{1}{2}},\,\,c_u\frac{(u_{T_w}\vee u_{T_h})\surd{s}\psi}{\underline\om}\Big],\nn\\
	&=& \xi_{\min}\max\Big[\frac{c_u\si}{\xi_{\min}}\Big\{\frac{s\log(p\vee T_wT_h)}{T_wT_h\underline\om}\Big\}^{\frac{1}{2}},\,\,c_u\frac{(u_{T_w}\vee u_{T_h})\surd{s}\psi}{\xi_{\min}\underline\om}\Big]\nn\\
	&=& \xi_{\min}\Big[R_1,R_2\Big]
	\eenr	
	Here the first equality is simply an algebraic manipulation. Now for a suitable chosen $c_{u1}>0,$ we have from assumption (\ref{eq:23}) that,
	\benr
	\frac{c_u\si}{\xi_{\min}}\Big\{\frac{s\log (p\vee T_wT_h)}{T_wT_h\underline\om}\Big\}^{\frac{1}{2}}\le c_{u1},\nn
	\eenr
	which provides a bound for term $R_1$ on the RHS of (\ref{eq:22}). Next we bound term $R_2$ of the same expression. Substituting the choice of $u_T$ from (\ref{eq:28}) in term $R_2$, together with the earlier bound for $R_1,$ we obtain,
	\benr\max_{1\le j\le 4}\|\check\theta_{(1)}-\theta_{(1)}^0\|_2\le
	\xi_{\min}\Big[R_1,R_2\Big]\le c_{u1}\xi_{\min},\nn
	\eenr
	with probability $1-o(1).$ Thereby the requirement Condition C(ii)(a,b) are met and this completes the proof of the lemma.
\end{proof}

\begin{proof}[{\bf Proof of Corollary \ref{cor:alg1.validity}}] The argument to follow is described visually in Figure \ref{fig:schematic}. We show that once Algorithm \ref{alg:single} is initialized under Condition F, then, under the assumed rate conditions on model parameters all remaining conditions fall in line for Step 1 and Step 2, thereby allowing applicability of the main results of Subsection \ref{subsec:plugin}.
	
	Lemma \ref{lem:step1mean} establishes that $\check\theta_{(j)}=\tilde\theta_{(j)}(\check\tau),$ $j=1,2,3,4,$ of Step 1 of Algorithm \ref{alg:single} satisfies Condition C(ii)(a,c), under the dimensional rate assumption (\ref{eq:23}), which is weaker than Condition E, and therefore this result continues to hold. Also observe that Condition C(i)(a) is weaker than assumed Condition F on initializer $\check\tau.$ Thus, a direct application of Part (1a) and (1b) of Theorem \ref{thm:cpoptimal} yields the rate of convergence of $\h\tau$ of Step 1 of Algorithm \ref{alg:single} as,
	\benr\label{eq:49}
	&&|\h\tau_w-\tau^0_w|\le c_{u}\si^2T_h^{-1}\xi^{-2}_ws\log^2(p\vee T_wT_h),\quad{\rm and}\nn\\
	&&|\h\tau_h-\tau^0_h|\le c_{u}\si^2T_w^{-1}\xi^{-2}_hs\log^2(p\vee T_wT_h),
	\eenr	
	with probability at least $1-o(1).$ This completes the proof of Part (a) of this Corollary.

	Next observe the bounds (\ref{eq:49}) together with rate assumption of  Condition E implies that $\h\tau=(\h\tau_w,\h\tau_h)^T$ of Step 1 satisfies the stronger Condition C(i)(b). Next we show that the updated mean estimates $\h\theta_{(j)},$ $j=1,2,3,4,$ satisfy Condition C(ii)(a,c). For this purpose, note that from (\ref{eq:49}) we have that $|\h\tau_w-\tau_w^0|\le Tu_{T_w},$ and $|\h\tau_h-\tau_h^0|\le Tu_{T_h}$ with probability $1-o(1),$ where,
	\benr\label{eq:ut}
	(u_{T_w}\vee u_{T_h})\le c_{u}\si^2T_{w}^{-1}T_h^{-1}\xi^{-2}_{\min}s\log^2(p\vee T_wT_h),
	\eenr
	Moreover, (\ref{eq:49}) and Condition E also imply that with the same probability as above, we have $|Q_j(\h\tau)|\ge c_uT_wT_h\underline\om.$ Now applying Theorem \ref{thm:unifmean} with,
	\benr\label{eq:la.step2.choice}
	\la\,\,{\rm as\, prescribed\,\, in\,\, (\ref{eq:la})\,\, with}\,\,(u_{T_w}\vee u_{T_h})\,\,{\rm as\,\,in\,\,(\ref{eq:ut})},
	\eenr
	yields $\h\theta_{(j)}=\tilde\theta_{(j)}(\h\tau),$ $j=1,...,4,$  of Step 2 of Algorithm \ref{alg:single} satisfies Condition C(ii)(a). Further,
	\benr\label{eq:25a}
	\max_{1\le j\le 4}\big\|\h\theta_{(j)}-\theta_{(j)}^0\big\|_2&\le& \max\Big[\si\Big\{\frac{c_{u}s\log(p\vee T_wT_h)}{T_wT_h\underline\om}\Big\}^{\frac{1}{2}},\,\,\frac{c_u\surd{s}\psi}{\underline\om}(u_{T_w}\vee u_{T_h})\Big].\nn\\
	&=&\frac{\xi_{\min}}{s^{1/2}\log (p\vee T_wT_h)} \max\Big[\si\Big\{\frac{c_{u}s\log^{3/2}(p\vee T_wT_h)}{\xi_{\min}\surd{(T_wT_h\underline\om)}}\Big\},\,\nn\\
	&&\hspace{4cm}\,\frac{c_us\log (p\vee T_wT_h)\psi}{\underline\om\xi_{\min}}(u_{T_w}\vee u_{T_h})\Big]\nn\\
	&=&\frac{\xi_{\min}}{s^{1/2}\log (p\vee T_wT_h)}\max\big[R1,R2\big]
	\eenr
	with probability at $1-o(1).$ The first equality is simply an algebraic manipulation. From Condition E we have that $R1\le c_{u1},$ where $c_{u1}>0,$ is an appropriately chosen small constant. Next consider term $R2$ of (\ref{eq:25a}). Substituting $(u_{T_w}\vee u_{T_h})$ from (\ref{eq:ut}) in term $R2$ we obtain,
	\benr
	c_us\log (p\vee T_wT_h)\frac{\psi}{\underline\om\xi_{\min}}(u_{T_w}\vee u_{T_h})&\le &c_{u}\si^2\Big(\frac{\psi}{\xi_{\min}}\Big)\Big\{\frac{s^2\log^3(p\vee T_wT_h)}{\xi^{2}_{\min}T_{w}T_h\underline\om}\Big\},\nn\\
	&\le& c_u\Big\{\Big(\frac{\si}{\xi_{\min}}\Big)\frac{s\log^{2} (p\vee T_wT_h)}{\surd(T_wT_h\om)}\Big\}^2
	\le c_{u1}.\nn
	\eenr
	The second inequality follows from the assumption $(\psi\big/\xi)\le \log(p\vee T_wT_h).$ The third inequality follows from Condition E. Substituting the bounds for terms $R1$ and $R2$ back in (\ref{eq:25a}) yields,
	\benr\label{eq:25aa}
	\max_{1\le j\le 4}\big\|\h\theta_{(j)}-\theta_{(j)}^0\big\|_2\le
	\frac{c_{u1}\xi_{\min}}{s^{1/2}\log (p\vee T_wT_h)}
	\eenr
	with probability at $1-o(1).$ Thus, estimates $\h\theta_1$ and $\h\theta_2$ of Step 2 of Algorithm 1 satisfy requirement of Condition C(ii)(a,c). We now appeal to Theorem \ref{thm:cpoptimal}, Theorem \ref{thm:wc.vanishing} and Theorem \ref{thm:wc.vanishing} which yields Part (b) and Part (c) of this Corollary. This completes the proof.
\end{proof}


\section{Deviation bounds}\label{app:deviation}
\begin{lemma}\label{lem:nearoptimalcross.subE.special} Suppose Condition A and B(i) holds. Then, we have,
	\benr
	\Big\|\sum_{w=\tau^0_w+1}^{\tau_w}\sum_{h=\tau_h^0+1}^{\tau_h}\vep_{(w,h)}\Big\|_{\iny}\le 2c_u\si\log (p\vee T_wT_h)\surd \big((\tau_w-\tau_w^0)(\tau_h-\tau_h^0)\big)\nn
	\eenr
	with probability at least $1-2\exp\{-(c_u-2)\log(p\vee T_wT_h)\},$ uniformly over all possible $(\tau_w,\tau_h).$	
\end{lemma}

\begin{proof}[Proof of Lemma \ref{lem:nearoptimalcross.subE.special}] Consider any $k\in\{1,2,...,p\}$ and any $\tau_w>\tau^0_w,$ $\tau_h>\tau_h^0$ and apply the Bernstein's inequality (Lemma \ref{lem:bernstein}) for any $d>0$ to obtain,
	\benr\label{eq:6}
	pr\Big(\big|\sum_{w=\tau^0_w+1}^{\tau_w}\sum_{h=\tau_h^0+1}^{\tau_h}\vep_{(w,h,k)}\big|>d(\tau_w-\tau^0_w)(\tau_h-\tau^0_h)\Big)\le\hspace{2cm}\nn\\
	2\exp\Big\{-\frac{1}{2}(\tau_w-\tau^0_w)(\tau_h-\tau^0_h)\Big(\frac{d^2}{\si^2}\wedge\frac{d}{\si}\Big)\Big\}.
	\eenr	
	Choose $d=2c_u\si\{\log^2 (p\vee T_wT_h)\big/\big((\tau_w-\tau^0_w)(\tau_h-\tau_h^0)\big)\}^{1/2},$ and note that,
	\benr\label{eq:1}
	(\tau_w-\tau^0_w)(\tau_h-\tau_h^0)\frac{d^2}{2\si^2}&=&2c_u^2\log^2(p\vee T_wT_h),\quad {\rm and},\nn\\
	(\tau_w-\tau^0_w)(\tau_h-\tau_h^0)\frac{d}{2\si}&\ge& c_u\log(p\vee T_wT_h),
	\eenr
	where we have used $(\tau_w-\tau^0_w)\ge \ge 1,$ $(\tau_h-\tau^0_h)\ge 1.$ Thus, substituting this choice of $d$ in (\ref{eq:6}) and recalling that by choice $c_u\ge 1$, we obtain,
	\benr
	\Big|\sum_{w=\tau^0_w+1}^{\tau_w}\sum_{h=\tau_h^0+1}^{\tau_h}\vep_{(w,h,k)}\Big|&\le& 2c_u\si(\tau_w-\tau^0_w)^{1/2}(\tau_h-\tau^0_h)^{1/2}\{\log^2(p\vee T_wT_h)\}^{1/2}\nn
	\eenr
	with probability at least $1-2\exp\{-c_u\log (p\vee T_wT_h)\}.$ Uniformity is supplied by applying a union bound over $k=1,...,p,$ and at most $T_wT_h$ choices of $\tau_w=1,...,T_w,$ and $\tau_h=1,...T_h.$
\end{proof}

\begin{lemma}\label{lem:optimalcross} Suppose Condition A and B(i) hold. Then for any $0<a<1,$ one can find a large enough $c_a$ that depends only on $a,$ so that,
	\benr
	\sup_{(\tau_w-\tau_w^0)\ge c_aT_h^{-1}\phi^2\xi_w^{-2}}\frac{1}{(\tau_w-\tau_w^0)}\Big|\sum_{w=\tau^0_w+1}^{\tau_w}\Big(\sum_{h=\tau^0_h+1}^{T_h}\vep_{(w,h)}^T\eta^0_{(1)}+\sum_{h=1}^{\tau^0_h}\vep_{w,h}^T\eta^0_{(3)}\Big)\Big|\le
	\frac{T_h\xi_w^2}{2},\nn
	\eenr
	with probability at least $1-a.$	
\end{lemma}

\begin{proof} Begin by defining for any $w$ the r.v.,
	\benr
	\psi_{w}=\sum_{h=\tau^0_h+1}^{T_h}\vep_{(w,h)}^T\eta^0_{(1)}+\sum_{h=1}^{\tau^0_h}\vep_{(w,h)}^T\eta^0_{(3)}\nn
	\eenr
	Then from Lemma \ref{lem:lcsubE} we have $\psi_{w}\sim{\rm subE}(\la^2),$ where $\la^2=\phi^2\big\{(T_h-\tau^0_h)\xi_1^2+\tau^0_h\xi_3^2\}=\phi^2T_h\xi_w^2.$ Consequently from Lemma \ref{lem:momentprop} we also have ${\rm var}\psi_{w}\le 16\la^2.$

	Next, for any $0<a<1$ let $c_a$ be a suitably chosen constant that depends only on $a.$ In particular, $a$ and $c_a$ be inversely related, i.e. $c_a$ is larger for a smaller value of $a.$  Now consider the collection of $\tau_w$ satisfying, $
	(\tau_w-\tau_w^0)\ge c_aT_h^{-1}\phi^2\xi_w^{-2}.$ Then applying the H\'ajek--R\'enyi inequality (Theorem \ref{thm:hajek.renyi}) with uniformity over this collection, we obtain for any $\al>0,$ 
	\benr\label{eq:hr1}
	\sup_{(\tau_w-\tau_w^0)\ge c_aT_h^{-1}\phi^2\xi_w^{-2}}\frac{1}{(\tau_w-\tau_w^0)}\Big|\sum_{w=\tau_w^0+1}^{\tau_w}\psi_{w}\Big|>\al
	\eenr
	with probability at most $\big[c_u\phi^2T_h\xi_w^2\big]\Big/\big[c_aT_h^{-1}\phi^2\xi_w^{-2} \al^2\big]=\big[c_uT_h^2\xi_w^4\big]\big/\big[c_a\al^2\big].$ Choosing $\al=T_h\xi_w^2/2$ yields the statement of this Lemma. 
\end{proof}

\begin{lemma}\label{lem:1.to.tau.bound} Assume Condition A and B(i) holds. 	Additionally assume for $c_u> 0$ that $c_uT_wT_h\underline\om\ge\log (p\vee T_wT_h).$ Then for any $c_{u1}>0,$ we have,
	\benr\label{eq:18}
	\max_{1\le j\le 4}\sup_{\substack{\tau_w\in\{1,.....,T\};\\ \substack{\tau_h\in\{1,.....,T\};\\
				|Q_j(\tau)|\ge c_uT_wT_h\underline\om}}}\Big\|\frac{1}{|Q_j(\tau)|}\underset{(w,h)\in Q_j(\tau)}{\sum\sum}\vep_{(w,h)}\Big\|_{\iny}\le  \si\Big\{\frac{2c_{u1}\log(p\vee T_wT_h)}{c_uT_wT_h\underline\om}\Big\}^{\frac{1}{2}}\nn
	\eenr
	with probability at least $1-8\exp\big\{-(c_{u2}-2)\log (p\vee T_wT_h)\big\},$ where $c_{u2}=c_{u1}\wedge \surd(c_uc_{u1}/2).$
\end{lemma}

\begin{proof}[{\bf Proof of Lemma \ref{lem:1.to.tau.bound}}]
	First consider the case of $j=1,$ where $Q_1(\tau)=\big\{(w,h)\in \{\tau_w+1,...T_w\}\times\{\tau_h+1,...,T_h\}\big\},$ then we have $\sum_{w=\tau_w+1}^{T_w}\sum_{h=\tau_h+1}^{T_h}\vep_{(w,h,k)}\sim{\rm subE}\big(|Q_1(\tau)|\si^2\big).$  Now, applying Bernstein's inequality (Lemma \ref{lem:bernstein}) for any $d>0,$ we have,
	\benr\label{eq:19}
	pr\Big(\Big|\sum_{w=\tau_w+1}^{T_w}\sum_{h=\tau_h+1}^{T_h}\vep_{(w,h,k)}\Big|>d|Q_1(\tau)|\Big)\le 2\exp\Big\{-\frac{|Q_1(\tau)|}{2}\Big(\frac{d^2}{\si^2}\wedge\frac{d}{\si}\Big)\Big\}.
	\eenr	
	Choose $d=\si\{2c_{u1}\log (p\vee T_wT_h)\big/|Q_1(\tau)|\}^{1/2},$ and due to the assumption $|Q_1(\tau)|\ge c_uT_wT_h\underline\om\ge \log (p\vee T_wT_h),$ we have,
	\benr
	|Q_1(\tau)|\frac{d^2}{2\si^2}&=&c_{u1}\log(p\vee T_wT_h),\quad {\rm and},\nn\\
	|Q_1(\tau)|\frac{d}{2\si}&\ge&\surd(c_{u1}/2)(c_uT_wT_h\underline\om)^{1/2}\{\log(p\vee T_wT_h)\}^{1/2}\ge \surd(c_uc_{u1}/2)\log(p\vee T_wT_h).\nn
	\eenr	
	Now substituting this choice of $d$ in (\ref{eq:19}),  we obtain,	
	\benr
	\frac{1}{|Q_1(\tau)|}\Big|\sum_{w=\tau_w+1}^{T_w}\sum_{h=\tau_h+1}^{T_h}\vep_{(w,h,k)}\Big|\le \si\{2c_{u1}\log (p\vee T_wT_h)\big/|Q_1(\tau)|\}^{1/2}\le\si\Big\{\frac{2c_{u1}\log(p\vee T_wT_h)}{c_uT_wT_h\underline\om}\Big\}^{1/2}\nn
	\eenr
	with probability at least $1-2\exp\{-c_{u2}\log (p\vee T_wT_h)\},$ where $c_{u2}=c_{u1}\wedge\surd(c_uc_{u1}/2).$ Uniformity of the inner collection in the lemma follows by applying union bounds over all values of $\tau_w,\tau_h$ and $k.$ Uniformity over $j=1,2,3,4$ can be obtained by proceeding with identical arguments as above for each respective quadrant to obtain the same upper bound and finally applying a union bound to obtain the statement of the lemma.
\end{proof}

\section{Definitions and auxiliary results}\label{app:auxiliary}

The following definition's and results provide basic properties of subexponential distributions. These are largely reproduced from \cite{vershynin2019high} and \cite{rigollet201518}. We also refer to Appendix B and Appendix F of \cite{Kaul2021single} and \cite{kaul2023inference}, respectively, where these results and some additional proofs have been compiled.

\begin{definition}\label{def:sube}[Subexponential r.v.] $X\in\R$ is said to be sub-exponential with parameter $\si^2>0$ \big(denoted by $X\sim{\rm subE(\si^2)}$\big) if $E(X)=0$ and its moment generating function
	\benr
	E(\e^{tX})\le \e^{t^2\si^2/2},\qquad \forall\,\, |t|\le \frac{1}{\si}\nn
	\eenr
\end{definition}

\begin{definition}\label{def:submult} A random vector $X\in\R^p$ shall said to be subexponential with parameter $\si^2,$ if the inner product $\langle X, v\rangle\sim {\rm subE}(\si^2),$ respectively, for any $v\in\R^p$ with $\|v\|_2 = 1.$
\end{definition}

Following is the elementary definition of uniform tightness. 
\begin{definition}\label{def:utight} A sequence of random variables $X_n$ is said to be uniformly tight if for every $\ep>0,$ there is a compact set $K$ such that $pr(X_n\in K)>1-\ep.$
\end{definition}

\begin{lemma}\label{lem:tailb}[Tail bounds] If $X\sim {\rm subE}(\si^2),$ then
	\benr
	pr(|X|\ge \la)\le 2\exp\Big\{-\frac{1}{2}\Big(\frac{\la^2}{\si^2}\wedge\frac{\la}{\si}\Big) \Big\}.\nn
	\eenr
\end{lemma}

\begin{lemma}[Moment bounds]\label{lem:momentprop} If $X\sim {\rm subE}(\si^2),$ then
	\benr
	E|X|^k\le 4\si^k k^k, \qquad k> 0.\nn
	\eenr	
\end{lemma}

\begin{lemma}\label{lem:lcsubE} Assume that $X\sim{\rm subE(\si^2)},$ and that $\al\in\R,$ then $\alpha X\sim{\rm subE}(\alpha^2\si^2).$ Moreover, assume that $X_1\sim{\rm subE(\si_1^2)}$ and $X_2\sim{\rm subE(\si_2^2)},$ then $X_1+X_2\sim{\rm subE((\si_1+\si_2)^2)},$ additionally, if $X_1$ and $X_2$ are independent, then $X_1+X_2\sim{\rm subE(\si_1^2+\si_2^2)}.$
\end{lemma}

\begin{lemma}[Bernstein's inequality]\label{lem:bernstein} Let $X_1,X_2,...,X_T$ be independent random variables such that $X_t\sim {\rm subE}(\la^2).$ Then for any $d>0$ we have,
	\benr
	pr(|\bar X|>d)\le 2\exp\Big\{-\frac{T}{2}\Big(\frac{d^2}{\la^2}\wedge \frac{d}{\la}\Big)\Big\}\nn
	\eenr
\end{lemma}

%
%

Following is H\'ajek--R\'enyi inequality reproduced from \cite{bai1994}.

\begin{theorem}\label{thm:hajek.renyi}(H\'ajek--R\'enyi inequality) Let $\vep_1,\vep_2,...,\vep_T$ be a sequence of martingale differences with $E\vep_t^2=\si^2,$ and $\{c_k\}$ be a decreasing positive sequence of constants. 
	\benr
	pr\Big(\max_{m\le k\le T}c_k\Big|\sum_{t=1}^{k}\vep_t\Big|>\al\Big)\le \frac{\si^2}{\al^2} \Big(mc_m^2+\sum_{t=m+1}^{T}c_t^2\Big)\nn
	\eenr
\end{theorem}

As a consequence of this result, choosing $c_k=1/k,$  we also have the uniform bound, \benr
pr\Big(\sup_{k\ge m}\frac{1}{k}\Big|\sum_{t=1}^{k}\vep_t\Big|\ge \al\Big)\le \frac{c_u\si^2}{m\al^2}.
\eenr
where the inequality follows since $\sum_{k=m}^{\iny}k^{-2}=O(m^{-1}).$


Following is the `Argmax' theorem (Theorem 3.2.2 of \cite{vaart1996weak}).
\begin{theorem}[Argmax Theorem]\label{thm:argmax} Let $\cM_n,\cM$ be stochastic processes indexed by a metric space $H$ such that $\cM_n\Rightarrow\cM$ in $\ell^{\iny}(K)$ for every compact set $K\subseteq H$. Suppose that almost all sample paths $h\to \cM(h)$ are upper semicontinuous and posses a unique maximum at a (random) point $\h h,$ which as a random map in $H$ is tight. If the sequence $\h h_n$ is uniformly tight and satisfies $\cM_n(\h h_n)\ge \sup_h \cM_n(h)-o_p(1),$ then $\h h_n\Rightarrow \h h$ in $H.$
\end{theorem}

\section{Additional details and results}\label{app:numerical.supplement}

Here we provide remaining results of Section \ref{sec:numerical}. Subsection \ref{subsec:app:par.est} provides details of estimation of additional parameters such as drifts and asymptotic variances, and on computation of quantiles which are in turn necessary for computation of confidence intervals presented in Section \ref{sec:numerical}.

\subsection{Estimation of drifts, asymptotic variances and quantiles}\label{subsec:app:par.est}

We begin with a discussion on the estimation of $\xi_{w},$ $\xi_{h},$ and $\si^2_{(w,\iny)},\si^2_{(h,\iny)}$ which  utilized for computation of confidence intervals for $\tau^0=(\tau^0_w,\tau^0_h)^T$ using the result of Theorem \ref{thm:wc.vanishing} and Theorem \ref{thm:wc.non.vanishing}. To avoid redundancy, we only describe this estimation process in context of the width change parameter $\tau^0_w,$ the procedure for the height parameter is symmetrical.

First, in order to alleviate finite sample regularization biases we utilize refitted mean estimates computed as $\breve\theta_{(j)}=\big[\bar x_{(j)}(\tilde\tau)\big]_{\h S_j}$ $j=1,2,3,4,$ where $\tilde\tau=(\tilde\tau_w,\tilde\tau_h)^T$ is the change point estimate of Algorithm 1. Here $\h S_j=\{k\,\,\h\theta_{(1k)}\ne 0\},$ $j=1,2,3,4,$ are the estimated sparsity sets, where $\h\theta_{(j)},$ $j=1,2,3,4$ are the Step 2 mean estimates of Algorithm 1. It is well known in the literature that refitted mean estimates preserve the rate of convergence of the regularized version while reducing finite sample biases, see, e.g. \cite{belloni2017pivotal}. The jump vector and individual jump size are then evaluated as $\breve\eta_{(j)}$ and $\breve\xi_j,$ $j=1,2,3,4,$ as plug-in estimates as per the relations (\ref{def:jump.size.quad}). The width and height-wise proportions are estimated as the observed versions, i.e., $\breve\om_w=(T_w-\tilde\tau^0_w)/T_w$ and $\breve\om_h=(T_h-\tilde\tau_h)/T_h.$ The directional jump sizes are obtained as $\breve\xi_w=\breve\om_h\breve\xi_1^2+(1-\breve\om_h)\breve\xi_4^2$ and symmetrically for $\breve\xi_h.$

Next consider asymptotic variance $\si^2_{(w,\iny)}$ of Condition D. Note the finite sample representation of this parameter, $\xi_{w}^{-2}\big[\om_h\eta^0_{(1)}\Si\eta^0_{(1)}+(1-\om_h)\eta^0_{(3)}\Si\eta^0_{(3)}\big].$  A plug in version $\breve\si^2_{w,\iny}$ is computed by utilizing the above described estimated parameters. The covariance matrix $\Si$ is estimated as the sample covariance $\breve\Si$ computed on the entire data set with centering done in correspondence with the estimated mean parameters over quadrants. We note that since we are not interested in the estimation of the covariance itself but instead the quadratic form described above, thus utilizing the sample covariance here is effectively identical to utilizing refitted covariance on the adjacency matrix estimated by the jump vectors $\breve\eta_{(1)},$ and $\breve\eta_{(3)},$ in turn making this shortcut valid despite potential high dimensionality.

Finally, regarding quantiles of distributions characterized in Theorem \ref{thm:wc.vanishing} and Theorem \ref{thm:wc.non.vanishing} in the vanishing and non-vanishing regimes, respectively. For the quantiles of the former case, we utilize the cdf presented in \cite{yao1987approximating}. For the latter case, we assume in all calculations that underlying distribution is Gaussian and consequently the distribution of the increments $\cP$ of Condition A$'$ is also Gaussian. The above estimated parameters are then utilized to produce realizations of this incremental distribution, which are then used to produce realizations of the two-sided random walk and in turn those of its {\it argmax}. The quantiles are then estimated as a monte-carlo approximation.

\subsection{Additional simulation results of Section \ref{sec:numerical}}\label{subsec:app:simulations}

Below are the additional results of the simulations described in Section \ref{sec:numerical}. Table \ref{tab:wres1}, Table \ref{tab:hres3} and Table \ref{tab:hres4} provide results on estimation of $\tau_h^0.$ Table \ref{tab:hres1} -Table \ref{tab:hres4} provide results on estimation of the height change parameter $\tau_h^0.$ All results are in accordance to the discussion of Section \ref{sec:numerical}.

\begin{table}[]
	\centering
	\resizebox{0.62\textwidth}{!}{
		\begin{tabular}{llcccccc}
			\toprule
			\multicolumn{2}{c}{\multirow{2}{*}{\begin{tabular}[c]{@{}c@{}}$T_h=30,$\\ $\tau^0_h/T_h=0.2$\end{tabular}}} & \multicolumn{3}{c}{$p=10$}                                                                                & \multicolumn{3}{c}{$p=50$}                                                                               \\ \cmidrule{3-8}
			\multicolumn{2}{c}{}                                                                                        & \multirow{2}{*}{bias (rmse)}      & \multicolumn{2}{c}{coverage (av. ME)}                                 & \multirow{2}{*}{bias (rmse)}      & \multicolumn{2}{c}{coverage (av. ME)}                                \\ \cmidrule{1-2} \cmidrule{4-5} \cmidrule{7-8}
			\multicolumn{1}{c}{$\tau^0_w/T_w$}                         & \multicolumn{1}{c}{$T_w$}                        &                                   & Vanishing                         & Non-Vanishing                     &                                   & Vanishing                         & Non-Vanishing                    \\ \midrule
			0.2                                                      & 30                                               & 0.008 (0.155)                     & 0.976 (0.468)                     & 0.978 (0.018)                     & 0.01 (0.148)                      & 0.978 (0.394)                     & 0.978 (0)                        \\
			0.2                                                      & 35                                               & 0.02 (0.2)                        & 0.96 (0.508)                      & 0.966 (0.04)                      & 0.022 (0.195)                     & 0.968 (0.401)                     & 0.968 (0)                        \\
			0.2                                                      & 40                                               & 0.016 (0.21)                      & 0.962 (0.514)                     & 0.964 (0.032)                     & 0.02 (0.2)                        & 0.96 (0.409)                      & 0.96 (0)                         \\
			0.2                                                      & 45                                               & 0.032 (0.253)                     & 0.958 (0.516)                     & 0.96 (0.02)                       & 0.006 (0.214)                     & 0.972 (0.418)                     & 0.972 (0)                        \\ \midrule
			0.4                                                      & 30                                               & 0.02 (0.228)                      & 0.954 (0.476)                     & 0.956 (0.02)                      & 0.034 (0.344)                     & 0.964 (0.471)                     & 0.964 (0.004)                    \\
			0.4                                                      & 35                                               & 0.028 (0.268)                     & 0.958 (0.522)                     & 0.958 (0.02)                      & 0.064 (0.477)                     & 0.954 (0.472)                     & 0.954 (0.002)                    \\
			0.4                                                      & 40                                               & 0.006 (0.224)                     & 0.966 (0.528)                     & 0.968 (0.018)                     & 0.05 (0.293)                      & 0.962 (0.48)                      & 0.962 (0.002)                    \\
			0.4                                                      & 45                                               & 0.02 (0.228)                      & 0.96 (0.523)                      & 0.964 (0.012)                     & 0.008 (0.141)                     & 0.98 (0.482)                      & 0.982 (0.004)                    \\ \midrule
			0.6                                                      & 30                                               & 0.06 (0.358)                      & 0.948 (0.495)                     & 0.95 (0.028)                      & 0.078 (0.361)                     & 0.936 (0.504)                     & 0.938 (0.014)                    \\
			0.6                                                      & 35                                               & 0.008 (0.245)                     & 0.97 (0.522)                      & 0.972 (0.02)                      & 0.056 (0.303)                     & 0.954 (0.501)                     & 0.954 (0.008)                    \\
			0.6                                                      & 40                                               & 0.006 (0.265)                     & 0.958 (0.533)                     & 0.96 (0.026)                      & 0.054 (0.326)                     & 0.946 (0.507)                     & 0.946 (0.004)                    \\
			0.6                                                      & 45                                               & 0.034 (0.224)                     & 0.962 (0.535)                     & 0.966 (0.02)                      & 0.048 (0.228)                     & 0.948 (0.515)                     & 0.952 (0.01)                     \\ \midrule
			0.8                                                      & 30                                               & 0.118 (0.882)                     & 0.944 (0.529)                     & 0.948 (0.064)                     & 0.212 (0.976)                     & 0.876 (0.474)                     & 0.886 (0.036)                    \\
			0.8                                                      & 35                                               & 0.036 (0.237)                     & 0.968 (0.521)                     & 0.974 (0.054)                     & 0.13 (0.65)                       & 0.916 (0.475)                     & 0.916 (0.014)                    \\
			0.8                                                      & 40                                               & 0.024 (0.179)                     & 0.97 (0.513)                      & 0.972 (0.022)                     & 0.092 (0.42)                      & 0.932 (0.488)                     & 0.932 (0.012)                    \\
			0.8                                                      & 45                                               & \multicolumn{1}{l}{0.046 (0.265)} & \multicolumn{1}{l}{0.958 (0.519)} & \multicolumn{1}{l}{0.966 (0.034)} & \multicolumn{1}{l}{0.066 (0.319)} & \multicolumn{1}{l}{0.938 (0.488)} & \multicolumn{1}{l}{0.94 (0.006)} \\ \bottomrule
	\end{tabular}}
	\vspace{1mm}
	\caption{\footnotesize{Simulation results for estimation of $\tau^0_w$ based on 500 replications. All reported metrics rounded to three decimals. Other data generating parameters: $T_h=30,$ $\tau^0_h=\lfloor 0.2\cdotp T_h\rfloor$ and $p\in\{10,50\}.$}}
	\label{tab:wres1}
\end{table}

\begin{table}[]		\centering
	\resizebox{0.62\textwidth}{!}{
		\begin{tabular}{llllllll}
			\toprule
			\multicolumn{2}{c}{\multirow{2}{*}{\begin{tabular}[c]{@{}c@{}}$T_h=30,$\\ $\tau^0_h/T_h=0.4$\end{tabular}}} & \multicolumn{3}{c}{$p=10$}                                                                                           & \multicolumn{3}{c}{$p=50$}                                                                                           \\ \cmidrule{3-8}
			\multicolumn{2}{c}{}                                                                                        & \multicolumn{1}{c}{\multirow{2}{*}{bias (rmse)}} & \multicolumn{2}{c}{coverage (av. ME)}                             & \multicolumn{1}{c}{\multirow{2}{*}{bias (rmse)}} & \multicolumn{2}{c}{coverage (av. ME)}                             \\ \cmidrule{1-2} \cmidrule{4-5} \cmidrule{7-8}
			\multicolumn{1}{c}{$\tau^0_w/T_w$}                        & \multicolumn{1}{c}{$T_w$}                       & \multicolumn{1}{c}{}                             & \multicolumn{1}{c}{Vanishing} & \multicolumn{1}{c}{Non-Vanishing} & \multicolumn{1}{c}{}                             & \multicolumn{1}{c}{Vanishing} & \multicolumn{1}{c}{Non-Vanishing} \\ \midrule
			0.2                                                       & 30                                              & 0.028 (0.219)                                    & 0.958 (0.464)                 & 0.962 (0.028)                     & 0.106 (1.093)                                    & 0.954 (0.421)                 & 0.954 (0.02)                      \\
			0.2                                                       & 35                                              & 0.056 (0.303)                                    & 0.936 (0.517)                 & 0.948 (0.048)                     & 0.054 (0.88)                                     & 0.964 (0.443)                 & 0.964 (0.012)                     \\
			0.2                                                       & 40                                              & 0.02 (0.19)                                      & 0.97 (0.512)                  & 0.972 (0.03)                      & 0.008 (0.155)                                    & 0.976 (0.441)                 & 0.976 (0.004)                     \\
			0.2                                                       & 45                                              & 0.028 (0.237)                                    & 0.956 (0.519)                 & 0.958 (0.036)                     & 0.004 (0.19)                                     & 0.964 (0.442)                 & 0.964 (0)                         \\ \midrule
			0.4                                                       & 30                                              & 0.002 (0.279)                                    & 0.964 (0.467)                 & 0.964 (0.022)                     & 0.02 (0.219)                                     & 0.968 (0.443)                 & 0.968 (0.006)                     \\
			0.4                                                       & 35                                              & 0.01 (0.272)                                     & 0.96 (0.524)                  & 0.96 (0.03)                       & 0.002 (0.195)                                    & 0.968 (0.445)                 & 0.968 (0.004)                     \\
			0.4                                                       & 40                                              & 0.01 (0.214)                                     & 0.96 (0.534)                  & 0.962 (0.032)                     & 0.002 (0.161)                                    & 0.974 (0.456)                 & 0.974 (0.004)                     \\
			0.4                                                       & 45                                              & 0.014 (0.272)                                    & 0.964 (0.531)                 & 0.964 (0.016)                     & 0.014 (0.184)                                    & 0.966 (0.459)                 & 0.966 (0.002)                     \\ \midrule
			0.6                                                       & 30                                              & 0.026 (0.205)                                    & 0.964 (0.471)                 & 0.964 (0.012)                     & 0.04 (0.522)                                     & 0.968 (0.474)                 & 0.97 (0.012)                      \\
			0.6                                                       & 35                                              & 0.022 (0.326)                                    & 0.948 (0.523)                 & 0.952 (0.024)                     & 0.016 (0.179)                                    & 0.974 (0.466)                 & 0.974 (0.006)                     \\
			0.6                                                       & 40                                              & 0.022 (0.232)                                    & 0.968 (0.521)                 & 0.97 (0.01)                       & 0.018 (0.195)                                    & 0.968 (0.476)                 & 0.968 (0.004)                     \\
			0.6                                                       & 45                                              & 0.008 (0.21)                                     & 0.962 (0.529)                 & 0.966 (0.016)                     & 0.01 (0.184)                                     & 0.966 (0.49)                  & 0.968 (0.006)                     \\ \midrule
			0.8                                                       & 30                                              & 0.028 (0.261)                                    & 0.95 (0.486)                  & 0.952 (0.034)                     & 0.146 (1.15)                                     & 0.926 (0.46)                  & 0.926 (0.016)                     \\
			0.8                                                       & 35                                              & 0.028 (0.253)                                    & 0.954 (0.505)                 & 0.96 (0.034)                      & 0.076 (0.927)                                    & 0.956 (0.465)                 & 0.956 (0.008)                     \\
			0.8                                                       & 40                                              & 0.022 (0.224)                                    & 0.968 (0.519)                 & 0.972 (0.042)                     & 0.028 (0.19)                                     & 0.964 (0.472)                 & 0.966 (0.004)                     \\
			0.8                                                       & 45                                              & 0.048 (0.253)                                    & 0.948 (0.523)                 & 0.956 (0.036)                     & 0.038 (0.3)                                      & 0.962 (0.475)                 & 0.962 (0.006)                     \\ \bottomrule
	\end{tabular}}
	\vspace{1mm}
	\caption{\footnotesize{Simulation results for estimation of $\tau^0_w$ based on 500 replications. All reported metrics rounded to three decimals. Other data generating parameters: $T_h=30,$ $\tau^0_h=\lfloor 0.4\cdotp T_h\rfloor$ and $p\in\{10,50\}.$}}
	\label{tab:wres3}
\end{table}

\begin{table}[]
	\centering
	\resizebox{0.62\textwidth}{!}{
		\begin{tabular}{llllllll}
			\toprule
			\multicolumn{2}{c}{\multirow{2}{*}{\begin{tabular}[c]{@{}c@{}}$T_h=30,$\\ $\tau^0_h/T_h=0.4$\end{tabular}}} & \multicolumn{3}{c}{$p=100$}                                                                                          & \multicolumn{3}{c}{$p=250$}                                                                                          \\ \cmidrule{3-8}
			\multicolumn{2}{c}{}                                                                                        & \multicolumn{1}{c}{\multirow{2}{*}{bias (rmse)}} & \multicolumn{2}{c}{coverage (av. ME)}                             & \multicolumn{1}{c}{\multirow{2}{*}{bias (rmse)}} & \multicolumn{2}{c}{coverage (av. ME)}                             \\ \cmidrule{1-2} \cmidrule{4-5} \cmidrule{7-8}
			\multicolumn{1}{c}{$\tau^0_w/T_w$}                        & \multicolumn{1}{c}{$T_w$}                       & \multicolumn{1}{c}{}                             & \multicolumn{1}{c}{Vanishing} & \multicolumn{1}{c}{Non-Vanishing} & \multicolumn{1}{c}{}                             & \multicolumn{1}{c}{Vanishing} & \multicolumn{1}{c}{Non-Vanishing} \\ \midrule
			0.2                                                       & 30                                              & 0.09 (1.122)                                     & 0.968 (0.4)                   & 0.968 (0.008)                     & 0.076 (0.908)                                    & 0.966 (0.357)                 & 0.966 (0.008)                     \\
			0.2                                                       & 35                                              & 0.03 (0.488)                                     & 0.966 (0.412)                 & 0.966 (0.004)                     & 0.008 (0.179)                                    & 0.968 (0.367)                 & 0.968 (0)                         \\
			0.2                                                       & 40                                              & 0.01 (0.195)                                     & 0.962 (0.418)                 & 0.962 (0)                         & 0.046 (0.866)                                    & 0.97 (0.387)                  & 0.97 (0.006)                      \\
			0.2                                                       & 45                                              & 0.006 (0.161)                                    & 0.974 (0.425)                 & 0.974 (0)                         & 0.014 (0.173)                                    & 0.976 (0.397)                 & 0.976 (0)                         \\ \midrule
			0.4                                                       & 30                                              & 0.002 (0.224)                                    & 0.956 (0.438)                 & 0.962 (0.008)                     & 0.008 (0.2)                                      & 0.966 (0.421)                 & 0.966 (0.008)                     \\
			0.4                                                       & 35                                              & 0.032 (0.316)                                    & 0.976 (0.434)                 & 0.976 (0.002)                     & 0.03 (0.392)                                     & 0.97 (0.432)                  & 0.97 (0)                          \\
			0.4                                                       & 40                                              & 0.01 (0.205)                                     & 0.97 (0.438)                  & 0.97 (0)                          & 0.002 (0.214)                                    & 0.96 (0.437)                  & 0.96 (0.002)                      \\
			0.4                                                       & 45                                              & 0.01 (0.148)                                     & 0.978 (0.441)                 & 0.978 (0)                         & 0.004 (0.11)                                     & 0.988 (0.445)                 & 0.988 (0)                         \\ \midrule
			0.6                                                       & 30                                              & 0.018 (0.173)                                    & 0.97 (0.459)                  & 0.97 (0.004)                      & 0.014 (0.173)                                    & 0.976 (0.46)                  & 0.976 (0)                         \\
			0.6                                                       & 35                                              & 0.038 (0.605)                                    & 0.97 (0.471)                  & 0.97 (0.004)                      & 0.018 (0.195)                                    & 0.968 (0.479)                 & 0.968 (0.004)                     \\
			0.6                                                       & 40                                              & 0.006 (0.118)                                    & 0.986 (0.482)                 & 0.986 (0.002)                     & 0.006 (0.148)                                    & 0.978 (0.475)                 & 0.978 (0)                         \\
			0.6                                                       & 45                                              & 0.022 (0.232)                                    & 0.958 (0.484)                 & 0.96 (0.004)                      & 0.002 (0.184)                                    & 0.978 (0.487)                 & 0.978 (0)                         \\ \midrule
			0.8                                                       & 30                                              & 0.14 (1.103)                                     & 0.94 (0.434)                  & 0.942 (0.01)                      & 0.138 (1.127)                                    & 0.94 (0.403)                  & 0.94 (0.016)                      \\
			0.8                                                       & 35                                              & 0.052 (0.486)                                    & 0.954 (0.437)                 & 0.954 (0.002)                     & 0.06 (0.54)                                      & 0.954 (0.417)                 & 0.954 (0.004)                     \\
			0.8                                                       & 40                                              & 0.034 (0.214)                                    & 0.96 (0.45)                   & 0.96 (0)                          & 0.04 (0.253)                                     & 0.954 (0.426)                 & 0.954 (0.004)                     \\
			0.8                                                       & 45                                              & 0.052 (1.083)                                    & 0.978 (0.463)                 & 0.978 (0.002)                     & 0.02 (0.2)                                       & 0.966 (0.433)                 & 0.966 (0)                         \\ \bottomrule
	\end{tabular}}
	\vspace{1mm}
	\caption{\footnotesize{Simulation results for estimation of $\tau^0_w$ based on 500 replications. All reported metrics rounded to three decimals. Other data generating parameters: $T_h=30,$ $\tau^0_h=\lfloor 0.4\cdotp T_h\rfloor$ and $p\in\{100,250\}.$}}
	\label{tab:wres4}
\end{table}

\begin{table}[]	
	\centering
	\resizebox{0.62\textwidth}{!}{
		\begin{tabular}{llllllll}
			\toprule
			\multicolumn{2}{c}{\multirow{2}{*}{\begin{tabular}[c]{@{}c@{}}$T_w=30,$\\ $\tau^0_w/T_w=0.2$\end{tabular}}} & \multicolumn{3}{c}{$p=10$}                                                                                           & \multicolumn{3}{c}{$p=50$}                                                                                           \\ \cmidrule{3-8}
			\multicolumn{2}{c}{}                                                                                        & \multicolumn{1}{c}{\multirow{2}{*}{bias (rmse)}} & \multicolumn{2}{c}{coverage (av. ME)}                             & \multicolumn{1}{c}{\multirow{2}{*}{bias (rmse)}} & \multicolumn{2}{c}{coverage (av. ME)}                             \\ \cmidrule{1-2} \cmidrule{4-5} \cmidrule{7-8}
			\multicolumn{1}{c}{$\tau^0_h/T_h$}                        & \multicolumn{1}{c}{$T_h$}                       & \multicolumn{1}{c}{}                             & \multicolumn{1}{c}{Vanishing} & \multicolumn{1}{c}{Non-Vanishing} & \multicolumn{1}{c}{}                             & \multicolumn{1}{c}{Vanishing} & \multicolumn{1}{c}{Non-Vanishing} \\ \midrule
			0.2 & 30 & 0.004 (0.167) & 0.978 (0.469) & 0.984 (0.026) & 0.002 (0.205) & 0.964 (0.395) & 0.964 (0.002) \\
			0.2 & 35 & 0.028 (0.228) & 0.956 (0.513) & 0.964 (0.036) & 0.010 (0.184) & 0.966 (0.409) & 0.966 (0.004) \\
			0.2 & 40 & 0.022 (0.195) & 0.974 (0.518) & 0.974 (0.026) & 0.002 (0.205) & 0.964 (0.414) & 0.964 (0.000) \\
			0.2 & 45 & 0.022 (0.257) & 0.956 (0.516) & 0.958 (0.024) & 0.010 (0.161) & 0.974 (0.420) & 0.974 (0.002) \\
			\midrule
			0.4 & 30 & 0.022 (0.205) & 0.970 (0.469) & 0.972 (0.018) & 0.078 (0.516) & 0.938 (0.458) & 0.940 (0.010) \\
			0.4 & 35 & 0.028 (0.210) & 0.968 (0.523) & 0.968 (0.014) & 0.058 (0.431) & 0.960 (0.472) & 0.960 (0.006) \\
			0.4 & 40 & 0.028 (0.297) & 0.960 (0.531) & 0.962 (0.028) & 0.026 (0.195) & 0.968 (0.475) & 0.968 (0.000) \\
			0.4 & 45 & 0.024 (0.210) & 0.962 (0.526) & 0.964 (0.006) & 0.018 (0.184) & 0.972 (0.494) & 0.972 (0.002) \\
			\midrule
			0.6 & 30 & 0.036 (0.253) & 0.948 (0.505) & 0.952 (0.028) & 0.056 (0.322) & 0.960 (0.494) & 0.964 (0.012) \\
			0.6 & 35 & 0.002 (0.257) & 0.952 (0.522) & 0.956 (0.028) & 0.052 (0.290) & 0.950 (0.497) & 0.950 (0.004) \\
			0.6 & 40 & 0.004 (0.253) & 0.954 (0.526) & 0.954 (0.014) & 0.040 (0.228) & 0.954 (0.505) & 0.954 (0.008) \\
			0.6 & 45 & 0.028 (0.253) & 0.942 (0.532) & 0.950 (0.026) & 0.036 (0.237) & 0.950 (0.516) & 0.956 (0.006) \\
			\midrule
			0.8 & 30 & 0.084 (0.490) & 0.938 (0.509) & 0.952 (0.048) & 0.274 (1.072) & 0.862 (0.483) & 0.866 (0.044) \\
			0.8 & 35 & 0.044 (0.303) & 0.936 (0.512) & 0.946 (0.038) & 0.198 (0.987) & 0.906 (0.487) & 0.908 (0.024) \\
			0.8 & 40 & 0.034 (0.249) & 0.956 (0.520) & 0.958 (0.018) & 0.076 (0.335) & 0.932 (0.482) & 0.936 (0.010) \\
			0.8 & 45 & 0.038 (0.249) & 0.956 (0.525) & 0.960 (0.042) & 0.088 (0.522) & 0.944 (0.493) & 0.944 (0.004) \\ \bottomrule
	\end{tabular}}
	\vspace{1mm}
	\caption{\footnotesize{Simulation results for estimation of $\tau^0_h$ based on 500 replications. All reported metrics rounded to three decimals. Other data generating parameters: $T_w=30,$ $\tau^0_w=\lfloor 0.2\cdotp T_w\rfloor$ and $p\in\{10,50\}.$}}
	\label{tab:hres1}
\end{table}

\begin{table}[]
	\centering
	\resizebox{0.62\textwidth}{!}{
		\begin{tabular}{llllllll}
			\toprule
			\multicolumn{2}{c}{\multirow{2}{*}{\begin{tabular}[c]{@{}c@{}}$T_w=30,$\\ $\tau^0_w/T_w=0.2$\end{tabular}}} & \multicolumn{3}{c}{$p=100$}                                                                                          & \multicolumn{3}{c}{$p=250$}                                                                                          \\ \cmidrule{3-8}
			\multicolumn{2}{c}{}                                                                                        & \multicolumn{1}{c}{\multirow{2}{*}{bias (rmse)}} & \multicolumn{2}{c}{coverage (av. ME)}                             & \multicolumn{1}{c}{\multirow{2}{*}{bias (rmse)}} & \multicolumn{2}{c}{coverage (av. ME)}                             \\ \cmidrule{1-2} \cmidrule{4-5} \cmidrule{7-8}
			\multicolumn{1}{c}{$\tau^0_h/T_h$}                        & \multicolumn{1}{c}{$T_h$}                       & \multicolumn{1}{c}{}                             & \multicolumn{1}{c}{Vanishing} & \multicolumn{1}{c}{Non-Vanishing} & \multicolumn{1}{c}{}                             & \multicolumn{1}{c}{Vanishing} & \multicolumn{1}{c}{Non-Vanishing} \\ \midrule
			0.2 & 30 & 0.016 (0.200) & 0.966 (0.347) & 0.966 (0.000) & 0.042 (0.819) & 0.958 (0.308) & 0.958 (0.002) \\
			0.2 & 35 & 0.012 (0.322) & 0.950 (0.365) & 0.950 (0.002) & 0.024 (0.253) & 0.964 (0.323) & 0.964 (0.000) \\
			0.2 & 40 & 0.024 (0.179) & 0.968 (0.379) & 0.968 (0.000) & 0.034 (1.001) & 0.964 (0.336) & 0.964 (0.002) \\
			0.2 & 45 & 0.002 (0.205) & 0.964 (0.382) & 0.964 (0.000) & 0.006 (0.173) & 0.970 (0.351) & 0.970 (0.000) \\
			\midrule
			0.4 & 30 & 0.054 (0.332) & 0.950 (0.460) & 0.950 (0.002) & 0.066 (0.527) & 0.948 (0.437) & 0.950 (0.004) \\
			0.4 & 35 & 0.032 (0.237) & 0.956 (0.468) & 0.956 (0.000) & 0.042 (0.265) & 0.954 (0.450) & 0.954 (0.000) \\
			0.4 & 40 & 0.058 (0.313) & 0.936 (0.470) & 0.936 (0.000) & 0.068 (0.303) & 0.942 (0.461) & 0.942 (0.000) \\
			0.4 & 45 & 0.032 (0.210) & 0.962 (0.479) & 0.962 (0.000) & 0.022 (0.195) & 0.962 (0.468) & 0.962 (0.000) \\
			\midrule
			0.6 & 30 & 0.076 (0.469) & 0.932 (0.484) & 0.934 (0.010) & 0.096 (0.415) & 0.928 (0.464) & 0.928 (0.006) \\
			0.6 & 35 & 0.118 (0.736) & 0.922 (0.498) & 0.926 (0.014) & 0.166 (0.977) & 0.924 (0.480) & 0.924 (0.012) \\
			0.6 & 40 & 0.058 (0.349) & 0.944 (0.503) & 0.944 (0.008) & 0.080 (0.363) & 0.932 (0.487) & 0.932 (0.002) \\
			0.6 & 45 & 0.058 (0.307) & 0.936 (0.505) & 0.940 (0.004) & 0.044 (0.261) & 0.938 (0.496) & 0.940 (0.004) \\
			\midrule
			0.8 & 30 & 0.232 (0.910) & 0.862 (0.459) & 0.862 (0.032) & 0.492 (1.810) & 0.774 (0.424) & 0.776 (0.038) \\
			0.8 & 35 & 0.288 (1.486) & 0.870 (0.467) & 0.870 (0.020) & 0.296 (1.293) & 0.842 (0.429) & 0.844 (0.018) \\
			0.8 & 40 & 0.120 (0.405) & 0.894 (0.467) & 0.894 (0.006) & 0.218 (0.781) & 0.870 (0.438) & 0.874 (0.018) \\
			0.8 & 45 & 0.136 (0.537) & 0.902 (0.470) & 0.902 (0.010) & 0.144 (0.490) & 0.890 (0.440) & 0.892 (0.006) \\
			\bottomrule
	\end{tabular}}
	\vspace{1mm}
	\caption{\footnotesize{Simulation results for estimation of $\tau^0_h$ based on 500 replications. All reported metrics rounded to three decimals. Other data generating parameters: $T_w=30,$ $\tau^0_w=\lfloor 0.2\cdotp T_w\rfloor$ and $p\in\{100,250\}.$}}
	\label{tab:hres2}
\end{table}

\begin{table}[]
	\centering
	\resizebox{0.62\textwidth}{!}{
		\begin{tabular}{llllllll}
			\toprule
			\multicolumn{2}{c}{\multirow{2}{*}{\begin{tabular}[c]{@{}c@{}}$T_w=30,$\\ $\tau^0_w/T_w=0.4$\end{tabular}}} & \multicolumn{3}{c}{$p=10$}                                                                                           & \multicolumn{3}{c}{$p=50$}                                                                                           \\ \cmidrule{3-8}
			\multicolumn{2}{c}{}                                                                                        & \multicolumn{1}{c}{\multirow{2}{*}{bias (rmse)}} & \multicolumn{2}{c}{coverage (av. ME)}                             & \multicolumn{1}{c}{\multirow{2}{*}{bias (rmse)}} & \multicolumn{2}{c}{coverage (av. ME)}                             \\ \cmidrule{1-2} \cmidrule{4-5} \cmidrule{7-8}
			\multicolumn{1}{c}{$\tau^0_h/T_h$}                        & \multicolumn{1}{c}{$T_h$}                       & \multicolumn{1}{c}{}                             & \multicolumn{1}{c}{Vanishing} & \multicolumn{1}{c}{Non-Vanishing} & \multicolumn{1}{c}{}                             & \multicolumn{1}{c}{Vanishing} & \multicolumn{1}{c}{Non-Vanishing} \\ \midrule
			0.2 & 30 & 0.000 (0.210) & 0.962 (0.470) & 0.964 (0.032) & 0.060 (0.764) & 0.950 (0.421) & 0.952 (0.008) \\
			0.2 & 35 & 0.036 (0.261) & 0.948 (0.514) & 0.956 (0.042) & 0.068 (0.876) & 0.966 (0.436) & 0.966 (0.010) \\
			0.2 & 40 & 0.026 (0.214) & 0.970 (0.525) & 0.978 (0.044) & 0.038 (1.042) & 0.976 (0.447) & 0.976 (0.004) \\
			0.2 & 45 & 0.010 (0.173) & 0.976 (0.520) & 0.978 (0.024) & 0.006 (0.205) & 0.958 (0.448) & 0.958 (0.002) \\
			\midrule
			0.4 & 30 & 0.002 (0.249) & 0.968 (0.466) & 0.968 (0.012) & 0.000 (0.268) & 0.956 (0.443) & 0.956 (0.010) \\
			0.4 & 35 & 0.016 (0.210) & 0.968 (0.525) & 0.970 (0.028) & 0.004 (0.126) & 0.984 (0.448) & 0.984 (0.002) \\
			0.4 & 40 & 0.002 (0.257) & 0.956 (0.530) & 0.956 (0.026) & 0.000 (0.126) & 0.984 (0.451) & 0.986 (0.002) \\
			0.4 & 45 & 0.026 (0.214) & 0.960 (0.529) & 0.964 (0.018) & 0.002 (0.161) & 0.980 (0.456) & 0.980 (0.000) \\
			\midrule
			0.6 & 30 & 0.012 (0.219) & 0.958 (0.475) & 0.958 (0.020) & 0.042 (0.427) & 0.950 (0.469) & 0.950 (0.006) \\
			0.6 & 35 & 0.008 (0.253) & 0.954 (0.519) & 0.954 (0.014) & 0.030 (0.257) & 0.962 (0.479) & 0.964 (0.012) \\
			0.6 & 40 & 0.002 (0.205) & 0.976 (0.529) & 0.976 (0.020) & 0.016 (0.200) & 0.972 (0.483) & 0.972 (0.004) \\
			0.6 & 45 & 0.014 (0.257) & 0.956 (0.532) & 0.958 (0.014) & 0.006 (0.173) & 0.970 (0.477) & 0.972 (0.004) \\
			\midrule
			0.8 & 30 & 0.036 (0.261) & 0.950 (0.485) & 0.956 (0.040) & 0.144 (1.200) & 0.932 (0.459) & 0.936 (0.020) \\
			0.8 & 35 & 0.026 (0.265) & 0.958 (0.511) & 0.960 (0.028) & 0.030 (0.195) & 0.974 (0.457) & 0.974 (0.002) \\
			0.8 & 40 & 0.040 (0.316) & 0.952 (0.519) & 0.952 (0.034) & 0.156 (1.741) & 0.956 (0.474) & 0.956 (0.012) \\
			0.8 & 45 & 0.024 (0.237) & 0.964 (0.520) & 0.966 (0.018) & 0.032 (0.322) & 0.966 (0.473) & 0.966 (0.000) \\
			\bottomrule
	\end{tabular}}
	\vspace{1mm}
	\caption{\footnotesize{Simulation results for estimation of $\tau^0_h$ based on 500 replications. All reported metrics rounded to three decimals. Other data generating parameters: $T_w=30,$ $\tau^0_w=\lfloor 0.4\cdotp T_w\rfloor$ and $p\in\{10,50\}.$}}
	\label{tab:hres3}
\end{table}

\begin{table}[]
	\centering
	\resizebox{0.62\textwidth}{!}{
		\begin{tabular}{llllllll}
			\toprule
			\multicolumn{2}{c}{\multirow{2}{*}{\begin{tabular}[c]{@{}c@{}}$T_w=30,$\\ $\tau^0_w/T_w=0.4$\end{tabular}}} & \multicolumn{3}{c}{$p=100$}                                                                                          & \multicolumn{3}{c}{$p=250$}                                                                                          \\ \cmidrule{3-8}
			\multicolumn{2}{c}{}                                                                                        & \multicolumn{1}{c}{\multirow{2}{*}{bias (rmse)}} & \multicolumn{2}{c}{coverage (av. ME)}                             & \multicolumn{1}{c}{\multirow{2}{*}{bias (rmse)}} & \multicolumn{2}{c}{coverage (av. ME)}                             \\ \cmidrule{1-2} \cmidrule{4-5} \cmidrule{7-8}
			\multicolumn{1}{c}{$\tau^0_h/T_h$}                        & \multicolumn{1}{c}{$T_h$}                       & \multicolumn{1}{c}{}                             & \multicolumn{1}{c}{Vanishing} & \multicolumn{1}{c}{Non-Vanishing} & \multicolumn{1}{c}{}                             & \multicolumn{1}{c}{Vanishing} & \multicolumn{1}{c}{Non-Vanishing} \\ \midrule
			0.2 & 30 & 0.052 (0.623) & 0.960 (0.396) & 0.960 (0.004) & 0.104 (1.251) & 0.966 (0.355) & 0.966 (0.002) \\
			0.2 & 35 & 0.002 (0.161) & 0.974 (0.406) & 0.974 (0.000) & 0.002 (0.205) & 0.964 (0.370) & 0.964 (0.000) \\
			0.2 & 40 & 0.012 (0.190) & 0.964 (0.415) & 0.964 (0.000) & 0.012 (0.200) & 0.972 (0.385) & 0.972 (0.000) \\
			0.2 & 45 & 0.020 (0.237) & 0.972 (0.429) & 0.972 (0.000) & 0.008 (0.190) & 0.970 (0.391) & 0.970 (0.000) \\
			\midrule
			0.4 & 30 & 0.006 (0.241) & 0.960 (0.439) & 0.960 (0.006) & 0.006 (0.232) & 0.976 (0.420) & 0.976 (0.002) \\
			0.4 & 35 & 0.016 (0.179) & 0.974 (0.442) & 0.974 (0.006) & 0.018 (0.232) & 0.958 (0.430) & 0.958 (0.002) \\
			0.4 & 40 & 0.004 (0.379) & 0.974 (0.444) & 0.974 (0.000) & 0.010 (0.148) & 0.978 (0.428) & 0.978 (0.002) \\
			0.4 & 45 & 0.018 (0.205) & 0.964 (0.452) & 0.964 (0.002) & 0.000 (0.141) & 0.980 (0.443) & 0.980 (0.000) \\
			\midrule
			0.6 & 30 & 0.044 (0.410) & 0.970 (0.469) & 0.970 (0.006) & 0.016 (0.245) & 0.958 (0.463) & 0.958 (0.004) \\
			0.6 & 35 & 0.004 (0.210) & 0.974 (0.481) & 0.974 (0.004) & 0.024 (0.210) & 0.978 (0.463) & 0.978 (0.000) \\
			0.6 & 40 & 0.016 (0.190) & 0.970 (0.479) & 0.970 (0.002) & 0.036 (0.237) & 0.960 (0.477) & 0.960 (0.000) \\
			0.6 & 45 & 0.020 (0.253) & 0.964 (0.474) & 0.964 (0.006) & 0.010 (0.173) & 0.970 (0.484) & 0.970 (0.000) \\
			\midrule
			0.8 & 30 & 0.044 (0.253) & 0.942 (0.430) & 0.942 (0.008) & 0.068 (0.438) & 0.940 (0.395) & 0.940 (0.002) \\
			0.8 & 35 & 0.078 (1.005) & 0.956 (0.441) & 0.958 (0.010) & 0.026 (0.184) & 0.972 (0.413) & 0.972 (0.002) \\
			0.8 & 40 & 0.026 (0.214) & 0.966 (0.456) & 0.966 (0.006) & 0.078 (1.015) & 0.958 (0.425) & 0.958 (0.004) \\
			0.8 & 45 & 0.012 (0.155) & 0.976 (0.458) & 0.976 (0.000) & 0.040 (0.607) & 0.968 (0.438) & 0.970 (0.004) \\
			\bottomrule
	\end{tabular}}
	\vspace{1mm}
	\caption{\footnotesize{Simulation results for estimation of $\tau^0_h$ based on 500 replications. All reported metrics rounded to three decimals. Other data generating parameters: $T_w=30,$ $\tau^0_w=\lfloor 0.4\cdotp T_w\rfloor$ and $p\in\{100,250\}.$}}
	\label{tab:hres4}
\end{table}

\end{document}